\documentclass[11pt]{article}
\usepackage[colorlinks]{hyperref}
\hypersetup{linkcolor=cyan,filecolor=cyan,citecolor=cyan,urlcolor=cyan}
\usepackage{xspace}
\usepackage{thm-restate,color,xcolor}
\usepackage{amsmath,amssymb,bbm,amsthm}
\usepackage{fullpage}
\usepackage{boxedminipage}
\usepackage[boxed]{algorithm}
\usepackage{epigraph}
\usepackage[sc]{mathpazo}
\usepackage{amsmath}
\usepackage{mathtools}
\usepackage{framed}
\usepackage[framemethod=tikz]{mdframed}
\usepackage{titlesec}
\usepackage{lipsum}%
\usepackage{cleveref,aliascnt}
\usepackage{tikz}
\usepackage{wrapfig}
\def\dnsparagraph#1{\par\vspace{2pt}\noindent{\bf #1}.}
\usetikzlibrary{shapes.geometric}
\usetikzlibrary{arrows}
\usetikzlibrary{arrows.meta}
\usetikzlibrary{patterns}
\usetikzlibrary{shapes.misc}

\newcommand{\congest}{${\mathsf{CONGEST}}$}

\newcommand{\f}{\frac}
\newcommand{\cd}{\cdot}

\newcommand{\sr}{\sqrt}

\newcommand{\lds}{\ldots}

\newcommand{\s}{\subseteq}

\newcommand{\BE}{\begin{enumerate}}
\newcommand{\EE}{\end{enumerate}}
\newcommand{\im}{\item}
\newcommand{\BI}{\begin{itemize}}
\newcommand{\EI}{\end{itemize}}

\newcommand{\logn}{\log n}
\newcommand{\tuple}{\textup{tuple}}

\newcommand{\R}{\mathbb R}
\newcommand{\Z}{\mathbb Z}

\newcommand{\N}{\mathbb N}

\newcommand{\e}{\epsilon}
\newcommand{\de}{\delta}
\newcommand{\De}{\Delta}

\newcommand{\be}{\beta}

\newcommand{\Om}{\Omega}
\newcommand{\el}{\ell}

\newcommand{\Th}{\Theta}
\newcommand{\m}{\mathcal}

\newcommand{\lf}{\lfloor}
\newcommand{\rf}{\rfloor}

\newcommand{\poly}{\textup{poly}}

\newcommand{\norm}[1]{\left\lVert#1\right\rVert}

\newcommand{\lp}{\left(}
\newcommand{\rp}{\right)}

\newcommand{\lmt}{\left[\begin{matrix}}
\newcommand{\rmt}{\end{matrix}\right]}

\newtheorem{theorem}{Theorem}[section]
\newtheorem{lemma}[theorem]{Lemma}
\newtheorem{definition}[theorem]{Definition}
\newtheorem{corollary}[theorem]{Corollary}
\newtheorem{observation}[theorem]{Observation}
\newtheorem{fact}[theorem]{Fact}
\newtheorem{claim}[theorem]{Claim}
\newcommand{\BT}{\begin{theorem}}
\newcommand{\ET}{\end{theorem}}
\newcommand{\BL}{\begin{lemma}}
\newcommand{\EL}{\end{lemma}}
\newcommand{\BD}{\begin{definition}}
\newcommand{\ED}{\end{definition}}
\newcommand{\BC}{\begin{corollary}}
\newcommand{\EC}{\end{corollary}}
\newcommand{\BO}{\begin{observation}}
\newcommand{\EO}{\end{observation}}
\newcommand{\BCL}{\begin{claim}}
\newcommand{\ECL}{\end{claim}}
\newcommand{\BP}{\begin{proof}}
\newcommand{\EP}{\end{proof}}
\newcommand{\BPS}{\begin{proof}[Proof (Sketch)]}
\newcommand{\EPS}{\end{proof}}

\Crefname{observation}{Observation}{Observations}
\Crefname{claim}{Claim}{Claims}
\newcommand{\tO}{\widetilde{O}}
\newcommand{\diam}{\textup{diam}}

\newcommand{\ol}{\overline}
\newcommand{\sep}{\textup{sep}}

\newcommand{\len}{\textup{len}}

\newcommand{\bd}{\mathbf{d}}
\newcommand{\bx}{\mathbf{x}}
\newcommand{\Cut}{\textsc{Cut}}
\renewcommand{\Join}{\textsc{Join}}
\newcommand{\Get}{\textsc{GetValue}}
\newcommand{\Add}{\textsc{AddSubtree}}

\newcommand{\Value}{\textsc{Value}}
\newcommand{\dilation}{\mbox{\tt d}}
\newcommand{\congestion}{\mbox{\tt c}}
\renewcommand{\int}{\textup{int}}

\newcounter{note}[section]

\newenvironment{boxfig}[2]{\begin{figure}[#1]\fbox{\begin{minipage}{\linewidth}
                                \vspace{0.2em}
                                \makebox[0.025\linewidth]{}
                                \begin{minipage}{0.95\linewidth}
                                        {{
                                                        #2 }}
                                \end{minipage}
                                \vspace{0.2em}
        \end{minipage}}}{\end{figure}}
\newcommand{\pprotocol}[4]{
        \begin{boxfig}{h!}{
                        \begin{center}
                                \textbf{#1}
                        \end{center}
                        #4
                        \vspace{0.2em} } \caption{\label{#3} #2}
        \end{boxfig}
}

\newcommand{\protocol}[4]{
        \pprotocol{#1}{#2}{#3}{#4} }

\begin{document}

\title{Planar Diameter via Metric Compression}
\author{
 Jason Li\\
  \small CMU\\
  \small jmli@cs.cmu.edu 
\and
Merav Parter\\
        \small Weizmann Institute \\
        \small merav.parter@weizmann.ac.il
}
\date{}
\maketitle
\begin{abstract}
We develop a new approach for distributed distance computation in planar graphs that is based on a variant of the metric compression problem recently introduced by Abboud et al. [SODA'18].  
In our variant of the Planar Graph Metric Compression Problem, one is given an $n$-vertex planar graph $G=(V,E)$, a set of $S \subseteq V$ source terminals lying on a single face, and a subset of target terminals $T \subseteq V$. The goal is to compactly encode the $S\times T$ distances. 

One of our key technical contributions is in providing a compression scheme that encodes all $S \times T$  distances using $\widetilde{O}(|S|\cdot\poly(D)+|T|)$ bits\footnote{As standard, $\widetilde{O}$ is used to hide $\poly\log n$ factors.}, for unweighted graphs with diameter $D$. This significantly improves the state of the art of $\widetilde{O}(|S|\cdot 2^{D}+|T| \cdot D)$ bits. We also consider an approximate version of the problem for \emph{weighted} graphs, where the goal is to encode $(1+\epsilon)$ approximation of the $S \times T$ distances, for a given input parameter $\epsilon \in (0,1]$. Here, our compression scheme uses $\widetilde{O}(\poly(|S|/\epsilon)+|T|)$ bits. In addition, we describe how these compression schemes can be computed in near-linear time. 
At the heart of this compact compression scheme lies a VC-dimension type argument on planar graphs, using the well-known Sauer’'s lemma. 

This efficient compression scheme leads to several improvements and simplifications in the setting of diameter computation, most notably in the distributed setting:
\begin{itemize}
\item
There is an $\widetilde{O}(D^5)$-round randomized distributed algorithm for computing the diameter in planar graphs, w.h.p. 
\item
There is an $\widetilde{O}(D^3)+\poly(\log n/\epsilon)\cdot D^2$-round randomized distributed algorithm for computing an $(1+\epsilon)$ approximation of the diameter in weighted graphs with polynomially bounded weights, w.h.p.
\end{itemize}
No sublinear round algorithms were known for these problems before. 
These distributed constructions are based on a new recursive graph decomposition that preserves the (unweighted) diameter of each of the subgraphs up to a logarithmic term. Using this decomposition, we also get an \emph{exact} SSSP tree computation within $\widetilde{O}(D^2)$ rounds. 
\end{abstract}

\newpage 
\tableofcontents
\newpage
\section{Introduction}
Computing the diameter of a graph is one of the most central problems in planar graph algorithms.
In general weighted graphs, the best diameter algorithm is based on solving the All-Pairs Shortest Paths
(APSP) problem. In planar graphs, however, the diameter can be solved considerably faster. In recent years there has been a substantial progress on this problem both for the exact as well as for the approximate setting.  
\dnsparagraph{Exact Diameter}
Frederickson \cite{federickson1987fast} gave the first $O(n^2)$ algorithm for the problem using APSP. A poly-logarithmic improvement was given by Wulff-Nilsen \cite{wulff2008wiener}, providing the first indication that diameter is indeed \emph{easier} than APSP.
The question of whether one can compute the diameter in sub-quadratic time was one of the most important open problems in the area for quite some time. In a breakthrough result, Cabello \cite{cabello2017subquadratic}, building upon the heavy machinery of Voronoi diagrams in planar graphs, presented the first truly sub-quadratic diameter algorithm that runs in time $\widetilde{O}(n^{11/6})$. This works even for weighted and directed planar graphs. 
Soon after, by simplifying and extending the approach of Cabello, Gawrychowski et al. \cite{gawrychowski2018voronoi} improved the bound to $O(n^{5/3})$, which is currently the state of the art. The techniques developed in \cite{cabello2017subquadratic,gawrychowski2018voronoi} led to subsequent improvements in the related setting of compact distance oracles \cite{cohen2017fast,gawrychowski2018better,charalampopoulos2018exact}.

\dnsparagraph{Approximate Diameter}
In lack of truly efficient algorithms for diameter computation over the years, the area turned to consider the approximate setting. 
The most notable work in this context is by Weimann and Yuster \cite{weimann2016approximating} that provided the first $(1+\epsilon)$ approximation in time $\widetilde{O}(2^{O(1/\epsilon)}\cdot n)$, hence \emph{linear} for any \emph{constant} $\epsilon$. Unlike the heavy machinery used by the exact algorithms, their approximate algorithm is based on a simple divide and conquer approach using \emph{shortest path separators}. Ideas along this line were first introduced by \cite{thorup2004compact} in the distance oracle setting. 
We elaborate more on this approach in the technical overview section. Chan and Skrepetos \cite{chan2017faster} combined the exact and approximate worlds by combining the algorithm of Weimann and Yuster \cite{weimann2016approximating} with the abstract Voronoi diagram tool of \cite{cabello2017subquadratic}. They achieve a randomized $(1+\epsilon)$ approximation in time $\widetilde{O}(\poly(1/\epsilon)\cdot n)$. We note that one implication of our results is a considerably simpler deterministic ``divide and conquer" algorithm for this problem that has the same time complexity of $\widetilde{O}(\poly(1/\epsilon)\cdot n)$ but avoids the use of Voronoi diagrams.


\subsection{Distributed Algorithms for Planar Graphs}
Throughout, we use a standard message passing model of distributed computing called $\mathsf{CONGEST}$\cite{Peleg:2000}. The network is abstracted as an $n$-node graph $G=(V, E)$, with one processor on each network node. Initially, these processors do not know the graph. They solve the given graph problems via communicating with their neighbors. Communication happens in synchronous rounds. Per round, nodes can send $O(\log n)$-bit message to each of their neighbors.

\dnsparagraph{The Distributed View Point} There is a subtle gap between the centralized and distributed point of views on planar graphs (and on global graph problems in general). In the centralized world, one usually thinks of the graph diameter $D$ in terms of the worst-case $\Omega(n)$ bound. For this reason, an $\sqrt{n}$-size separator is way more preferable over shortest path separators. In contrast, the prevalent viewpoint in distributed graph algorithms thinks of the graph's diameter as being a small number (independent of $n$). With this view, shortest-path separators are preferable over $\sqrt{n}$-size separators. 
This viewpoint has two justifications. First, as argued by Garay, Kutten, and Peleg in their seminal work \cite{Garay-Kutten-Peleg, Kutten-Peleg}, real world networks usually do have small diameter. In addition, global graph problems admit a trivial $\Omega(D)$ lower bound in the distributed setting. Thus, a separator with $O(D)$ vertices is small w.r.t to the total round complexity.

\dnsparagraph{Distributed Planar Graphs via Low-Congestion Shortcuts} 
The area of distributed planar algorithm was initiated by Ghaffari and Haeupler \cite{ghaffari2016distributed}, who introduced the notion of \emph{low-congestion shortcuts}. Roughly speaking, low-congestion shortcuts augment vertex disjoint subgraphs of potentially large diameter, with edges from the original graph in order to considerably reduce their diameter. Using this machinery, \cite{ghaffari2016distributed} has provided improved algorithms for MST and minimum-cut. Low-congestion shortcuts and their algorithmic applications have been studied extensively since then \cite{haeupler2016LCSwoEmbed,haeupler2016LCSonBoundedParam,HaeuplerLZ18,li2018distributed,HaeuplerHW18}.
Recently, Ghaffari and Parter \cite{GhaffariP17} presented a distributed construction of shortest path separator in nearly optimal time. We will use this algorithm extensively in our constructions.

\dnsparagraph{Lack of Efficient Shortest Path Algorithms}
Low-congestion shortcuts provide the fundamental communication backbone for many global graph problems. 
However, when it comes to distance related problems, the shortcuts by them-self seem to be insufficient. 
One of the key contributions in this paper is to provide a new recursive graph decomposition that preserves some distance related measures in each of the recursive pieces. This decomposition along with the low-congestion shortcuts provide the communication backbones for our algorithms.

An exception for the above, is a recent work by Haeupler and Li \cite{HaeuplerL18} that used low-congestion shortcuts to compute $(\log n)^{O(1/\e)}$-approximate SSSP trees within $O(n^{\epsilon}\cdot D)$ rounds. 
\dnsparagraph{Distributed Shortest Paths in General Graphs}
In contrast to planar graphs, the problem of distributed diameter computation in general graphs is fully understood.
Frischknecht et al. showed a lower bound of $\widetilde{\Omega}(n)$ rounds that holds even for networks with constant diameter. Abboud, Censor-Hillel and Khoury \cite{abboud2016near} showed the same lower bound holds even if (i) the network is sparse (and with small diameter), or (ii) if we relax to an $(3/2-\epsilon)$ approximation in sparse graphs. A matching upper bound is known by Peleg, Roddity and Tal \cite{peleg2012distributed}. 

Unlike diameter, distributed shortest path computation for weighted graphs is a subject of an active research, attracting a lot of recent attention. Becker et al. presented a deterministic $(1+o(1))$-approximate shortest paths in $\widetilde{O}(D+\sqrt{n})$. Elkin \cite{elkin2017distributed} provided the first sublinear-time algorithm for
\emph{exact} single source shortest paths on undirected graphs. Huang et al. \cite{huang2017distributed}
presented an improved algorithm for the exact all pairs shortest paths. Recently, Ghaffari and Li \cite{ghaffari2018improved} improved Elkin's result and presented an $\widetilde{O}(n^{3/4}\cdot D^{1/4})$. This was improved even more recently by Forster and Nanongkai \cite{krinninger2017faster}. The lack of efficient distributed algorithms for these problems in general graphs provides the motivation for studying these problems in planar networks.

\subsection{Our Results}
We study the problem of distributed diameter computation (and related problems) by means of  metric compression point of view. This approach is inspired by the approximate diameter algorithm of Weimann and Yuster \cite{weimann2016approximating}, and the metric compression problem by Abboud at el. \cite{abboud2018near}.
We start by defining the following problem, a special case of Abboud at el., which  will underlie the combinatorial basis for our diameter computation. 
\paragraph{The Metric Compression Problem}
\begin{definition}[The OS Metric Compression Problem]
In the OS (Okamura Seymour) Metric Compression Problem\footnote{The setting where the terminal vertices are on the boundary of a face is called an Okamura Seymour instance.} one is given an unweighted, undirected planar $n$-vertex graph $G$, a subset of sources $S \subseteq V$ of vertices lying on a single face in $G$, and a subset of target terminals $T \subseteq V$. The goal is to compute a bit
string $\mathcal{S}$ that encodes all $S \times T$ distances. That is, there is a decoding
function $f$ that given the encoding $\mathcal{S}$ and any two nodes $s, t \in S \times T$ returns the distance $d_G(s,t)$.
\end{definition}
This problem can observed as a special case of the metric compression problem studied by Abboud et al. \cite{abboud2018near}. In particular, \cite{abboud2018near} considered an \emph{arbitrary} subset $S \subseteq V$ with the objective to compress the $S \times S$ distances (rather than the $S\times T$ distances). 
Our formulation is motivated by diameter computation, where the set $S$ corresponds to the cycle separator of the graph and $T=V$, we then wish to compress the two sides across the cycle separator to speed up the computation of the diameter. We note the our solution is technically not related to \cite{abboud2018near}. In the latter, the main challenge is in handling the case where $S$ is not lying on a single face. In our case the challenge is in handling $S \times V$ distance rather than a small set of $S \times S$ distances. Indeed, our approach is different than that 
of \cite{abboud2018near}, and it is based on VC-dimension type arguments. We are unaware of previous use of such arguments in the context of distance computation in planar graphs.
\begin{mdframed}[hidealllines=false,backgroundcolor=gray!20]
\vspace{-2pt}
\begin{theorem}[Exact Compression]\label{thm:exact}
Given an $n$-vertex unweighted planar graph $G=(V,E)$, a set $S \subseteq V$ of sources lying consecutively on a single face, and subset $T \subseteq V$, there exists an algorithm that computes a compression of all $S \times T$ distances in $G$ using $\widetilde{O}(|S|^3 \cdot D+|T|)$ bits. 
\end{theorem}
\end{mdframed}
For the case of weighted graphs $G$ with aspect ratio\footnote{The ratio between the largest distance and smallest distance among all pairs in $G$} $W$, we also provide an $(1+\epsilon)$-approximate compression scheme.
\begin{mdframed}[hidealllines=false,backgroundcolor=gray!20]
\vspace{-2pt}
\begin{theorem}[Approximate Weighted Compression]\label{thm:approx}
Given an $n$-vertex weighted planar graph $G=(V,E,\omega)$ with aspect ratio $W$, a set $S \subseteq V$ of sources lying (not necessarily consecutively) on a single face and set of terminal $T \subseteq V$, there exists an algorithm that computes a compression of $(1+\e)$-approximate $S \times T$ distances in $G$ using $\widetilde{O}((\poly(|S|/\e)+|T|)\log W)$ bits. 
\end{theorem}
\end{mdframed}
We complement these results by providing an efficient algorithm that computes the compressions in linear time (in the input and output size), this improves upon the na\"ive algorithm that takes $O(|S|\cdot n)$ time. 

\dnsparagraph{Distributed Diameter Computation}
We are making a first step of progress on the distributed complexity of this classical problem, by presenting a $\poly(D)$ round algorithm for $D$-diameter planar graphs. No sublinear round algorithm was known for the problem before.
\begin{mdframed}[hidealllines=false,backgroundcolor=gray!20]
\vspace{-2pt}
\begin{theorem}[Distributed Planar Diameter] \label{thm:distr-exact}Given an $n$-vertex unweighted, undirected planar graph with diameter $D$, there is a randomized distributed algorithm that computes the diameter in $\widetilde{O}(D^5)$ rounds, with high probability.
\end{theorem}
\end{mdframed}
We also consider the problem of computing a $(1+\epsilon)$ approximation of the weighted diameter. 
Our end result is:
\begin{mdframed}[hidealllines=false,backgroundcolor=gray!20]
\vspace{-2pt}
\begin{theorem}[Approximate Weighted Compression]\label{thm:approx-diam}
Given an $n$-vertex weighted, undirected planar graph with unweighted/hop diameter $D$ and aspect ratio $W$, for every $\epsilon \in (0,1]$, there exists a distributed approximate planar diameter algorithm that computes a $(1+\epsilon)$ approximation of the diameter in $\widetilde{O}(D^3)+\poly(\log (n W)/\epsilon)\cdot D^2$ rounds, with high probability. 
\end{theorem}
\end{mdframed}
\dnsparagraph{Distance Labels and (Exact) SSSP}
It is well known that distributed shortest path computations in weighted graphs are considerably more challenging (and provably harder in general graphs). The above mentioned $(1+\epsilon)$ approximation results are based upon additional set of tools and constructions, most notably is a construction of an exact SSSP tree. This problem has attracted a lot of attention recently in general graphs.
\begin{mdframed}[hidealllines=false,backgroundcolor=gray!20]
\vspace{-2pt}
\begin{theorem}[Exact SSSP Tree]\label{thm:sssp} There is a randomized distributed algorithm that given a source vertex $s$ computes an exact SSSP tree for any $n$-node planar undirected weighted network with unweighted diameter $D$ in $\widetilde{O}(D^2)$ rounds, with high probability.
\end{theorem}
\end{mdframed}
Interestingly, this result does not use the low-congestion shortcuts machinery. Instead, it is made possible due to our new recursive decomposition technique which preserves the unweighted diameter of each component throughout the recursion. 


\section{Technical Overview} 
\emph{Separators} are subgraphs whose removal from the graph leaves connected components that are all a constant factor smaller than the initial graph. They provide the key tool in working with planar graphs (in the centralized setting).  Typically, one desires the separator to be small, i.e, of size $\sqrt{n}$. In the distributed point of view, $D$ is typically considered to be smaller than $\sqrt{n}$ and thus in this context an $O(D)$-size separator is considered to be small. 
A celebrated result of Lipton and Tarjan \cite{lipton1979separator} demonstrates the existence of a \emph{separator path} in planar graphs. Their proof shows that:
\begin{mdframed}[hidealllines=false,backgroundcolor=white!25]
For any SSSP tree $T$ in a planar graph $G$, there is a non-tree edge $e$ (possible $e \notin G$) such that the strict interior
and strict exterior of the unique simple cycle $C$ in $T \cup \{e\}$ each contains at most
$2/3 \cdot  n$ vertices. Thus, $C$ forms a separator containing two
shortest paths in $T$. 
\end{mdframed}

\dnsparagraph{The High Level Approach for Diameter Computation}
Our diameter computation is based on a common divide and conquer approach introduced by \cite{thorup2004compact} using cycle separators. 

\begin{wrapfigure}{r}{0.25\textwidth} 
                \vspace{-10pt}
                \includegraphics[width=0.3\textwidth]{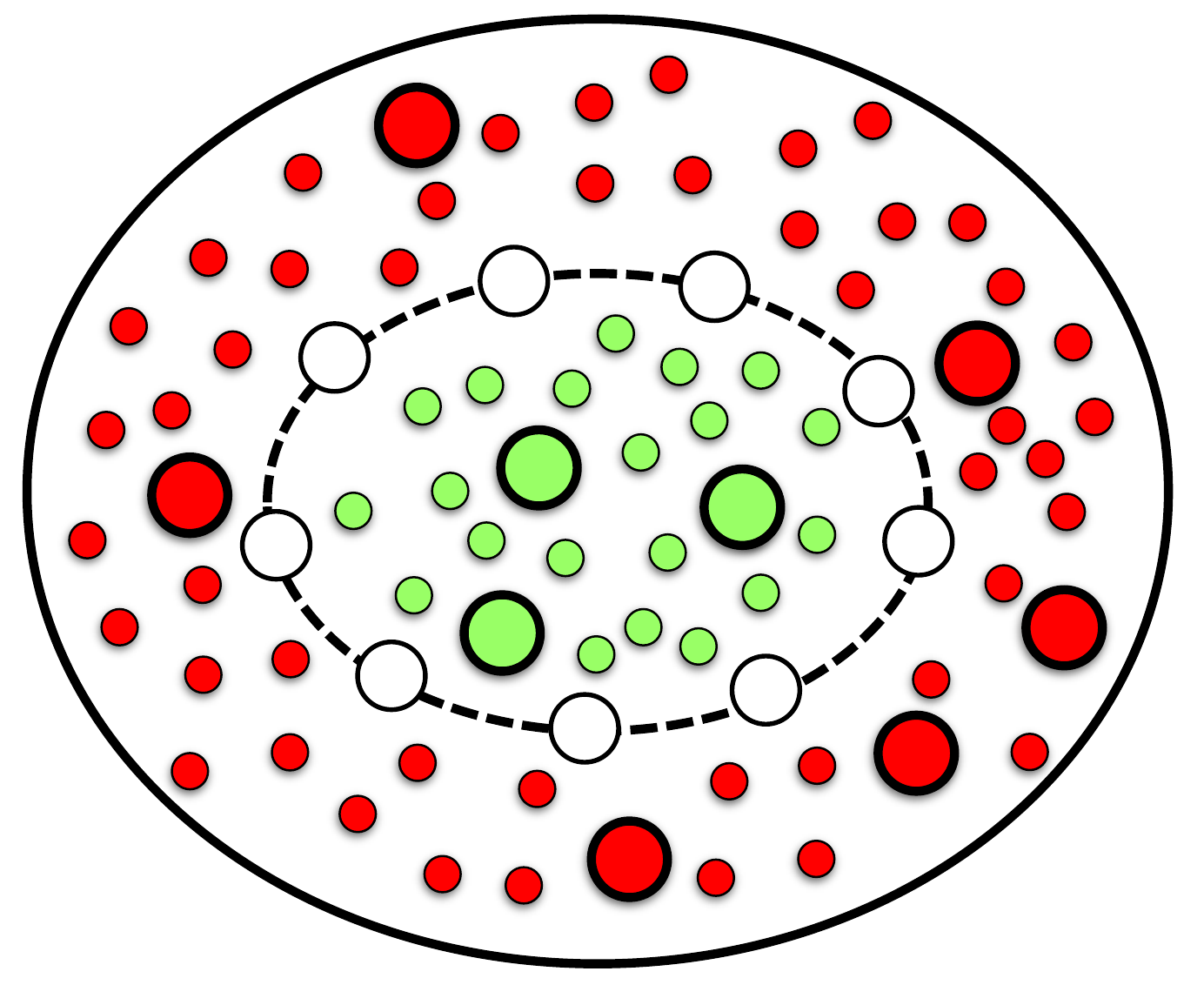}
                \label{fig:  scheme}
                \vspace{-20pt}
\end{wrapfigure} 

In any independent step of the recursion, one is given a subgraph $G' \subseteq G$ and the goal is to compute the
largest distance in $G$ between vertex pairs in $G'$. 
To do that, a cycle separator $C$ is computed in $G'$ which subdivides $G'$ into two subgraphs: the interior $G^+$ and the exterior $G^-$. 
The key task is to compute the largest $G$-distance among all pairs that are separated by the separator $C$. Since $G^+$ and $G^-$ might have $\Omega(n)$ vertices, computing the distances between all pairs of vertices across the separator is inefficient both in the centralized and the distributed setting. 

Our approach, inspired by \cite{weimann2016approximating}, is based on compressing the two sides across the separator, $G^+$ and $G^-$, into a small number of critical vertices $V^{*+} \subset G^+$ and $V^{*-} \subset G^-$, such that the vertex pair of largest distance in $G^+,G^-$ is contained $V^{*+},V^{*-}$. We call these critical sets \emph{core-sets}\footnote{We are aware to the fact that core-sets have similar yet a different context in the literature. We still use this term as it follows the same spirit of other existing core-sets.}.
In the figure, shown is a cycle separator, and the two parts inside and outside the cycle, $G^+$ and $G^-$. The core-sets are the filled large circles inside these regions. 
Having small size core-sets leads to a simple recursive scheme for diameter computation:
\protocol
{Recursive Diameter via Metric Compression:}
{A Simple Recursive Scheme by Compressing the Two Sides Across the Cycle Separator}
{fig:key-problem}
{
\begin{itemize}
\item Compute a cycle separator $C$ in $G$, which decomposes $G$ into $G^+$ and $G^-$ (e.g., see \cite{weimann2016approximating} for the precise definition of $G^+$ and $G^-$).
\item Compute shortest path distances from $C$ to all vertices in $G$.
\item Compute (via metric compression algorithm) the core-sets $V^{*+} \subset G^+$ and $V^{*-} \subset G^-$.
\item Find the farthest pair in $V^{*+}$ and $V^{*-}$.
\item Recurse on $G^+$ and $G^-$.
\end{itemize}
}
%
%
%

\subsection{The Metric Compression Problem}
To provide the high level ideas, we start by considering an unweighted $D$-diameter planar graph $G$. 
Let $S \subseteq V(G)$ be a subset of vertices lying consecutively on a common face. For each vertex $v\in V$, define the \emph{distance tuple} of $v$ to be its distances to $S$ stored as a tuple, defined as follows. 
\BD[Distance tuple]
Let $G=(V,E)$ be a graph, and let $S=\{s_1,\ldots, s_\el\}\s V$ be a subset of vertices. The \emph{$S$-distance tuple} of a vertex $v\in V$ denoted by $\tuple_S(v)$ is the function that maps each vertex $v\in V$ to the vector of distances $\tuple_S(v)=\langle d_G(s_1,v),\ldots, d_G(s_{\ell},v) \rangle$. When the set $S$ is clear from the context, we simply use \emph{distance tuple} and $\tuple(v)$.
\ED
To compress the distances, we will actually show that, perhaps surprisingly, there are only $O(|S|^3D)$ many possible distinct distance tuples. That is, if $n \gg O(|S|^3D)$, then there are many repeated distance tuples among the vertices. This means that we can simply keep a table of the $O(|S|^3D)$ distance tuples, and store for each vertex $t \in T$ the index into the tuple, which has size $O(\log n)$. Therefore, the size of the compression is $O(|S|^3D+|T|\log(|S|D))$. 
In Sec. \ref{sec:exact-unweighted} we show:
\begin{theorem}[\Cref{thm:exact}, Compression] 
Given an $n$-vertex unweighted planar graph $G=(V,E)$ and a set $S \subseteq V$ for sources lying consecutively on a single face, the number of distinct $S$-distance tuples (in $V$) is bounded by $O(|S|^3D)$.
\end{theorem}
We start by representing the distance tuple information as a \emph{set system}. Let $S=\{s_1,\ldots, s_\el\}\s V$ be sorted according to their appearance on the face. 
For each $s_i \in S$, 
we define two sets $A_i^{\Delta}$ for $\Delta \in \{+,-\}$ containing vertices in $V$, where: 
$$A_i^{-}=\{v ~\mid~ d_G(v,s_i)<d_G(v,s_{i+1})\} \mbox{~~and~~} A_i^{+}=\{v ~\mid~ d_G(v,s_i)> d_G(v,s_{i+1})\}.$$
For each vertex $v$, we define a set $F(v)=\{(i,\Delta) ~\mid~ v \in A_i^{\Delta}\}$. The set $F(v)$ can be considered as a weaker version of the distance tuple of $v$. Specifically, two vertices $v$ and $u$ with $F(v)=F(u)$, do not necessarily have the same distance tuples. 
We then define equivalence classes based on the $F(v)$ sets, where $u$ and $v$ are in the 
same equivalence class if $F(u)=F(v)$. Our goal is to show that there are $|S|^3$ equivalence classes.
Assuming this, we are mostly done: once one knows $d_G(s_1,v)$ and $F(v)$, the entire distance tuple of $v$ is determined. As there are $D$ options for the starting value $d_G(s_1,v)$, the total number of tuples will be bounded by $|S|^3\cdot D$.

We will bound the number of equivalence classes using the VC dimension theory. In particular, we will use the well-known Sauer--Shelah Lemma \cite{sauer1972density}:
\begin{lemma}[Sauer--Shelah Lemma \cite{sauer1972density}]
Let $\mathcal{F}=\{A_1,\ldots, A_k\}$ be a family of sets over universe of size $m$, and let $T$ be another set. We say that $\mathcal{F}$ \emph{shatters} $T$ if for every subset $T' \subseteq T$, there exists $A_i \in \mathcal{F}$ such that $A_i \cap T=T'$. The VC-dimension of $\mathcal{F}$ is the largest set $T$ that can be shattered by  $\mathcal{F}$. Then, if the VC-dimension of $\mathcal{F}$ is $k$, then $|\mathcal{F}|=O(m^k)$.
\end{lemma}
In our setting, $\mathcal{F}$ contains one representative set $F(v)$ from each of the equivalence classes. Thus the size of $\mathcal{F}$ is the same as the number of equivalence classes. The universe is $\mathcal{S}=\{ (s_i,\Delta) ~\mid~ s_i \in S, \Delta \in \{+,-\}\}$. Thus the universe size is $2|S|$. Suppose towards contradiction that 
the VC-dimension of $\mathcal{F}$ is four. 
By the Sauer--Shelah Lemma, we get that there is a set $T=\{(i_1,\Delta_{i_1}), (i_2,\Delta_{i_2}),(i_3,\Delta_{i_3}),(i_4,\Delta_{i_4})\}$ that is shattered by 
$\mathcal{F}$.
We first argue that we can assume w.l.o.g. that $i_1<i_2<i_3<i_4$. 
We then consider two subsets $T_1=\{(i_1,\Delta_{i_1}), (i_3,\Delta_{i_3})\}$ and $T_2=\{(i_2,\Delta_{i_2}), (i_4,\Delta_{i_4})\}$, and show that it cannot be that there are two vertices $t_1,t_2$ such that $F(t_1)\cap T=T_1$ and 
$F(t_2)\cap T=T_2$. This implies that the VC dimension of $\mathcal{F}$ is at most $3$, and thus that $|\mathcal{F}|=O(|S|^3)$. The proofs of these arguments are quite tedious, as we need to consider many cases, but aside from that, each of the cases is rather easy and follows immediately from the planar embedding.

\dnsparagraph{Fast Computation}
We continue with the unweighted case of $D$-diameter planar graph $G=(V,E)$ with $S \subseteq V$ sources on a face. A na\"{\i}ve computation applies a SSSP (single-source shortest-path) algorithm for each $s \in S$, which takes $O(|S|\cdot n)$ time. Our goal is compute all distance tuples in time $\widetilde{O}(n +|S|^3 \cdot D)$ time. 

At a high level, we will follow the multiple-source shortest path (MSSP) algorithm with all sources lying on a common face, from \cite{klein2005multiple}, while maintaining \emph{hashes} of distance tuples. Observe that we cannot explicitly maintain the size-$|S|$ distance tuple for each vertex, since that is $n\cd |S|$ integers, which exceeds the promised time bound of $\tO(n+|S|^3 \cdot D)$ if $|S|$ is large (say, $n^{\Om(1)}$). Therefore, we maintain hashes of size $O(\log n)$ instead. The tricky part is to efficiently update the distance tuples while running the MSSP algorithm. Our hash function is motivated by Rabin-Karp string hashing. For a distance tuple $\bd=(d_1,\lds,d_\el)$, we define its \emph{hash value} under \emph{base} $b$ and \emph{modulus} $p$ (for $p$ prime) as
$h(\bd,b,p) := \sum_{i=1}^\el d_i b^i \bmod p$.
Clearly, if two distance tuples are equal, then their hash values under the same base and modulus are equal. We then claim that for two distinct distance tuples, their hash values are likely to be different under a \emph{random} base, as long as the modulus is large enough.

The MSSP algorithm begins with computing the SSSP tree on the first source $s_1$. It then travels along the face segment in the order $(s_2,s_3,\lds,s_\el)$, temporarily setting each $s_j$ as the source, while maintaining a \emph{dynamic forest} $F$ of values, one for each vertex in $V$. (It also maintains a dynamic forest on the \emph{dual graph}, but we do not need to discuss that here.) 
Our algorithm will maintain another dynamic forest $F'$ on the vertices in $V$ that is updated alongside the MSSP algorithm, so that at the end, the value at each vertex $v$ is precisely the hash value of its tuple. See Sec. \ref{sec:fastunweighted} for the detailed description.

The weighted case will be similar to the unweighted one with a key crucial difference. Here, since we are looking for approximate distances, what we need is not a hash function, but a ``clustering'' function that groups together vertices whose distance tuples are close together (say, in $\el_2$-distance). For this, we will use the Johnson-Lindenstrauss (JL) dimension reduction scheme. The complete algorithm appears on Sec. \ref{sec:fastweighted}.

\dnsparagraph{Implications to Distance Oracles} Bounding the number of tuples by $O(|S|^3D)$, immediately leads to an efficient compact distance oracles scheme for maintaining $S \times T$ distances. The oracle will contain the $O(|S|^3\cdot D)$ distinct distance tuples using 
 $O(|S|^3\cdot D\log n)$ bits. Next, the distance tuple of each vertex $t \in T$ can be encoded with $O(\log(|S|D))$ bits (i.e., encoding the index of the tuple of $t$ in the list of all tuples). Given a query $s,t \in S \times T$, the oracle can compute $d_G(s,t)$ by extracting this information for the tuple of $t$ in constant time. 
The preprocessing time of the construction is linear (in the size of the oracle), due to the fast computation of the tuples.

\dnsparagraph{Implications to Diameter Computation}
In the context of diameter computation, $S$ will be the cycle separator of size $O(D)$. To compute the core-set $V^{*+}$, we simply take one representative vertex in $G^+$ for each of the $O(D^4)$ equivalence classes. The core-set $V^{*-}$ is defined analogously. This leads immediately to an $\widetilde{O}(\poly(D) \cdot n)$-time deterministic algorithm for the unweighted case, by plugging it in the recursive procedure described before. Similar bounds are obtained by \cite{chan2017faster}, using the heavy machinery of abstract Voronoi diagram.

\paragraph{The Weighted Case.}
For the weighted case, in Sec. \ref{sec:weighted-compr}, we consider an $(1+\epsilon)$-approximate compression scheme that maintains a $(1+\epsilon)$ approximation for all $S \times T$ distances. 
Here, we first provide a compression scheme with an \emph{additive} error with respect to the \emph{weighted diameter} of the graph. We then reduce the multiplicative error case to the additive via the use of low-diameter decompositions. We remark that for the purpose of diameter computation, the additive compression scheme suffices. 
The additive approximate compression is based on the notion of \emph{(additive) close} and \emph{(additive) core-set}. 
\BD[$\de$-Additive Close]
Let $G=(V,E)$ be a graph, let $S=\{s_1,\lds,s_\el\}\s V$ be a subset of vertices. Two vertices $u,v\in V$ are $\de$-additive close with respect to $S$ if
\[ |d(u,s_i) - d(v,s_i)| \le \de \qquad \forall i\in[\el] .\]
\ED
\BD[Additive Core-Set]
Let $G=(V,E)$ be a graph, let $S=\{s_1,\lds,s_\el\}\s V$ be a subset of vertices, and let $\de\ge0$ be an additive error parameter. A subset $V'\s V$ is a \emph{$\de$-additive core-set} with respect to $S$ if for all vertices $v\in V$, there exists a vertex $v'\in V'$ that is $\de$-additive close to $v$ w.r.t $S$.
\ED
Our goal will be to prove the existence of a core-set of size $\poly(|S|/\e)$.
\BT
Let $G=(V,E)$ be a weighted graph, and let $d>0$ be a parameter. Let $S=(s_1,\lds,s_\el)$ be a sequence of points on a common face arranged in cyclic order, such that the distance between any two consecutive points is at most $d$. Then, there exists a $\de$-additive core-set of size $O(\el^6(d/\de)^4)$.
\ET
We will define sets $A_i^\De$ similarly to the ones in the unweighted case, but with more values of $\De$. In particular, we will consider all multiples of $\de' /\el$ from roughly $-d$ to $d$. We will then apply 
a similar VC dimension argument as in the unweighted case, but with several subtleties as the weighted sets $A_i^\De$  are more involved. 

\subsection{Distributed Tools and Unweighted Diameter}

\dnsparagraph{The challenge} The efficiency our of recursive diameter computation critically depends on the size of the separator. Recall that also the bound on the size of the core-sets is a function of the separator size and the diameter of the graph. 
\cite{GhaffariP17} provided an $\widetilde{O}(D)$-round algorithm that computes a shortest path separator, thus a separator of size $O(D)$. In our algorithm the separator should be computed recursively, until all components are sufficiently small.  The key challenge is that already after the first computation, once we remove the vertices of the separator from the graph, the diameter of each of the components might be $\Omega(n)$.  
We note that although the DFS construction of \cite{GhaffariP17} also applied the separator algorithm in a recursive manner, for their purposes it was sufficient for the separator to be a \emph{path}, and its length could be arbitrarily large. In our setting in contrast, we need to come up with a different recursive scheme that preserves the diameter of the subgraphs throughout all recursion layers. This is our motivation for defining the bounded diameter decomposition.
%
%
%
\dnsparagraph{Bounded Diameter Decomposition (BDD)}
Informally, the bounded diameter decomposition is described by a recursive procedure that given a subgraph $G'$ of diameter $D'$, breaks down $G'$ into small components $G'_1,\ldots, G'_k$ each with at most $|G'_j|\leq |G'|/c$ vertices, for some constant $c$, such that:
\begin{itemize}
\item The diameter of each $G'$ is bounded as a function of $D'$.
\item Each edge $e \in G'$ appears on a small number of $G'_j$ subgraphs.
\end{itemize}
The first property is important for being able to compute an $O(D')$-separator recursively. The second property is important for parallelizing the computation on all the components (i.e.,  via the random-delay approach). 
Our formal definition of BDD is in fact considerably more delicate for the following reasons.
Let $T$ be a BFS tree on which the shortest-path separator is computed in the first recursion level. Let $S$ be the balanced cycle separator\footnote{This cycle separator might contain a non-$G$ edge, which we simulate as a virtual edge.} of $G$. To define the child components of $G$, there are two options. The first defines the 
child components in $G \setminus S$. while this satisfies the second property, it might violate the first property. Alternatively, one might first define the interior and exterior subgraphs w.r.t the cycle $S$ and then augment both parts with the vertices in $S$. This satisfies the first property, but now as $S$ is added to both parts, an edge might appear later on, on many subgraphs in the same recursion layer. 

To get out of this impasse, our technique adds segments of $S$ to each of the components, while guaranteeing that the second property holds. Since we do not add $S$ entirely to both parts, this might increase the diameter of the components. We then show that this increase is rather controlled, by an additive $+D$ term, in each recursive level. Thus, after all $O(\log n)$ recursion levels, the diameter is still bounded by $O(D\log n)$. This recursive decomposition continues until the components have size $O(D \log n)$. A useful property of the components is they have an $O(D\log n)$-depth spanning tree that contains at most $O(\log n)$ edges that are not in the BFS tree. This property will become useful in the weighted setting. 

\paragraph{Handling $1$-Connected Subgraphs} 
The basic separator algorithm of \cite{GhaffariP17} (which we will use throughout) requires that the boundary of each face is a simple cycle which indeed holds for biconnected subgraphs. Since we need to compute the separator recursively, even if the original input graph is biconnected, after one recursive layer it can become $1$-connected. \cite{GhaffariP17} handled this by computing the biconnected components of the graph and solving the problem for each piece separately. We take a rather different approach that allows us to simulate the separator algorithm for biconnected graphs in $1$-connected graphs with a small overhead in the number of rounds. This reduction is based on adding ``virtual" edges to the $1$-connected graph in order to make it biconnected. We then simulate these virtual edges by providing low-congestion and short path between the endpoints of the virtual edges. This tool of \emph{distributed biconnected augmentation} might provide a cleaner and more general way to handle $1$-connected graphs in other settings as well.

\paragraph{Exact Distance Labels}
To facilitate the recursive diameter computation, we first compute exact labels of $\widetilde{O}(D)$-bits. Our labels are based on the well-known scheme of Gavoille et al.\ \cite{gavoille2004distance} that has the following recursive structure. Each label $L_G(v)$ consists of (i) the $G$-distances from $v$ to all the vertices in the separator $S$ of $G$, (ii) the component ID of $v$ in $G \setminus  S$, and (iii) the label $L_{G'}(v)$, where $G'$ is the component of $v$ in $G \setminus S$. 

Our goal is to compute these labels in a top-down manner over the recursion tree of the BDD decomposition. The key challenge here is that unlike the recursion of \cite{gavoille2004distance} in which the child components are vertex-disjoint, here subgraphs are not vertex disjoint, and a vertex might belong to many subgraphs in the same recursion level. The fact that the components of \cite{gavoille2004distance} are disjoint, implies that a vertex belongs to $O(\log n)$ components in total, thus keeping the label small. In our case, we will not be able to keep a sub-label of $v$ for each of the subgraphs for which it belongs. To handle that, we will use the fact that the components in the recursive partitioning of \cite{gavoille2004distance} are fully contained in the components of the BDD decomposition. 

\paragraph{Distributed Unweighted Diameter.}
For the sake of explanation, we sketch here a $\poly(D)$-round algorithm. Obtaining the $\widetilde{O}(D^5)$- round algorithm calls for various combinations of techniques that shave off some of the $D$-factors. We first compute the distance labels of the vertices. 
The subsequent diameter computation works in a bottom-up manner on the BDD decomposition tree. The invariant that we will maintain is that in step $i$, all vertices in the $(D-i+1)$-level subgraphs $G'$ (in the BDD decomposition) have already computed the largest distance in $G$ (and not in $G'$) over all vertex pairs in $G'$. 
The leaf components of the BDD decomposition have $O(D)$ many vertices, and thus every vertex can collect the distance label of all the vertices in its leaf components in $\widetilde{O}(D^2)$ rounds. Since all these subgraphs are almost-edge disjoint, this computation can be done in parallel in all of these subgraphs.

Consider the $i^{th}$ phase of this process. For every $(D-i+1)$-level subgraph $G'$, there are two options. Either the farthest pair $u,v$ in $G'$ is contained in one of the child components, or that $u$ and $v$ are in different child components. The interesting case is the second one. Since $u$ and $v$ are in different component in the BDD, they are separated by the shortest path separator $S$ of $G'$. To compute the largest distance between vertices in different child components, we first let all vertices in the separator send their label to a global leader in $G'$. This is a total of $\widetilde{O}(|S|\cdot D)$ bits of information, and since $D(G')=O(D\log n)$, using standard pipeline procedure it can be implemented in $\widetilde{O}(D^2)$ rounds. At this points, all the vertices in $G'$ can compute their distance tuple with respect to $S$. Thanks to our metric compression solution, there are $O(D^4)$ distinct distance tuples. The last step is to aggregate all these tuples at a leader in $G'$. This  can be done in $O(D^6)$ rounds, but in our algorithm we do it more efficiently in $\widetilde{O}(D^5)$ rounds  by compressing each distance tuple into $O(\log n)$ bits. Once the leader in $G'$ receives all the distance tuples, it has all the information to compute the distance in $G$ between each pair in different child components. This holds since every such $u$-$v$ shortest-path must intersect $S$ in some vertex $s$, thus $d_G(u,v)=d_G(u,s)+d_G(s,v)$, these distances are contained in the distance tuple information. 

\subsection{Distributed Weighted Diameter}
The computation of $(1+\epsilon)$ approximate diameter is considerably more involved. It consists of several steps, and thus intermediate results (e.g., SSSP tree) which are important on their own.

\dnsparagraph{Step (1): SSSP via Distance Labels}
The major step here is the computation of exact distance labels with $\widetilde{O}(D)$-bits. Once we compute such labels, an SSSP can be easily defined: the source node $s$ sends its label $L_G(s)$ on the BFS tree to all the vertices. This allows each vertex $v$ to compute $d_G(s,v)$ based on its own label $L_G(v)$, and label of $L_G(s)$ of the source. By exchanging these distances with their neighbors, every vertex can compute its parent in the SSSP tree. Since the label size is $\widetilde{O}(D)$, we get that given exact labels, the SSSP can be computed within extra $\widetilde{O}(D^2)$ rounds. 

From that point on, we focus on labels computation.
Interestingly, we will compute these labels using the \emph{unweighted} communication backbone of the \emph{unweighted BDD}, namely, a BDD on a BFS tree. The key difference from the unweighted distance labels is that here we cannot afford to compute an SSSP from each of the separator vertices. Recall that efficient SSSP is the reason for computing this labels from first place. The key idea of our algorithm is it ``morally" applies the scheme of Gavoille et al.\ \cite{gavoille2001distance} but in a \emph{bottom-up} rather than a top-down manner. 

We will start from the leaf components which have $O(D)$ vertices. In each such leaf component $G'$, a vertex $v$ can collect the entire subgraph and locally compute a label $L_{G'}(v)$ that consists of all its distances to the vertices in $G'$. We now work from the leaf up, and consider the level-$i$ components in the BDD recursion tree. 
Consider a subgraph $G'$ and its children components $G'_1,\ldots, G'_k$ in level $i+1$. Let $S$ be the separator of $G'$ (based on which the child components are defined). By the induction assumption, we assume that each vertex $v$ in $G'_j$ has already computed its label $L_{G'_j}(v)$. Thus by the recursive label's structure of Gavoille et al.\, to compute $L_{G'}(v)$, it is sufficient to augment the sub-label $L_{G'}(v)$ with the distances (in $G'$) to each of the separator vertices in $S$. By letting each separator vertex $s \in S$ send their labels $\{L_{G'_j}(s), ~\mid~ s \in G'_j\}$ to all vertices in $G'$, every vertex $v \in G'_j$ for every $j$ has now the sufficient information to compute its distance to $S$ in $G'$ (using also its own label $L_{G'_j}(v)$). In the analysis, we show that each separator vertex appears on at most two child components of $G'$ (but potentially on many other subgraphs in this level), thus the total amount of information to be sent is $O(D^2)$. 
Since a separator vertex in level $i$ might appear in many subgraphs of levels $j\geq i$, we will mimic again the Gavoille et al. label structure, and shorten the labels of the vertices once we get to a level in which they are part of the separator. These shortening will be vital to keep the labels small.

\dnsparagraph{Step (2): BDD Decomposition on the SSSP Tree}
At this point, we already have all the ingredients necessary for an $(1+\epsilon)$ approximation in $\poly(D)\cdot \poly(1/\epsilon)$ rounds. Such an algorithm can be obtained by applying the exact same algorithm for the unweighted diameter, with the only difference is that the core-set size will be $\poly(D)\cdot \poly(1/\epsilon)$.

Obtaining a better bound of $\widetilde{O}(D^3)+D^2\cdot\poly(\log n/\epsilon)$ calls for several improvements.
Instead of computing the tuples w.r.t all vertices on the cycle separator, we will select $O(1/\epsilon)$ portal nodes on this cycle (as in \cite{weimann2016approximating}). For this approach to work, the portals cannot be selected arbitrarily, but rather should be selected carefully on a \emph{shortest-path} separator. Since the separator computed on a BFS tree is no longer a shortest-path in a weighted graph, we need to apply the BDD scheme on the SSSP tree. This will guarantee that the separator computed in each recursive layer will consist of a concatenation of $O(\log n)$ shortest path segments. We will then be able to mark $O(1/\epsilon)$ portal vertices on each such segment, and ignore the remaining vertices on the separator. 

Computing a BDD on the SSSP tree brings along several complications. The major one is that the unweighted diameter of each component might be very large. The standard remedy for these kind of problems is low-congestion shortcuts. 
However, as the components in each recursion level are not vertex disjoint, an additional argument is required in order to be able to compute the low-congestion shortcuts. For that purpose, we define the BDD in a more careful manner that guarantees the following: each edge $e \in G$ might appear on at most \emph{two} components in the same level -- one from each side of that edge. This allows us to apply the graph simulation technique of \cite{GhaffariP17}. This technique projects $G$ into a different graph $G'$ that contains the vertices of $G$ plus additional vertices. The subgraphs in the BDD level are mapped in $G'$ to vertex disjoint subgraphs, which allows safe application of low-congestion shortcuts in $G'$. The vertices will then simulate the low-congestion computation in $G'$ and will translate it back to the edge of $G$. 
%
%
%
\dnsparagraph{Step (3): Recursive Diameter Computation}
The algorithm has the same high level structure as for the unweighted, only that is works on the weighted BDD (i.e., BDD on the shortest-path tree) and uses approximate core-set with respect to a collection of $O(\log n/\epsilon)$ \emph{portals} on the separator (defined as in \cite{weimann2016approximating}). In each independent level of the recursion, given a component $G'$, compute a cycle separator $S$ that defines the subgraphs $G^+$ and $G^-$.
We then compute exact distance labels in $G^+$ and in $G^-$ using a total of $\widetilde{O}(D^3)$ rounds. Next, we restrict attention only to $O(\log n/\epsilon)$ portals on the separator of $G'$, and 
compute the $(1+\epsilon)$ approximate core-sets in $G^+$ and $G^-$.
The portals then send their exact labels (in $G$) over the low-congestion shortcuts of $G'$. Since each edge of these shortcuts appears on $O(D)$ subgraphs, the total amount of information that we pass through an edge is $\poly(\log n/\epsilon)\cdot D^2$. The method of random delay \cite{ghaffari2015near,LMR94-routing} then allows us to work on all components in parallel with a total round complexity of $\widetilde{O}(D^3)+\poly(\log n/\epsilon)\cdot D^2$.  

\subsection{Preliminaries}
\dnsparagraph{Graph Notations}
For a weighted graph $G=(V,E,\omega)$, let $d_G(u,v)$ be the total weight of the shortest path between $u$ and $v$ in $G$. When $G$ is clear from the context, we may omit it.
For a tree $T \subseteq G$, let $T(z)$ be the 
subtree of $T$ rooted at $z$, and let $\pi(u,v,T)$ be the tree path between $u$ 
and $v$, when $T$ is clear from the context, we may omit it and simply write 
$\pi(u,v)$. For a subset of vertices $S_i\subseteq V(G)$, let $G[S]$ be the induced subgraph on 
$S$.

\dnsparagraph{Planar Embeddings} The geometric planar embedding of graph $G$ is a drawing of $G$ on a plane so that no two edges intersect. A combinatorial planar embedding of $G$ determines the clockwise ordering of the edges of each node $v\in G$ around that node such that all these orderings are consistent with a plane drawing (i.e., geometric planar embedding) of $G$. Ghaffari and Haeupler \cite{ghaffari2016embedding} gave a distributed algorithm that computes a combinatorial planar embedding in $O(D \min\{\log n, D\})$ rounds, where each node learns the clockwise order of its edges.

\dnsparagraph{Low-Congestion Shortcuts} In a subsequent paper \cite{ghaffari2016distributed}, Ghaffari and Haeupler introduced the notion of low-congestion shortcuts, which provides as basic communication backbone in many planar algorithms. The definition is as follows.
\begin{definition}
\label[definition]{def:lowcongshortcut} (\textbf{$\alpha$-congestion $\beta$-dilation shortcut})
Given a graph $G=(V,E)$ and a partition of $V$ into disjoint subsets $S_1, \ldots, S_N \subseteq V$, each inducing a connected subgraph $G[S_i]$, we call a set of subgraphs $H_1, \ldots, H_N \subseteq G$, where $H_i$ is a supergraph of $G[S_i]$, an $\alpha$-congestion $\beta$-dilation shortcut if we have the following two properties:
(1) For each $i$, the diameter of the subgraph $H_i$ is at most $\beta$, and  
(2) for each edge $e \in E$, the number of subgraphs $H_i$ containing $e$ is at most $\alpha$. 
\end{definition}
Ghaffari and Haeupler~\cite{ghaffari2016distributed} proved the existence of almost optimal low-congestion covers, as well as providing efficient algorithm to compute them. 

\begin{fact}[Optimal Low-Congestion Shortcuts \cite{GhaffariP17}]
Any partition of a $D$-diameter planar graph into disjoint subsets $S_1, \ldots, S_N \subseteq V$, each inducing a connected subgraph $G[S_i]$, admits an $\alpha$-congestion $\beta$-dilation shortcut where $\alpha= O(D\log D)$ and $\beta= O(D\log D)$. Moreover, there is a randomized algorithm that computes these shortcuts within $\tilde{O}(D)$, with high probability. 
\end{fact}

\dnsparagraph{Distributed Shortest-Path Separators} 
Throughout we will also make an extensive use of the separator algorithm by Ghaffari and Parter. This algorithm is based on tweaking the low-congestion shortcut machinery to allow fast computation on the \emph{dual} graph. 
\begin{fact}[Distributed Shortest-Path Separator, \cite{GhaffariP17}]
There is a randomized algorithm that given a $D$-diameter graph computes a shortest-path separator in $\widetilde{O}(D)$ rounds, with high probability. 
\end{fact}

\dnsparagraph{Distributed Scheduling of Algorithms} 
Our algorithms are based on recursive graph decomposition, where in every level of the recursion we will need to work in parallel on several subgraphs. Since the \congest model allows sending only $O(\log n)$ bits on each edge per round, we will use the scheduling framework of \cite{ghaffari2015near}. In our context, this framework implies that if each edge appears on a small number of subgraphs in each recursion level, then all the algorithms (one per subgraph) can be scheduled within almost same number of rounds as a single algorithm. 
\begin{fact}[Scheduling, \cite{ghaffari2015near}]
Given a sequence of algorithms $\mathcal{A}_1,\ldots, \mathcal{A}_k$, each taking at most $\dilation$ rounds, and where for each edge, at most $\congestion$ messages
are sent through it in total over all these algorithms, then all algorithms can run in total of $O(\dilation+\congestion)$ rounds. 
\end{fact}

\paragraph{The Balanced Cycle Separator Algorithm by \cite{GhaffariP17}.}
The algorithm of \cite{GhaffariP17} gets as input a biconnected graph $G$, and a spanning tree $T \subset G$. It outputs a cycle consisting of a tree path in $T$ plus one additional edge (possibly not in $G$). 
In the high-level the algorithm computes this separator by considering the dual-tree $T'$  of $T$. 
The nodes of this dual tree are the faces of $G$, and two dual-nodes are connected in $T'$ if their faces share an non-tree edge $e \notin T$. The dual tree is rooted at the outface, and each dual-node $v'$ is given a weight as follows: consider the superface obtained by merging all faces (dual-nodes) in the subtree of $T'$ rooted at $v$; then the weight of $v'$ is the number of nodes on its superface boundary plus the number of nodes inside the superface. 
The algorithm first compute a $(1+\epsilon)$ approximation of all dual nodes in $T'$.
This computation is done on the dual tree whose vertices and edges are not part of $G$ and thus call for special tool. 
Once all weights are computed, the algorithm first attempts at finding a \emph{balanced} dual-node, a dual-node whose weight is in $[n/(3(1+\epsilon)) ,2(1+\epsilon)n/3]$. If such balanced dual-node than the boundary of its superface is the fundamental cycle separator (all the edges of this cycle are in $G$). Otherwise, there must be a \emph{critical} dual-node such that its own weight is large but the weight of each of its children in $T'$ is small. 
In this case, the algorithm mimics Lipton-Tarjan algorithm by roughly speaking ``triangulating" the face of this critical dual node. In this latter case, the cycle separator contains one edge that is not in $G$. 

\dnsparagraph{Road-Map} We start by presenting the metric compression problem in Sec. \ref{sec:compr}. This provides the combinatorial basis for our diameter algorithms. Sec. \ref{sec:unbounded-diam} presents the key tool of bounded diameter decomposition and computation of exact diameter in unweighted graphs. In Sec. \ref{sec:labelssp} we describe the construction of exact distance labels in weighted graphs and a construction of SSSP (singe source shortest path) tree. Finally, we conclude with Sec. \ref{sec:approxdiam} that provides the computation of $(1+\epsilon)$ approximation for the diameter in weighted graphs.

\section{The Metric Compression Problem}\label{sec:compr}
In Sec. \ref{sec:exact-unweighted}, we consider the setting of exact compression for unweighted graphs.
In Sec. \ref{sec:weighted-compr}, we extend the result to the $(1+\epsilon)$-approximate compression for weighted graphs. Finally, in Sec. \ref{sec:fast}, we consider the computational aspects of the problem. We present a linear time algorithm for computing the tuples. 

\subsection{Exact Compression for Unweighted Graphs}\label{sec:exact-unweighted}

This section focuses on proving the compression part of \Cref{thm:exact}. Let $G$ be a planar graph under some planar embedding, and let $S$ be a subset of vertices lying consecutively on a common face. For each vertex $v\in V$, define the \emph{distance tuple} of $v$ to be its distances to $S$ stored in a tuple, defined formally as follows. 

\BD[Distance tuple]
Let $G=(V,E)$ be a graph, and let $S=\{s_1,\ldots, s_\el\}\s V$ be a subset of vertices. The \emph{$S$-distance tuple} of a vertex $v\in V$ denoted by $\tuple_S(v)$ is the function that maps each vertex $v\in V$ to the vector of distances $\tuple_S(v)=\langle d_G(s_1,v),\ldots, d_G(v,s_{\ell}) \rangle$. When the set $S$ is clear from the context, we simply use \emph{distance tuple} and $\tuple(v)$.
\ED

To compress the distances, we will actually show that, perhaps surprisingly, there are only $O(|S|^3D)$ many possible distinct distance tuples. That is, if $n >>|S|^3 \cdot D$, then there are many repeated distance tuples among the vertices. This means that we can simply keep a table of the $O(|S|^3D)$ distance tuples, and store for each vertex the index into the tuple, which has size $O(\log n)$. Therefore, the size of the compression is $O(|S|^3D+n\log(|S|D))$. We will show:

\begin{theorem}[\Cref{thm:exact}, Compression] \label{thm:exact-compr}
Given an $n$-vertex unweighted planar graph $G=(V,E)$ and a set $S \subseteq V$ for sources lying consecutively on a single face, the number of distinct $S$-distance tuples is bounded by $O(|S|^3D)$.
\end{theorem}

Our proof consists of two main steps. First, we define an alternative representation of distance tuples that utilizes the definition below.

\BD
For each $i\in[\el-1]$ and $\De\in\{-1,0\}$, define the set $A_i^\De:=\{v\in V:d(v,s_i) \le d(v,s_{i+1})+\De\}$. 
Define the family
\begin{gather*}
\m F(S):=\{ \{ (i,\De):v\in A_i^\De \} : v\in V\} .
\end{gather*}
\ED

We show that, modulo a factor of $O(D)$, our task of bounding the number of distance tuples reduces to bounding the size of $\m F(S)$. Second, we prove the size bound $|\m F(S)| \le O(|S|^3)$ using a VC-dimension argument. In particular, we show that the set system represented by $\m F(S)$ has VC dimension at most $3$. Combining these two steps proves the desired $O(\el^3 D)$ bound on distinct distance tuples.

\subsubsection{Reduction to VC Dimension Argument}

We now proceed with the technical details, beginning with the alternative representation step. Fix an arbitrary $r\in[\el]$, and define $[D]_0 := \{0,1,2,\lds, D\}=\{0\}\cup[D]$. Our domain will be $[D]_0\times \m F(S)$, and each vertex $v\in V$ will be represented by the tuple $(d(v,s_r),\{(i,\De):v\in A_i^\De\}) \in [D]_0\times \m F(S)$. Here, we use the fact that the graph diameter is at most $D$, so $d(v,s_r)\in[D]_0$. We first prove the following claim, which establishes a surjective map from $[D]_0\times \m F(S)$ to the set of distinct distance tuples. This bounds the number of distinct distance tuples by the size of the domain $[D]_0 \times \m F(S)$.

\BCL\label{clm:canon}
Let $r$ be any integer in $[\el]$. For each vertex $v\in V$, its distance tuple is determined by the value of $d(v,s_r)$ and which sets $A_i^\De$ contain $v$. More formally, there is a function $f$ from $\N\times 2^{[\el]\times\{-1,0\}}$ to the set of $S$-distance tuples such that for all $v$, the distance label of $v$ is precisely $f(d(v,s_r),\{(i,\De):v\in A_i^\De\})$.
\ECL
\BP
Fix a vertex $v\in V$; we will reconstruct the distance tuple for $v$ based on $d(v,s_r)$ and which sets $A_i^\De$ contain $v$. The value $d(v,s_r)$ is already known; we now proceed to calculate $d(v,s_i)$ for $i<r$. First, note that for all $i\in[\el-1]$, by the triangle inequality and the fact that $s_i,s_{i+1}$ are distance $1$ apart, we have $d(v,s_i)-d(v,s_{i+1})\in\{-1,0,1\}$. Whether or not $v\in A_{r-1}^{-1}$ determines whether or not $d(v,s_{r-1})-d(v,s_r)=-1$. If so, then we must have $d(v,s_{r-1})=d(v,s_r)-1$, and we are done. Otherwise, $d(v,s_{r-1})-d(v,s_r)$ is either $0$ or $1$. It must be $0$ if $v\in A_i^0$, and otherwise, it must be $1$, so in either case, we are done.

We can now proceed inductively from $r-1$ to $1$: knowing $d(v,s_i)$ for $i\in[2,r]$, we can deduce $d(v,s_{i-1})$. This gives us all distances $d(v,s_i)$ for $i\in[r]$. For the remaining distances $d(v,s_i)$ for $i\in[r+1,\el]$, we can proceed analogously.
\EP

\subsubsection{VC Dimension Argument}

Since $[D]_0$ has size $D+1$, to bound the domain size by $O(|S|^3D)$,  it suffices to bound $\m F(S)$ by $O(|S|^3)$. 
%
We will prove that the VC-dimension of $\m F(S)$ is at most $3$, and then apply the well-known \emph{Sauer's Lemma}. 

\begin{definition}[VC Dimension of a Set System]
Let $X$ be a set of elements, called the universe. A family $\m F$ of subsets of $X$ has \emph{VC dimension $d$} if $d$ is the largest possible size of a subset $Y\s X$ satisfying the following property: for any subset $Y'\s Y$, there exists subset $F\in \m F$ such that $Y\cap F=Y'$.
\end{definition}

\begin{theorem}[Sauer's lemma]\label{thm:Sauer}
Let $X$ be a set of elements. If a family $\m F$ of subsets of $X$ of VC dimension $d$, then $|\m F|=O(|X|^d)$.
\end{theorem}

We now proceed with the VC dimension argument. For convenience, we redefine $\m F(S)$ in the statement of the theorem.

\begin{theorem}[Bounded VC-Dimension]\label{thm:vc-dim}
Let $G=(V,E)$ be an unweighted planar graph, and let $S:=(s_1,s_2,\lds,s_\el)$ be consecutive vertices on a face, ordered in clockwise or counter-clockwise order. 

For each $i\in[\el-1]$ and $\De\in\{-1,0\}$, define the set $A_i^\De:=\{v\in V:d(v,s_i) \le d(v,s_{i+1})+\De\}$. 
Define the universe $X:=[\el]\times \{-1,0\}$, and the family
\begin{gather}
\m F(S):=\{ \{ (i,\De):v\in A_i^\De \} : v\in V\} \s 2^{X} .\label{eq:F}
\end{gather}

Then, the VC dimension of $\m F(S)$ (on universe $X$) is at most $3$.
\end{theorem}

To prove the theorem, we use the following auxiliary lemma involving drawings in the plane below, whose easy but tedious proof is deferred to \Cref{sec:planar-drawing}.

We say that two (non-self-intersecting) arcs $C_1$ and $C_2$ \emph{cross}, when the following holds: there are two points $p$ and $q$ on both $C_1$ and $C_2$ (possibly $p=q$) and a simple curve $C$ between $p$ and $q$ (possibly the single point $p$ if $p=q$) satisfying the following: there exists a $\e_0>0$ such that for all positive $\e<\e_0$, the boundary of the region $B_\e(C)$ of all points within distance $\e$ from a point on $C$ has exactly four intersection points with $C_1$ and $C_2$, and they can be arranged in clockwise order so that the first and third points are on $C_1$ but not $C_2$, and the second and fourth are on $C_2$ but not $C_1$.
\begin{restatable}{lemma}{drawing}\label{lem:drawing}
Consider a simple, closed curve drawn in the plane, with eight (not necessarily distinct) points $p_1,p_2,\lds,p_8$ placed clockwise around the curve that satisfy $p_1\ne p_2$, $p_3\ne p_4$,  $p_5\ne p_6$, and $p_7\ne p_8$. Consider two points $q_1,q_2$ inside the curve. It is impossible to draw arcs \linebreak $(q_1,p_1), (q_1,p_4), (q_1,p_5), (q_1,p_8),\ (q_2,p_2), (q_2,p_3), (q_2,p_6),(q_2,p_7)$ such that:
 \BE
 \im The arcs from $q_1$ do not pairwise cross, and the arcs from $q_2$ do not pairwise cross, and
 \im the arcs from $p_1$ and $p_2$ do not touch, the arcs from $p_3$ and $p_4$ do not touch, the arcs from $p_5$ and $p_6$ do not touch, and the arcs from $p_7$ and $p_8$ do not touch.
 \EE
\end{restatable}

Armed with \Cref{lem:drawing}, we now prove our VC dimension bound of \Cref{thm:vc-dim}.
\BP
Before we prove the theorem, we first remark that our proof will not actually use the fact that $\De$ takes on the values $-1,0$. Indeed, it can be adapted to work for $\De\in M$ for any set of real numbers $M$. This observation is needed for a smooth transition to the weighted case in \Cref{sec:small-additive}.

We argue by contradiction: suppose that the set $\m F$ has VC dimension at least $4$. Then, there exists a set $Y:=\{(i_1,\De_1),(i_2,\De_2),(i_3,\De_3),(i_4,\De_4)\}$ such that for each subset $Y'\s Y$, there exists vertex $v\in V$ such that for each $a\in[4]$,
\[v\in A_{i_a}^{\De_a} \iff (i_a,\De_a)\in Y'.\]

First, we argue that the values $i_1,i_2,i_3,i_4$ are all distinct. Suppose, otherwise, that $i_1=i_2$, and assume without loss of generality that $\De_1\le \De_2$. Consider the set $Y'=\{(i_1,\De_1)\}$; by assumption, there must be a vertex $v\in V$ satisfying $v\in A_{i_1}^{\De_1}$ and $v\notin A_{i_2}^{\De_2}=A_{i_1}^{\De_2}$. This means that
\[ d(v,s_{i_1})\le d(v,s_{i_1+1})+\De_1 \qquad\text{and}\qquad d(v,s_{i_1}) > d(v,s_{i_1+1})+\De_2 ,\]
but $\De_1\le \De_2$, so this is impossible, a contradiction.

Therefore, we can assume that $i_1,i_2,i_3,i_4$ are distinct, so assume without loss of generality that $i_1<i_2<i_3<i_4$. Consider the sets
\[ Y'_1:=\{(i_1,\De_1),(i_3,\De_3)\} \qquad\text{and}\qquad Y'_2:=\{(i_2,\De_2),(i_4,\De_4)\} .\]
By assumption, there must be a vertex $t_1\in V$ such that
\begin{align} d(t_1,s_{i_1})\le d(t_1,s_{i_1+1})+\De_1, \qquad  d(t_1,s_{i_2})> d(t_1,s_{i_2+1})+\De_2,\nonumber \\
d(t_1,s_{i_3})\le d(t_1,s_{i_3+1})+\De_3,
\qquad  d(t_1,s_{i_4})> d(t_1,s_{i_4+1})+\De_4,\label{eq:q1}\end{align}
and a vertex $t_2\in V$ such that
\begin{align} d(t_2,s_{i_1})> d(t_2,s_{i_1+1})+\De_1,\nonumber \qquad  d(t_2,s_{i_2})\le d(t_2,s_{i_2+1})+\De_2, \\
d(t_2,s_{i_3})> d(t_2,s_{i_3+1})+\De_3,
\qquad  d(t_2,s_{i_4})\le d(t_2,s_{i_4+1})+\De_4.\label{eq:q2}\end{align}

Consider the shortest paths between the pairs $(t_1,s_{i_1}),(t_1,s_{i_2+1}),(t_1,s_{i_3}),(t_1,s_{i_4+1}),$\linebreak $(t_2,s_{i_1+1}),(t_2,s_{i_2}),(t_2,s_{i_3+1}),(t_2,s_{i_4})$. 

We can assume that the paths from $t_1$ do not cross in their planar embeddings, since if two paths cross at a vertex $v\in V$, then we can modify one of the paths to agree with the other path up until vertex $v$, while still keeping it a shortest path.

Assume without loss of generality that the path $S=(s_1,\lds,s_\el)$ is numbered in clockwise order around the outer face. Since $i_1<i_2<i_3<i_4$, the points $s_{i_1},s_{i_1+1},s_{i_2},s_{i_2+1},s_{i_3},s_{i_3+1},s_{i_4},s_{i_4+1}$ are also in clockwise order around the outer face. Moreover, $s_{i_1}\ne s_{i_1+1}$, $s_{i_2}\ne s_{i_2+1}$,  $s_{i_3}\ne s_{i_3+1}$, and $s_{i_4}\ne s_{i_4+1}$. Therefore, we can invoke \Cref{lem:drawing} on $(s_{i_1},s_{i_1+1},s_{i_2},s_{i_2+1},s_{i_3},s_{i_3+1},s_{i_4},s_{i_4+1})$ with the shortest paths from $t_1,t_2$. Note that Condition~1 of \Cref{lem:drawing} is true, so it must be Condition~2 that is false. In other words, one of the following cases must hold:
 \BE
 \im The shortest paths between $(t_1,s_{i_1})$ and $(t_2,s_{i_1+1})$ intersect at some vertex $v_1\in V$.
 \im The shortest paths between $(t_1,s_{i_2+1})$ and $(t_2,s_{i_2})$ intersect at some vertex $v_2\in V$.
 \im The shortest paths between $(t_1,s_{i_3})$ and $(t_2,s_{i_3+1})$ intersect at some vertex $v_3\in V$.
 \im The shortest paths between $(t_1,s_{i_4+1})$ and $(t_2,s_{i_4})$ intersect at some vertex $v_4\in V$.
 \EE
All four cases are similar, but for completeness, we will go through all the cases in order, starting from Case~1. Assuming Case~1, define $P_{1,1}$ and $P_{2,2}$ to be the shortest paths between $(t_1,s_{i_1})$ and $(t_2,s_{i_1+1})$, respectively. Now, consider a path $P_{1,2}$ that travels from $t_1$ to $v_1$ along $P_{1,1}$, and then from $v_1$ to $s_{i_1+1}$ along $P_{2,2}$. Also, consider a path $P_{2,1}$ that travels from $t_2$ to $v_1$ along $P_{2,2}$, and then from $v_1$ to $s_{i_1}$ along $P_{1,1}$. For a path $P$, let $\len(P)$ be the number of edges on the path $P$. We have 
\begin{gather} d(t_1,s_{i_1+1})+d(t_2,s_{i_1})\le \len(P_{1,2})+\len(P_{2,1})=\len(P_{1,1})+\len(P_{2,2})=d(t_1,s_{i_1})+d(t_2,s_{i_1+1})   .\label{eq:sum}\end{gather}
However, summing up the inequality $d(t_1,s_{i_1})\le d(t_1,s_{i_1+1})+\De_1$ from (\ref{eq:q1}) and the inequality $d(t_2,s_{i_1+1})+\De_1< d(t_2,s_{i_1})$ from (\ref{eq:q2}) gives
\[ d(t_1,s_{i_1})+d(t_2,s_{i_1+1})+\De_1 < d(t_1,s_{i_1+1})+\De_1+d(t_2,s_{i_1}) ,\]
which contradicts (\ref{eq:sum}).

For Case~2, define $P_{1,1}$ and $P_{2,2}$ to be the shortest paths between $(t_1,s_{i_2+1})$ and $(t_2,s_{i_2})$, respectively. Now, consider a path $P_{1,2}$ that travels from $t_1$ to $v_2$ along $P_{1,1}$, and then from $v_2$ to $s_{i_2}$ along $P_{2,2}$. Also, consider a path $P_{2,1}$ that travels from $t_2$ to $v_2$ along $P_{2,2}$, and then from $v_2$ to $s_{i_2+1}$ along $P_{1,1}$. For a path $P$, let $\len(P)$ be the number of edges on the path $P$. We have 
\begin{gather} d(t_1,s_{i_2})+d(t_2,s_{i_2+1})\le \len(P_{1,2})+\len(P_{2,1})=\len(P_{1,1})+\len(P_{2,2})=d(t_1,s_{i_2+1})+d(t_2,s_{i_2})   .\label{eq:sum2}\end{gather}
However, summing up the inequality $d(t_1,s_{i_2})> d(t_1,s_{i_2+1})+\De_2$ from (\ref{eq:q1}) and the inequality $d(t_2,s_{i_2+1})+\De_2\ge d(t_2,s_{i_2})$ from (\ref{eq:q2}) gives
\[ d(t_1,s_{i_2})+d(t_2,s_{i_2+1})+\De_2 > d(t_1,s_{i_2+1})+\De_2+d(t_2,s_{i_2}) ,\]
which contradicts (\ref{eq:sum2}).

Finally, Case~3 is identical to Case~1 with every $i_1$, $v_1$, and $\De_1$ replaced by $i_3$, $v_3$, and $\De_3$, respectively, and Case~4 is identical to Case~2 with every $i_2$, $v_2$, and $\De_2$ replaced by $i_4$, $v_4$, and $\De_4$, respectively.
\EP

Thus, Sauer's lemma implies the following corollary, which concludes the proof of \Cref{thm:exact-compr}.
\BC\label{lem:F-bound}
For the family $\m F(S)$ defined in (\ref{eq:F}), we have $|\m F(S)|=O(|S|^3)$.
\EC

\subsection{$(1+\epsilon)$ Compression for Weighted Graphs}\label{sec:weighted-compr}
This section focuses on proving the compression part of \Cref{thm:approx}. The setting is the same as that in \Cref{sec:exact-unweighted}, except that $G$ is now a weighted graph.

\begin{theorem}\label{thm:approx-compr}
Given an $n$-vertex weighted planar graph $G=(V,E)$, and a set $S \subseteq V$ for sources lying on a single face, there exists a $(1+\e)$-compression of all $S \times V$ distances in $G$ using $\widetilde{O}(\poly(|S|/\e)+n)$ bits.
\end{theorem}

\subsubsection{Additive Error Case}

We first prove the result below for \emph{additive} error, based on the \emph{weighted diameter} of the graph. We then reduce the multiplicative error case to the additive via the use of low-diameter decompositions. We remark that we do not need the full power of \Cref{thm:approx-compr} in our diameter application; rather, the additive error result suffices.

\BL\label{lem:approx-additive}
Given an $n$-vertex weighted planar graph $G=(V,E)$ with weighted diameter $d$, and a set $S \subseteq V$ for sources lying on a single face, there exists an $(\e d)$-additive compression of all $S \times V$ distances in $G$ using $\widetilde{O}(\poly(|S|/\e)+n)$ bits.
\EL
To prove this lemma, we will define the notion of \emph{(additive) close} and \emph{(additive) core-set}. 

\BD[$\de$-Additive Close]
Let $G=(V,E)$ be a graph, let $S=\{s_1,\lds,s_\el\}\s V$ be a subset of vertices. Two vertices $u,v\in V$ are $\de$-additive close with respect to $S$ if
\[ |d(u,s_i) - d(v,s_i)| \le \de \qquad \forall i\in[\el] .\]
\ED

\BD[Additive Core-Set]
Let $G=(V,E)$ be a graph, let $S=\{s_1,\lds,s_\el\}\s V$ be a subset of vertices, and let $\de\ge0$ be an additive error parameter. A subset $V'\s V$ is a \emph{$\de$-additive core-set} with respect to $S$ if for all vertices $v\in V$, there exists a vertex $v'\in V'$ that is $\de$-additive close to $v$ with respect to $S$.
\ED

Our goal will be to prove the existence of a core-set of size $\poly(|S|/\e)$, from which \Cref{lem:approx-additive} immediately follows by the definition above.

\subsubsection{Proof of Small Additive Core-Sets}\label{sec:small-additive}

The statement is given below, whose proof follows by a similar VC dimension argument as in \Cref{thm:vc-dim} for the unweighted case. The main difference is that our sets $A_i^\De$ are more sophisticated.
\BT\label{thm:approx-core-set}
Let $G=(V,E,w)$ be a weighted graph, and let $d>0$ be a parameter. Let $S=(s_1,\lds,s_\el)$ be a sequence of points on a common face arranged in cyclic order, such that the distance between any two consecutive points is at most $d$. Then, there exists a $\de$-additive core-set of size $O(|S|^6(d/\de)^4)$.
\ET

Let $\de':=\Th(\de)$ be a parameter a constant factor smaller than $\de$, whose precise value is to be determined. To prove \Cref{thm:approx-core-set}, we define sets $A_i^\De$ similarly to the ones in the unweighted case, but with more values of $\De$. Define
\[ M:=\left\{ -\left\lceil \f{\el d}{\de'} \right\rceil\f{\de' }\el,\ -\lp\left\lceil \f{\el d}{\de'} \right\rceil-1\rp\f{\de' }\el,\ \lds,\ 0,\ \f{\de' }\el,\ 2\f{\de' }\el,\ 3\f{\de' }\el,\ \lds,\ \left\lceil \f{\el d}{\de'} \right\rceil\f{\de' }\el \right\} ,\]
that is, all multiples of $\de' /\el$ from roughly $-d$ to $d$. Note that $|M|=O(\el d/\de')=O(\el d/\de)$.
 Following the statement of \Cref{thm:vc-dim}, we define (1) $A_i^\De:=\{v\in V:d(v,s_i) \le d(v,s_{i+1})+\De\}$ for each $i\in[\el-1]$ and $\De\in M$, (2) the universe $X:=[\el]\times M$ of size $O(\el^2d/\de)$, and (3) the family
\begin{gather}
\m F(S):=\{ \{ (i,\De):v\in A_i^\De \} : v\in V\} .\label{eq:F-2}
\end{gather}

\BL\label{lem:same-sets-S}
Fix $i\in[\el]$. If $u,v\in V$ are in the same sets $A_i^\De$ for all $\De\in M$, that is, $u\in A_i^\De\iff v\in A_i^\De$, then
\[ |(d(u,s_i)-d(u,s_{i+1}))-(d(v,s_i)-d(v,s_{i+1}))| \le \f{\de' d}\el .\]
\EL
\BP
Since $d(s_i,s_{i+1}) \le d$ by assumption, by the triangle inequality, we have
\[ d(u,s_i) - d(u,s_{i+1}) \ge -d \qquad\text{and}\qquad d(u,s_i) - d(u,s_{i+1})\le d.\]

Let $\De_{\min}:=-\left\lceil \f{\el d}{\de'} \right\rceil\f{\de'}\el$ be the smallest value in $W$, and let $\De_{\max}:=\left\lceil \f{\el d}{\de'} \right\rceil\f{\de'}\el$ be the largest value in $W$. By knowing whether $u\in A_i^\De$ for each $\De\in W$, we know that:
\BE
\im If $u\in A_i^{\De_{\min}}$, then $d(u,s_i)-d(u,s_{i+1}) \le \De_{\min}$, so $d(u,s_i)-d(u,s_{i+1}) \in [-d,\De_{\min}]$, which is a (possibly empty) interval of length at most $\de'/\el$.
\im If $u\notin A_i^{\De_{\max}}$, then $d(u,s_i)-d(u,s_{i+1})>\De_{max}$, so $d(u,s_i)-d(u,s_{i+1})\in[\De_{\max},d]$, which is a (possibly empty) interval of length at most $\de'/\el$.
\im Otherwise, $d(u,s_i)-d(u,s_{i+1})$ is contained within an interval of length $\de'/\el$ determined by which sets $A_i^\De$ contain $u$.
\EE
In particular, we can narrow down the value of $d(u,s_i)-d(u,s_{i+1})$ to an interval of length at most $\de'/\el$ solely based on which sets $A_i^\De$ contain $u$. Since $u$ and $v$ belong to the same sets $A_i^\De$, we can repeat the argument for $v$ and obtain the same interval. Thus, $u$ and $v$ lie inside a common interval of length at most $\de'/\el$, as desired.
\EP

Below, we state a generalized version of \Cref{thm:vc-dim}. The proof is essentially unchanged, since the proof of \Cref{thm:vc-dim} can be generalized to work on any set of $\De$ values. We omit the details.

\begin{theorem}[VC-Dimension, Generalized]\label{thm:vc-dim-2}
The set system $\m F$ has VC dimension at most $3$.
\end{theorem}

It follows by \Cref{thm:Sauer} that $|\m F| = O(|X|^3) = O(\el^6(d/\de)^3)$. We now construct our core-set as follows: for each set $F\in\m F$ and each $k\in\{0,1,2,\lds,\lf d/\de'\rf\}$, if there exists a vertex $v\in V$ satisfying
\[k\de' \le d(v,s_1) < (k+1)\de',\]
then add an arbitrary one to the core-set, and call it $v_F^k$. Call the resulting core-set $V'$, that is,
\[ V':=\{ v_F^k : F\in\m F, \, k\in\{0,1,2,\lds,\lf d/\de'\rf\} .\]
Observe that $|V'|=O(\el^6(d/\de)^4)$, meeting the desired bound from \Cref{thm:approx-core-set}. For the rest of this subsection, we focus on proving that $V'$ is a $\de$-additive core-set.

\BL\label{lem:additive-core-set-uv}
For any set $F\in\m F$ and $k\in\{0,1,2,\lds,\lf d/\de'\rf\}$, every two vertices $u,v\in E_F$ satisfying $d(u,s_1),d(v,s_1) \in [k\de',(k+1)\de')$ are $2\de'$-additive close w.r.t.\ $S$.
\EL

\BP
We will prove the following statement by induction from $i=1$ to $\el$:
\[ d(u,s_i) - d(v,s_i) \le \de' + (i-1)\f{\de' d}\el,\]
which clearly implies the lemma. The statement is true for $i=1$ by assumption. For the inductive step, assume the statement for $i$. By \Cref{lem:same-sets-S}, we have
\[ d(u,s_{i+1})-d(v,s_{i+1}) \le \lp d(u,s_i)-d(v,s_i)\rp +\f{\de' d}\el \le \lp\de'+(i-1)\f{\de' d}\el\rp + \f{\de' d}\el = \de' + i\f{\de' d}\el ,\]
completing the induction.
\EP

\BL\label{lem:additive-core-set-F-k}
For any vertex $u\in V$, consider the set $F\in\m F$ containing $u$ and the integer $k\in\{0,1,2,\lds,\lf d/\de'\rf\}$ satisfying $d(u,s_1) \in [k\de',(k+1)\de')$. Then, $u$ and $v_F^k$ are $2\de'$-additive close w.r.t.\  $S$.
\EL

\BP
Apply \Cref{lem:additive-core-set-uv} to set $F$, integer $k$, and vertices $u$ and $v_F^k$.
\EP

From \Cref{lem:additive-core-set-F-k}, it is easy to see that $V'$ is a $2\de'$-additive core-set w.r.t.\ $S$. Finally, setting $\de':=\de/2$ concludes \Cref{thm:approx-core-set}.

\subsubsection{Reduction to Additive Error}\label{sec:reduction-to-additive-error}

In this section, we prove \Cref{thm:approx-compr} using the additive error result of \Cref{lem:approx-additive}. We remark again that \Cref{thm:approx-compr}, and consequently this section, is not required for the results on approximate planar diameter. Instead, \Cref{thm:approx-compr} is for the sake of completing the picture of metric compression. 

The main idea is to run \emph{low-diameter decompositions} at varying values of diameter and compute additive approximations. We define the (strong diameter) version of (weighted) low-diameter decompositions below based on \cite{calinescu2005approximation,HaeuplerL18}.
\BD[Low-diameter decomposition (LDD)]
Given a weighted graph $G=(V,E,w)$, a \emph{low-diameter decomposition} with parameter $\be$ is a randomized partition of $V$ into vertex components  $V_1,V_2,\lds$ such that:
\BE
\im W.h.p., for each  component $V_i$, $G[V_i]$ has weighted diameter $O(\log n/\be)$.
\im For any two vertices $u,v\in V$, they belong to different components with probability at most $\be\, d_G(u,v)$.
\EE
\ED

The algorithm proceeds as follows. Suppose we scale the graph so that all edge weights are in the range $[1,W]$. For each $x$ a power of $2$ in the range $[1,\Th(W)]$, run LDD with parameter $\be=1/x$ for $O(\logn)$ times. For each of the $O(\logn\log W)$ LDDs, for each component $V_i$ in that LDD, we call \Cref{lem:approx-additive} on $G[V_i]$ with sources $S\cap V_i$ and additive error parameter $\e':=\Th(\e/\logn)$ to obtain a compression of size $\tO(\poly(|S\cap V_i|/\e)+|V_i|)$. Observe that if $S$ lies on a common face in $G$, then $S\cap V_i$ lies on a common face in $G[V_i]$, so calling \Cref{lem:approx-additive} is safe.

The total size of all compressions is clearly $\tO((\poly(|S|/\e)+n)\log W)$. Moreover, we remember the sets $V_i$ in  each LDD, which takes $\tO(n\log W)$ space. For two given vertices $v\in V$ and $s\in S$, to recover a $(1+\e)$-approximation of $d_G(v,s)$ from the compression, we look at all LDDs for which $v$ and $s$ belong to the same component. For each such LDD, we have an estimate of $d_G(v,s)$ computed in that compression, which is additively off by at most $\e'$ times the diameter of that component. We add $\e '\cd O(\log n/\be)$ to that estimate, where $\be$ is the parameter used in that LDD, and the $O(\cd)$ comes from property (1) of LDDs. This ensures that w.h.p., this final estimate is always an overestimate of the true distance $d_G(v,s)$. Finally, we output the smallest final estimate over all LDDs (for which $v$ and $s$ belong to the same component).

We now show that this output is a $(1+\e)$-approximation of $d_G(v,s)$, w.h.p. Since all distances are overestimates w.h.p., and since we take the minimum, it suffices to find one LDD for which the distance is at most $(1+\e)$ factor off. Let $x$ be the smallest power of $2$ larger than $2d_G(v,s)$. By property (2) of LDDs, the probability that each LDD executed with $\be=1/x$ separates $v$ and $s$ with probability at most $\be\, d_G(v,s) = d_G(v,s)/x \le 1/2$. Therefore, w.h.p., one of the LDDs over the $O(\log n)$ iterations groups $v$ and $s$ into the same component. W.h.p., this component has diameter at most $O(\logn/\be)$, so  our additive error is at most $\e' \cd O(\logn/\be)$ from \Cref{lem:approx-additive}, plus the extra $\e' \cd O(\logn/\be)$. In total, this is at most $2\e'\cd O(\logn/\be)$, which is at most $2\e' \cd O(\logn\cd d_G(v,s)) \le \e$,  using that $\be=\Th(1/d_G(v,s))$ and $\e '= \Th(\e/\logn)$.

\paragraph{Implications to Centralized Computation of Diameter and Oracles.}
With our improved bound on the core-set size, we can immediately replace the exponential dependency on $1/\epsilon$ in two previous algorithms, with a \emph{polynomial} dependency. The first is the $(1+\epsilon)$-approximate diameter algorithm of Weimann and Yuster \cite{weimann2016approximating} that takes time 
$\widetilde{O}(n \cdot 2^{1/\epsilon})$. The second is an $(1+\epsilon)$-approximate distance oracle of  Gu and Xu \cite{gu2015constant} that has a space of $\widetilde{O}(n \cdot 2^{1/\epsilon})$. These exponential terms in $1/\epsilon$ come for the same reason: this is the state-of-the-art bound on the size of a core-set with respect to subset of $1/\epsilon$ vertices lying on a face. Therefore by plugging our improved bound on the core-set we get a centerlized algorithm for computing the diameter in time $\widetilde{O}(n \cdot \poly(1/\epsilon))$, and $(1+\epsilon)$-approximate distance oracle of space $\widetilde{O}(n \cdot \poly(1/\epsilon))$.

We note that same two results are already presented by Chan \& Skrepetos \cite{chan2017faster} but there are two main differences in our approaches. \cite{chan2017faster} improves \cite{weimann2016approximating} and \cite{gu2015constant} by combining it with Cabello’s recent abstract Voronoi diagram based technique \cite{cabello2017subquadratic}. Indeed this is a nice indication where the tools for the exact setting and the weighted setting can be nicely combined together. Since Cabello’s algorithm is randomized, their final algorithms are also randomized\footnote{It is very possible that those algorithms can be made deterministic by replace the algorithm of \cite{cabello2017subquadratic} with the recent deterministic algorithm by \cite{gawrychowski2018voronoi}.}. In contrast, we use \cite{weimann2016approximating} and \cite{gu2015constant} in a black-box manner. Simply replacing the old exponential bound on the core-set with a polynomial one, the resulting algorithms are also deterministic.

\subsection{Fast Computation of Metric Compression}\label{sec:fast}
In this section, we turn to consider the computational aspects of \Cref{thm:exact,thm:approx}. These, along with the proofs of the compression parts (\Cref{thm:exact-compr,thm:approx-compr}), complete \Cref{thm:exact,thm:approx}. Similarly to \Cref{sec:compr}, we begin with the simpler, unweighted case and then present the weighted case.

\subsubsection{Computation of Unweighted Compression}\label{sec:fastunweighted}
We start by considering the fast computation of the exact compression scheme for $D$-diameter unweighted planar graphs. For simplicity, we consider the case where $T=V$, but same algorithm works for the case where $T \subseteq V$, in the latter case, we only keep the distance label tuple ID for the terminal in $T$.
\BL\label{lem:unweighted-faster}
Given an $n$-vertex unweighted planar graph $G=(V,E)$ and a set $S \subseteq V$ for sources lying consecutively on a single face, we can compute the $O(D \cdot |S|^3)$ distinct distance tuples over vertices in $V$, as well as which vertices contain each distance tuple, in time $\tO(n+D\cd|S|^3)$.
\EL

At a high level, our goal is to follow the multiple-source shortest path (MSSP) algorithm with all sources lying on a common face, from \cite{klein2005multiple}, while maintaining \emph{hashes} of distance tuples. Observe that we cannot explicitly maintain the size-$\el$ distance tuple for each vertex, since that is $n\cd\el$ integers, which exceeds the promised time bound of $\tO(n+D\cd\el^4)$ if $\el$ is large (say, $n^{\Om(1)}$). Therefore, we maintain hashes of size $O(\log n)$ instead. The tricky part is to efficiently update the distance tuples while running the MSSP algorithm.

We define a hash function motivated by Rabin-Karp string hashing. For a distance tuple $\bd=(d_1,\lds,d_\el)$, we define its \emph{hash value} under \emph{base} $b$ and \emph{modulus} $p$ (for $p$ prime) as
\[ h(\bd,b,p) := \sum_{i=1}^\el d_i b^i \bmod p .\]

Clearly, if two distance tuples are equal, then their hash values under the same base and modulus are equal. The next claim shows that for two distinct distance tuples, their hash values are likely to be different under a \emph{random} base, as long as the modulus is large enough.

\BCL
Consider two distinct distance tuples $\bd_1\ne\bd_2$. Then, for a fixed prime $p$,
\[ \Pr_{b \in \{0,1,\lds,p-1\}} [ h(\bd_1,b,p) = h(\bd_2,b,p) ] \le \el/p .\]
\ECL
\BP
We have $h(\bd_1,b,p)=h(\bd_2,b,p)$ if and only if $\sum_{i-1}^\el (d_{1,i}-d_{2,i})b^i \equiv 0 \bmod p$. The polynomial on the LHS is not identically zero since $\bd_1\ne\bd_2$. Since its degree is at most $\el$, it has at most $\el$ zeroes in $\{0,1,\lds,p-1\}$. Therefore, the probability that $b$ is one of these $\le \el$ zeroes is at most $\el/p$.
\EP

The specifics of the MSSP algorithm are irrelevant. The important properties are as follows.

Let us define the following dynamic tree operations, following \cite{Course-PlanarGraph-MIT}.
The MSSP algorithm begins with computing the SSSP tree on source $s_1$. It then travels along the face segment in the order $(s_2,s_3,\lds,s_\el)$, temporarily setting each $s_j$ as the source, while maintaining a \emph{dynamic forest} $F$ of values, one for each vertex in $V$. (It also maintains a dynamic forest on the \emph{dual graph}, but we do not need to discuss that here.) The algorithm has the following two invariants:

\BE
\im[(I1)] After finishing source $s_j$, for any vertex $v\in V$, its current value in $F$ is precisely $d(s_j,v)$.
\im[(I2)] It performs $O(n)$ dynamic forest operations in total, over all sources.
\EE

The dynamic forest $F$ has the following operations:

\BE
\im $\Cut(e)$: remove edge $e$ from forest
\im $\Join(e)$: add edge $e$ to join two trees
\im $\Get(v)$: returns the value currently stored at $v$
\im $\Add(\De,v)$: increase all values in the subtree rooted at $x$ by $\De$
\EE

These can be supported using Euler-Tour trees \cite{tarjan1984finding}. Our goal is to maintain another dynamic tree $F'$ on the vertices in $V$ that is updated alongside the MSSP algorithm, so that at the end, the value at each vertex $v$ is precisely the hash $h(\tuple_S(v),b,p)$. 
(From now on, we assume that the algorithm has fixed $b$ and $p$.)

To begin, define
\[ h(\tuple_S(v), b, p, j) := \lp \sum_{i=1}^{j-1}d(v,s_i)b^i + d(v,s_j)  \sum_{i=j}^\el b^i \rp \bmod p .\]
In particular, $h(\tuple_S(v),b,p,j)$ is exactly $h(\tuple_S(v),b,p)$, except all $d(v,s_i)$ terms in the expression for $i>j$ are replaced with $d(v,s_j)$.

Throughout the MSSP algorithm, we will maintain the following invariant:

\BE
\im[(I')] After the MSSP algorithm finishes source $s_j$, for any vertex $v\in V$, its current value in $F'$ is precisely $h(\tuple_S(v), b, p, j)$.
\EE
Note that, if Invariant~(I') holds at the end of the algorithm, then the value of each vertex $v\in V$ in $F$ is precisely $h(\tuple_S(v),b,p,\el)=h(\tuple_S(v),b,p)$, our desired hash value.

We maintain Invariant~(I') as follows. After the MSSP algorithm computes the SSSP for source $s_1$, initialize the dynamic forest $F'$ with value $h(\tuple_S(v),b,p,1)$ on each vertex $v\in V$.
Then, as the MSSP algorithm goes through sources $s_2$ through $s_\el$, we will modify $F'$ alongside $F$. In particular, whenever an operation is performed on $F$ under source $s_j$, we perform a similar operation on $F'$ as follows:
\BE
\im $\Cut(e)$ on $F$: Perform the same on $F'$
\im $\Join(e)$ on $F$: Perform the same on $F'$
\im $\Get(v)$: Do nothing on $F'$, since the MSSP algorithm does not use $F'$
\im $\Add(\De,v)$ on $F$: This is the only nontrivial case. We perform $\Add(\De',v)$ on $F'$, where
\[ \De' := \lp \De \cd \sum_{i=j}^\el b^i \rp \bmod p .\]
Of course, the value $\sum_{i=j}^\el b^i\bmod p$ can be precomputed at the beginning of each source $s_j$, so the update only takes $O(\log n)$ time.
\EE

\BCL
Invariant~(I') is satisfied throughout the MSSP algorithm.
\ECL
\BP
To do so, we maintain a more \emph{fine-grained} invariant, which we call Invariant~(I''): If the MSSP is at any point in processing source $s_j$, then for each vertex $v\in V$,
\[ \Value(F',v) = \lp \sum_{i=1}^{j-1}d(v,s_i)b^i + \Value(F,v) \sum_{i=j}^\el b^i \rp \bmod p ,\]
where $\Value(F,v)$ is the current value of $v$ in $F$ (that is, the integer returned by $\Get(v)$ on $F$), and similarly for $\Value(F',v)$ for $F'$. Clearly, Invariant~(I'') implies Invariant~(I').

First, we show that Invariant~(I'') is satisfied right before the MSSP algorithm starts with the source $s_2$. Since $\Value(F,v)=d(v,s_1)$ by the correctness of the MSSP algorithm, the value of $v$ at $F'$ is initialized to
\[ h(\tuple_S(v),b,p,1)=  \lp \sum_{i=1}^{j-1}d(v,s_i)b^i + \Value (F,v) \sum_{i=j}^\el b^i \rp \bmod p ,\]
satisfying the invariant.

There are two types of transitions of the MSSP algorithm. First, the algorithm can transition from source $s_j$ to $s_{j+1}$. By the correctness of the MSSP algorithm, we must have $\Value(F,v)=d(v,s_j)$, so
\begin{align*}
\Value(F',v) &= \lp \sum_{i=1}^{j-1}d(v,s_i)b^i + \Value(F,v) \sum_{i=j}^\el b^i \rp \bmod p
\\&= \lp \sum_{i=1}^{j}d(v,s_i)b^i + \Value(F,v) \sum_{i=j+1}^\el b^i \rp \bmod p ,
\end{align*}
as desired. Second, the algorithm can call $\Add(\De,v)$ on $F$, which increases $\Value(F,v)$ by $\De$. In this case, $\Value(F',v)$ increases by $ ( \De \cd \sum_{i=j}^\el b^i ) \bmod p$, which exactly restores the invariant.
\EP

We now complete our $\tO(n+D\cd\el^4)$ time algorithm. Fix a prime $p$ close to $n^c$ for some constant $c$ and a uniformly random $b\in\{0,1,\lds,p-1\}$ throughout the algorithm. After running the MSSP algorithm and updating $F'$ alongside it, we call $\Get(v)$ on each $v\in V$ in $F'$ to determine all values $h(\tuple_S(v),b,p)$. With probability at least $1-\el/n^c$, any given pair of vertices with different distance tuples also have different hashes. By a union bound, the probability that this occurs over all pairs is at least $1-\el n^2/n^c \ge 1-1/n^{c-3}$, which is sufficient for our w.h.p.\ bound.

The algorithm groups the vertices by their hash value, so that w.h.p., all vertices in the same group have the same distance tuple. By \Cref{thm:exact-compr}, there are at most $O(D\cd \el^3)$ distance tuples. The algorithm selects one vertex from each group, producing $O(D\cd\el^3)$ vertices. Next, it computes distance tuples for these selected vertices as follows: run the MSSP algorithm again, but after finishing each source $s_i$, call $\Get(v)$ in $F$ for each selected vertex $v$. This produces the desired $O(D\cd\el^3)$ distance tuples, one for each group. There are $O(D\cd\el^3)$ calls for each $s_i$, or $O(D\cd\el^3)\cd\el=O(D\cd\el^4)$ calls overall, so this part takes $\tO(n+D \cd \el^4)$ time. The rest of the algorithm takes $\tO(n)$ time, proving the running time promised by \Cref{lem:unweighted-faster}.

\subsubsection{Computation of Weighted Case}\label{sec:fastweighted}

Similarly to \Cref{sec:weighted-compr}, we reduce to and then solve the additive error case. Note that the LDD approach in \Cref{sec:reduction-to-additive-error} suffices for the reduction here as well. In particular, we can compute the additive case $O(\log n\log W)$ many times, one on each LDD, giving us an $O(\logn\log W)$ overhead, which is good enough for \Cref{thm:approx}. Therefore, for the rest of this section, we restrict to the additive error case.

\BT\label{lem:weighted-faster}
Let $G$ be a weighted planar graph with weighted diameter $d$, and let $S:=(s_1,s_2,\lds,s_\el)$ be consecutive vertices on a face, ordered in clockwise or counter-clockwise order. We can compute a $\de$-additive core-set $V'$ w.r.t.\ $S$ of size $\poly((d/\de)\log n)$, as well as their distance tuples, in time $\tO(n + \poly(\el d/\de))$. Moreover, for each $v\in V$, we can locate a vertex in the core-set that is $\de$-additive close to $v$.
\ET

The algorithm is similar to the one in \Cref{lem:unweighted-faster}, save for one crucial difference. Here, since we are looking for approximate distances, what we need is not a hash function, but a ``clustering'' function that groups together vertices whose distance tuples are close together (say, in $\el_2$-distance). For this, a simple Johnson-Lindenstrauss (JL) dimension reduction does the trick.

Let $r:=O(\logn)$ and let $\Phi:\R^\el\to\R^r$ be a random linear projection function satisfying the following w.h.p.:
\begin{gather}
\f12 \norm{u-v}_2 \le \norm{\Phi(u)-\Phi(v)}_2 \le 2\norm{u-v}_2 \qquad\forall u,v\in V. \label{eq:JL}
\end{gather}

Let $\de':=\de/\poly(\el \logn)$ be a parameter whose precise value is to be determined. By \Cref{thm:approx-core-set}, there is a $\poly(\el d/\de')$-sized $\de'$-additive core-set w.r.t.\ $S$. This means that if we consider the distance tuples $\tuple_S(v)$ as vectors in $\R^\el$, then for every $v\in V$, we have
\[\norm{\tuple_S(v)-\tuple_S(v')}_\infty \le \de' \implies \norm{\tuple_S(v)-\tuple_S(v')}_2 \le \de'\sr\el \]
for some vertex $v'$ in the core-set $V'$. In particular, this means that the vectors $\tuple_S(v)$ can be covered with $\poly(\el d/\de')$ $\el_2$-balls of radius $\de'\sr\el$, one centered at $\tuple_S(v')$ for each $v'\in V'$. Therefore, we know by (\ref{eq:JL}) that w.h.p., the projected tuples $\Phi(\tuple_S(v))$ can be covered with $\poly(\el d/\de')$ $\el_2$-balls of radius $2\de'\sr\el$, one centered at $\Phi(\tuple_S(v'))$ for each $v'\in V$. We thus have the following claim: 

\BCL\label{clm:small-cover}
W.h.p., the projected tuples $\Phi(\tuple_S(v))$ for $v\in V$ can be covered with $\poly(\el d/\de')$ $\el_2$-balls of radius $2\de'\sr\el$.
\ECL

Given the intuition above, the algorithm has two natural steps. First, it seeks to compute the projected distance tuple $\Phi(\tuple_S(v))$ for each $v\in V$. Then, based on the computed values, it needs to select a subset of vertices to form the $\de$-additive core-set $V'$.

The first step of the algorithm is almost identical to the one in \Cref{lem:unweighted-faster}, with the hash function $h$ replaced by the projection function $\Phi$.
Following the algorithm of \Cref{lem:unweighted-faster}, define
\begin{align*}
\Phi(\tuple_P(v),j) :=&\ \Phi( d(v,s_1), d(v,s_2),\lds,d(v,s_{j-1}),d(v,s_j),d(v,s_j),d(v,s_j),\lds,d(v,s_j) )
\\ =&\ \sum_{i=1}^j d(v,s_i)\Phi(e_i) + d(v,s_j)\sum_{i=j}^\el\Phi(e_i),
\end{align*}
where $e_i\in\R^\el$ is the unit vector in dimension $i$. (The second equality follows because $\Phi$ is a linear map.)

Similarly to \Cref{lem:unweighted-faster}, we maintain a dynamic forest $F'$ on the vertices in $V$ that is updated alongside the MSSP algorithm. This time, the values in $F'$ are vectors in $\R^\el$.\footnote{We can also imagine maintaining $r=O(\logn)$ many dynamic forests, each one keeping track of a single coordinate in $\R^r$.} Throughout the MSSP algorithm, we will maintain the following invariant:
\BE
\im[(I')] After the MSSP algorithm finishes source $s_j$, for any vertex $v\in V$, its current value in $F'$ is precisely $\Phi(v,j)$.
\EE
Like the algorithm in \Cref{lem:unweighted-faster}, the only nontrivial updates to $F'$ are when $\Add(\De,v)$ is called on $F$. In this case, we perform $\Add(\De',v)$ on $F'$, where
\[ \De' := \De\cd\sum_{i=j}^\el\Phi(e_i) .\]
Like in \Cref{lem:unweighted-faster}, the value $\sum_{i=j}^\el \Phi(e_i)$ can be precomputed at the beginning of each source $s_j$, so the update only takes $O(r \log n)=O(\log^2n)$ time. Following the analysis of \Cref{lem:unweighted-faster}, the algorithm for the first step takes $O(n\log^2n)$ time.

For the second step, we perform a simple hashing into a \emph{randomly shifted grid}, as follows. Define $Z:=\de'\sr\el\logn$ and let $z\in [0,Z)^\el$ be a vector whose coordinates are sampled independently and uniformly from the interval $[0,Z)$. Initialize a hash table $H$ indexed by $\Z^r$, and for each vertex $v\in V$, add it to the entry of $H$ indexed by
$\lf(\Phi(\tuple_P(v))+z)/Z\rf \in \Z^r$, where $\lf \bx\rf$ for a vector $\bx\in\R^r$ indicates replacing each entry of $\bx$ by its floor. Finally, from each non-empty entry in $H$, we add an arbitrary vertex $v$ in that entry to our core-set $V'$, and for any other vertex $u$ in that entry, we declare that $v\in V'$ is $\de$-additive close to $u$. Clearly, this second step so far takes $O(nr)$ time. Lastly, to compute the distance tuples for each vertex in $V'$, we can repeat the MSSP as in the algorithm of \Cref{lem:unweighted-faster}, taking $\tO(n+\el|V'|)$ time.

We will prove two claims: (1) the number of non-empty entries in $H$ is small, which also bounds the size of the core-set, and (2) all vertices belonging to the same entry in $H$ have (original) distance tuples that are close to each other.

\BL\label{lem:num-entries}
The expected number of non-empty entries in $H$ is $\poly(\el d/\de')$.
\EL
\BP
By \Cref{clm:small-cover}, the vectors $\Phi(\tuple_P(v))$ can be covered by $\poly(\el d/\de')$ $\el_2$-balls of radius $2\de'\sr\el$. Therefore, it suffices to show that the vertices in each $\el_2$-ball are mapped to $O(1)$ many entries in expectation.

Fix a ball $B\s R^r$ of radius $2\de'\sr\el$. 
It is not hard to see that for all coordinates $i\in[r]$, with probability at most
\[ \f{2\de'\sr\el}{Z} = \f2{\logn} ,\]
the vectors $\bx=(x_1,\lds,x_r)\in B$ take two different values of $\lf (x_i+z)/Z\rf$, and otherwise, the vectors take one value of $\lf(x_i+z)/Z\rf$. Moreover, for each $t\in[r]$, if $t$ many coordinates take two different values, then the number of non-empty entries in $H$ is at most $2^t$. Since the coordinates of $z$ are sampled independently, the expected number of entries in $H$ is at most
\[ \lp1+\f2{\logn}\rp^r = \lp1+\f2{\logn}\rp^{O(\log n)} = O(1) ,\]
as desired.
\EP

\BL\label{lem:same-entry-close}
W.h.p., any two vertices $u,v$ that map to the same entry in $H$ are $(2\de'\sr{\el r}\logn)$-additive close.
\EL
\BP
Clearly, any two vertices $u,v$ that map to the same entry in $H$ satisfy 
\[ \norm{\Phi(u)-\Phi(v)}_2 \le \sr r \norm{\Phi(u)-\Phi(v)}_\infty \le \sr rZ = \de'\sr{\el r}\logn .\]

Therefore, w.h.p., for all such vertices $u,v$,
\[
\norm{u-v}_\infty \le \norm{u-v}_2\stackrel{(\ref{eq:JL})}{\le}2\norm{\Phi(u)-\Phi(v)}_2 \le 2\de'\sr{\el r}\logn, \]
so $u$ and $v$ are $(2\de'\sr{\el r}\logn)$-additive close, as desired.
\EP

Although \Cref{lem:num-entries} only holds in expectation, by Markov's inequality, with probability at least $1/2$, the number of entries is at most twice the expectation. We can therefore try the hashing algorithm $O(\log n)$ times, and w.h.p., one of the tries has its number of entries at most twice the expectation, which is $\poly(\el d/\de')$.

Finally, we set $\de':=\de/(2\sr{\el r}\logn)$, so that by \Cref{lem:same-entry-close}, vertices that map to the same entry in $H$ are pairwise $\de$-additive close. 
Recalling that $r=O(\logn)$, we have $\poly(\el d/\de')=\poly(\el d/\de)$ many non-empty entries in $H$, and thus that many points in the core-set $V'$. This concludes \Cref{lem:weighted-faster}.

\newcommand{\new}{}

\section{Distributed Diameter in Unweighted Graphs}\label{sec:unbounded-diam}

\subsection{Bounded Diameter Decomposition (BDD)}\label{sec:bdd}
For a graph $G$, we describe a recursive (balanced) partitioning procedure represented by a tree $\mathcal{T}$, whose vertices, denoted as \emph{bags}, correspond to subgraphs in $G$.
This recursive partitioning must satisfy two crucial properties. First, the diameter of each bag is kept being bounded by $O(D\log n)$ throughout all recursion levels which enables the computation of an $O(D\log n)$-path separator recursively. The second property provides a small overlap between all subgraphs in the same recursion level, by guaranteeing that each edge $e$ belongs to at most two subgraphs in each recursion level. This allows one to work on all subgraphs of the same level simultaneously with the same round complexity (up to constant factor) of working on a single subgraph.

\BD[Bounded Diameter Decomposition]
Let $G=(V,E)$ be a graph under some planar embedding and let $T$ be a BFS tree of $G$ of radius $D$ rooted at some node $r$. A \emph{bounded diameter decomposition (BDD)} of $G$ is a rooted tree $\m T=(V_{\m T}, E_{\m T})$ whose vertices $X\in V_{\m T}$, called \emph{bags}, are subsets of $V$ satisfying:
 \BE
 \im The tree has depth at most $O(\log n)$.
 \im The root bag of $\m T$ equals $V$.
 \im For each non-leaf bag $X$, its child bags $X_1,\lds,X_\el$ satisfy $\bigcup_iG[X_i]=G[X]$.
 \im Every leaf bag $X$ has size $O(D\logn)$.
\im For each bag $X$, $G[X]$ is \new{connected}.
\im For each bag $X$, $T[X]$ consists of at most $O(\logn)$ connected components. 
\im For each non-leaf bag $X$, consider the set $S_X$ of nodes in $X$ present in more than one child bag. There exist some $O(\logn)$ paths inside $T[X]$ whose disjoint union of vertices is exactly $S_X$. (This, together with the fact that $T$\ has diameter $D$, implies that the size of this set is at most $O(D\logn)$.)
\im There exists a mapping $\m O$ from the bags $X$ to connected open subsets of $\R^2$ such that:
  \BE
  \im For each bag $X$, the planar embedding of the vertices and edges of $G[X]$ (which are points and simple curves in $\R^2$) are contained in $\overline{\m O(X)}$, the (topological) closure of $\m O(X)$.
  \im For each depth $d$, the subsets $\m O(X)$ over bags $X$ of depth $d$ are disjoint (in $\R^2$).
  \EE
\im For each bag $X$, there exists a closed curve in $\R^2$ through the nodes in $S_X$ that is disjoint from $\m O(X')$ for any child bag $X'$. 

%
%
%
%
 %
%
 \EE
\ED


\BT[Bounded diameter decomposition for planar graphs]\label{thm:SPSD}
Let $G=(V,E)$ be an unweighted planar graph with diameter $D$. There is a distributed algorithm that computes 
the recursive partitioning  of $G$ represented by a tree $\m T$ of height $O(\log n)$ within $\tO(D)$ rounds. In particular, every bag $X\in V_{\m T}$ has a unique ID and every node knows the IDs of all the bags that contain it. 
\ET

First, a few remarks. The set $S_X$ in property~(7) fits the role of a \emph{separator} in the graph $G[X]$.
Also, properties~(8)~and~(9) of the BDD may look cumbersome, but they necessary for the application of \Cref{thm:exact-compr} in diameter computation. Indeed, they are the only two properties which involve the planarity of the graph. And since \Cref{thm:exact-compr} assumes planarity, it makes sense that some aspects of planarity must be preserved in the BDD. Moreover, we use the topological notion of planarity (as opposed to combinatorial) because we need to track a fixed global embedding in our mind, so it is more geometrically intuitive to deal with regions in the plane.

Below, we present a few properties that are implied by the properties of the BDD that are sufficient for applications such as in \Cref{sec:exact-labels}.

\BL\label{lem:BDD-add}
If $\m T$ is a BDD for graph $G$, then the following additional properties hold:
\BE
\im[3'.] For each non-leaf bag $X$ and any two child bags $X_1$ and $X_2$, any path in $G[X]$ from a vertex in $X_1$ to a vertex in $X_2$ must intersect $S_X$.
\im[6'.] For each bag $X$, $G[X]$ has diameter $O(D\log n)$.
\im[7'.] For each non-leaf bag $X$, $|S_X|=O(D\log n)$.
\im[8'.] For each depth $d$, every edge of $G$ is in at most $2$ subgraphs $G[X]$.
\EE
\EL
\BP
We first prove property (3') from property (3) of BDDs. Let $v_1,\lds,v_\el$ be a path from $v_1\in X_1$ to $v_\el\in X_2$. Let $i\in[\el-1]$ be such that $v_i\in X_1$ and $v_{i+1}\notin X_1$. By property (3), the edge $(v_i,v_{i+1})$ must be in some child bag $X_3$. Then, $v_i$ is in $X_1$ and $X_3$, so by definition of $S_X$ (property~(7)), $v\in S_X$ and we are done. 

Property~(6') follows easily from  property~(6) and the fact that $T$ has diameter $O(D)$. Indeed, connecting together $O(\log n)$ trees of diameter $O(D)$ produces a tree of diameter $O(D\logn)$. Similarly, property~(7') also easily follows from property~(7) and the fact that $T$ has diameter $O(D)$.


We now prove property (8') from property (8) of BDDs. By property (8b), the regions $\m O(X)$ for bags $X$ of depth $d$ are disjoint. Fix an edge $(u,v)\in E$ inside some bag $X$. First, if either $u$ or $v$ is inside $\m O(X)$, then no other bag of depth $d$ can contain $x$ by the disjointness property. Otherwise, $u,v\in\ol{\m O(X)}\setminus\m O(X)$. Similarly, if the embedding of the edge $(u,v)$ intersects $\m O(X)$, then only $X$ can contain the edge; otherwise, the edge is a curve of positive length along the boundary $u,v\in\ol{\m O(X)}\setminus\m O(X)$. By simple topological arguments, since the regions $\m O(X')$ for bags $X'$ of depth $d$ are disjoint, any curve can be inside $\ol{\m O(X')}\setminus\m O(X')$ for at most two $X'$.
\EP

\subsection{Distributed Computation of BDD Decomposition}\label{sec:4.2}

The algorithm proceeds top-down, taking $O(\log n)$ iterations. On each iteration, every leaf bag $X$ with than $\Om(D \logn)$ vertices computes children $X_1,\lds,X_\el$ satisfying condition (3) of the BDD, as well as $|X_i|\le(5/6)|X|$ for each $i$. This is done in parallel over all bags that are leaves at the beginning of that iteration. It is easy to see that any bag $X$ of depth $i$ satisfies $|X|\le(5/6)^in$. Therefore, only $O(\log n)$ iterations are needed before every leaf bag $X$ has $|X|\le O(D)$. \new{We accomplish this by computing \emph{balanced cycle separators} using a modification of the algorithm from \cite{GhaffariP17}. Their algorithm requires that the planar graph network is \emph{biconnected}, so the technical modification is to make it work for graphs that are not biconnected. We defer the details to Appendix \ref{sec:balanced-cycle-general}; for this section, assume the following theorem as a black box:}

\BT\label{thm:cycle}
Given a $D$-diameter graph $G$ and a spanning tree $T$ (possibly of large diameter), there exists an $\widetilde{O}(D)$-round algorithm that computes a balanced cycle separator that consists of two tree paths of $T$ plus one additional edge (which is not necessarily in $G$).
\ET

We first focus on computing the children $X_1,\lds,X_\el$ for a single leaf bag $X$, deferring the parallelization over all leaf bags $X$ to the next part. Fix a leaf bag $X$, and suppose that every node in $X$ knows that it is in $X$. We can easily compute all edges in $T[X]$: it is all edges in $T$ that have both endpoints in $X$.

\paragraph{Step 1: Connect the trees in $T[X]$.}
By property (5), $G[X]$ is connected, so adding some subset of these edges produces a spanning tree of $G[X]$ which we call $T'$. By property (6), $T' \setminus T$ has $O(\log n)$ edges, which means that $T'$ has diameter $O(D\logn)$.




\paragraph{Step 2: Compute a cycle separator.}
We apply \Cref{thm:cycle} on the graph $G[X]$ with spanning tree $T'$, which computes a cycle separator $S$ with all edges inside $T'$ except possibly one, called the \emph{virtual} edge, that may not even be in $E$, the edges of the original graph $G$. We also have $|S|=O(D\logn)$ since $T'$ has diameter $O(D\logn)$. From now on, we view $S$ as a set of edges, but we say ``vertex $v$ is on $S$'' if $v$ is incident to some edge in $S$.

\paragraph{Defining the child bag $X^+$.}
Let $X^+$ denote the vertices enclosed by the cycle $S$ in the planar embedding. (We include the vertices in $S$ itself in $X^+$.) It is easy to see by planarity that $G[X^+]$ is connected.  Also, since $S$ is a cycle separator, we have $|X^+|\le(5/6)|X|$. Finally, for the mapping $\m O$ needed for property~(8), we define $\m O(X^+)$ to be the (topological) interior of the cycle $S$ in the planar embedding. That is, in this case, the topological closure $\ol{\m O(X^+)}$ is simply $S$. Note, that there is no need to actually compute $\m O$.

\BCL \label{clm:conn}
$T'[X^+]$ is connected.
\ECL
\BP
Root the tree $T'$ at a node $r$ on $S$. Then, the vertices on $S$ all lie on two paths that start from $r$ and travel down the tree. For any vertex $v$ in $X^+$, walk up the rooted tree until a node on $S$ is reached. Since $S$ separates $X^+$ from the rest of the graph, every vertex visited so far must lie in $X^+$. Finally, since $S$ consists of two paths from $r$, we can walk up one of the two paths and reach $r$. We have thus connected an arbitrary vertex $v\in X^+$ with $r$ along a path in $T'[X^+]$. It follows that $T'[X^+]$ is connected.
\EP

\begin{figure}[H]
\centering
\begin{tikzpicture}[scale=.85]

\draw[black, line width=2pt, fill=red!70!black] (-0.3784,-5.0955) node (v1) {} arc (-94.2471:-440:5.1095);
\node (v2) at (0.8919,-5.0374) {};

\draw[fill=red!86!white] (0,0) -- ++(-120:5.0542) arc (-120.4357:-60.1335:5) -- cycle;

\draw [dashed, line width=2pt] (v1.center) edge (v2.center);

\draw [red!60!white,fill=red!60!white] plot[smooth, tension=.7] coordinates {(-1.8186,-4.7397) (-1.3717,-3.0783) (-0.2815,-2.7681)} -- plot[smooth, tension=.7] coordinates { (0.1221,-2.8353) (0.2424,-3.3096) (-0.1806,-3.7338) (0.1221,-4.2046) (-0.2527,-4.5986) (-0.0124,-5.0598)}
--
(-0.3873,-5.0841) arc (-94.3563:-111.0752:5.0988);

\draw (-4.5032,-2.4056) arc (-151.889:-111.2584:5.1055) ;
\draw (2.2588,-4.5244) ;
\draw (2.2701,-4.5356) arc (-63.4117:-28.5843:5.072);
\draw [blue, line width=2pt, fill=green!50!white]
plot[smooth, tension=.7]coordinates {(1,-3) (1.7629,-3.6228) (2.2701,-4.5356)}
--
(2.2701,-4.5356) arc (-63.4117:-28.5843:5.072)
--
plot[smooth, tension=.7] coordinates {(4.4339,-2.4394) (3.3182,-0.7376) (2.7772,1.8207) (3.4872,3.7027)}
--
(3.4647,3.7253) arc (47.0758:157.7143:5.0874)
--
plot[smooth, tension=.7] coordinates {(-4.7286,1.9334) (-2.9367,0.4908) (-3.0719,-1.4364) (-4.492,-2.4056)}
--
(-4.5032,-2.4056) arc (-151.889:-111.2584:5.1055) 
--
  plot[smooth, tension=.7] coordinates {(-1.821,-4.7385) (-1.3702,-3.0931) (-0.2882,-2.7663)}
;
\draw  plot[smooth, tension=.7] coordinates {(-1.821,-4.7385) (-1.3702,-3.0931) (-0.2882,-2.7663)};
\draw  plot[smooth, tension=.7] coordinates {(1,-3) (1.7629,-3.6228) (2.2701,-4.5356)};
\draw  plot[smooth, tension=.7] coordinates {(4.4339,-2.4394) (3.3182,-0.7376) (2.7772,1.8207) (3.4872,3.7027)};

\draw[blue,line width=2pt, dashed] (-0.2815,-2.7681) node (v3) {} -- (1.0078,-3.0002);
\draw[line width=1.5pt]  plot[smooth, tension=.7] coordinates {(v3) (-0.8742,-3.5187) (-1.2537,-4.9425)};

\draw [line width=1.5pt] plot[smooth, tension=.7] coordinates {(1.4921,-3.3547) (0.5299,-3.8781) (0.9652,-5.0166)};
\draw[line width=1.5pt]  plot[smooth, tension=.7] coordinates {(0.7685,-3.678) (1.4934,-4.0174) (1.8943,-4.7429)};
\draw[line width=1.5pt]  plot[smooth, tension=.7] coordinates {(1.2516,-3.8774) (1.1243,-4.3802) (1.3661,-4.8702)};
\draw [line width=1.5pt] plot[smooth, tension=.7] coordinates {(-0.8788,-3.5086) (-0.4092,-3.9859) (-0.6089,-5.0542)};

\node[scale=1] at (0,0) {$\mathcal O(X^+)$};
\node[scale=1] at (-3.6384,-0.3474) {$\mathcal O(X^-_1)$};

\node[scale=1] at (4.2376,-0.8546) {$\mathcal O(X^-_3)$};
\node[scale=1] at (-1.5142,-5.6742) {$\mathcal O(X^-_{2,1})$};
\node[scale=1] at (1.6357,-5.637) {$\mathcal O(X^-_{2,2})$};
\node[blue,scale=1.5] at (0,4.5) {$S$};

\node (v4) at (-3.6803,0.3376) {};
\node (v5) at (-4.9601,1.0155) {};
\node (v6) at (-5.0591,-0.2294) {};
\node (v7) at (-4,-1.5) {};

\draw  [line width=1.5pt] plot[smooth, tension=.7] coordinates {(-3,0.5) (v4) (v5)};
\draw  [line width=1.5pt] plot[smooth, tension=.7] coordinates {(-3.6803,0.3376) (-4.2791,0.2231) (v6)};

\draw [line width=1.5pt]  plot[smooth, tension=.7] coordinates {(-3.0719,-1.4364)
 (v7) (-4.9284,-1.2739)};
\draw  [line width=1.5pt] plot[smooth, tension=.7] coordinates {(-4.4107,0.1259) (-4.4146,-0.7872) (v7)};

\draw [line width=1.5pt] plot[smooth, tension=.7] coordinates {(2.9095,2.4884) (4.1723,1.8873) (4.633,0.3516) (3.3167,-0.7379)};
\draw [line width=1.5pt] plot[smooth, tension=.7] coordinates {(2.7829,1.624) (3.6166,1.3681) (3.9164,0.4321) (3.0827,-0.109)};

\draw[blue,line width=2pt, dashed] (-0.2815,-2.7681) node (v3) {} -- (1.0078,-3.0002);
\draw [blue, line width=2pt]
plot[smooth, tension=.7]coordinates {(1,-3) (1.7629,-3.6228) (2.2701,-4.5356)}
--
(2.2701,-4.5356) arc (-63.4117:-28.5843:5.072)
--
plot[smooth, tension=.7] coordinates {(4.4339,-2.4394) (3.3182,-0.7376) (2.7772,1.8207) (3.4872,3.7027)}
--
(3.4647,3.7253) arc (47.0758:157.7143:5.0874)
--
plot[smooth, tension=.7] coordinates {(-4.7286,1.9334) (-2.9367,0.4908) (-3.0719,-1.4364) (-4.492,-2.4056)}
--
(-4.5032,-2.4056) arc (-151.889:-111.2584:5.1055) 
--
  plot[smooth, tension=.7] coordinates {(-1.821,-4.7385) (-1.3702,-3.0931) (-0.2882,-2.7663)}
;

\draw   [line width=1.5pt] plot[smooth, tension=.7] coordinates {(4.4324,-2.4229) (4.3334,-1.7119) (3.8474,-1.298)};
\draw   [line width=1.5pt] plot[smooth, tension=.7] coordinates {(4.3424,-1.7299) (4.5044,-1.55) (4.6213,-1.271)};
\end{tikzpicture}

        \caption{The case where $S$ has vertices on $\ol{\m O(X)}\setminus\m O(X)$. The different shades of red all comprise $\m O(X) \setminus \ol{\m O(X^+)}$, which is not connected.  Note that the blue closed curve does not belong in any open region $\m O(\cd)$.   The two dark red regions in $O(X) \setminus \ol{\m O(X^+)}$ become $\m O(X^-_1)$ and $\m O(X^-_3)$. The component $X^-_2$ becomes disconnected after removing the virtual (dotted) edges, so the corresponding region in $O(X) \setminus \ol{\m O(X^+)}$ is divided into the two regions $X^-_{2,1}$ and $X^-_{2,2}$ (the lighter shades of red).
}
\label{fig:bdd} 
\end{figure}

\paragraph{Step 3: Computing the child bag $X^+$.}
From the planar embedding, every node knows the clockwise ordering of its edges in the planar embedding. We first assign IDs to the vertices on $S$ in clockwise order from $0$ to $|S|-1$. Then, the vertex labeled $i$ knows that its incident edges to $X^+$ are precisely those from the edge connecting $i+1$ to the edge connecting $i-1$, inclusive. (Here, addition and subtraction are taken mod $|S|$.) We now compute a BFS through $G[X^+]$ as follows. Start from an arbitrary vertex on $S$. Every time we visit a new vertex on $S$ (including the initial vertex), we traverse through all its incident edges inside $G[X^+]$, which it has already computed. Every time we visit a vertex not on $S$, we traverse through all its incident edges. Since $G[X^+]$ has diameter $O(D\logn)$, the BFS will terminate in $O(D\logn)$ rounds. We now set $X^+$ as a child bag in the decomposition. 
\paragraph{Defining the remaining children.}
For the vertices on the other side of the separator, one attempt is to similarly define $X^-$ to be the vertices on the outside of the cycle $S$ in the planar embedding. (We include the vertices in $S$ itself in $X^-$.) However, it is not clear how to define $\m O(X^-)$: we could try to define it as $\m O(X) \setminus \ol{\m O(X^+)}$, but although this region is open, it might not be connected; see Figure~\ref{fig:bdd}. To preserve property~(8), we may need multiple bags $X^-_i$ instead of a single bag $X^-$.


First, suppose that $S$ does not contain any vertices in $\ol{\m O(X)}\setminus\m O(X)$, the boundary of $\m O(X)$. In this case,  $\m O(X) \setminus \ol{\m O(X^+)}$ is actually connected, so our initial attempt actually works. A straightforward adaptation of \Cref{clm:conn} shows that  \new{$T'[X^-]$} is also connected, so the algorithm proceeds identically to the $X^+$ case.

Otherwise, $S$ has vertices lying on $\ol{\m O(X)}\setminus\m O(X)$.   In this case, we first define a child bag $X^-_i$ for each connected region $\m O$ of $\m O(X)\setminus\ol{\m O(X^+)}$ as all vertices whose embedding is in $\ol{\m O}$, with one modification explained later. 

The following claim follows the same argument as the one in \Cref{clm:conn}. Observe that the first claim is not true if $(T'\cup S)[X^-_i]$ is replaced by $T'[X^-_i]$, since the single virtual edge in $S\setminus T'$ may be needed for connectivity; see Figure~\ref{fig:bdd}. 

\BCL\label{clm:conn2}
$(T'\cup S)[X^-_i]$ is connected.
\ECL
\BP
By construction, $S[X^-_i]$ forms a connected segment. If it does not contain the single virtual edge in $S$, then it is inside $T'$, and we can follow the proof of \Cref{clm:conn}. Otherwise, it can be broken up at the virtual edge into two segments in $T' \cap S[X^-_i]$. We can follow the proof of \Cref{clm:conn} to show that every vertex in $X^-_i$ is connected to one of the two segments. Finally, adding the single edge in $S\setminus T'$ (which is in $(T'\cup S)[X^-_i]$) connects $X^-_i$.
\EP

Since $(T'\cup S)[X^-_i]$ and $T'[X^-_i] \s G[X^-_i]$ only differ by the one virtual edge, \Cref{clm:conn2} implies that all but possibly one subgraph $G[X^-_i]$ are connected. For these $X^-_i$, we define $\m O(X^-_i)$ as the corresponding connected region in $\m O(X) \setminus \ol{\m O(X^+)}$. For the possibly one remaining $X^-_i$ (let's call it $X^-_j$) disconnected by the virtual edge, we break it up into two connected components $X^-_{j,1},X^-_{j,2}$ instead, and divide the corresponding region in $\m O(X) \setminus \ol{\m O(X^+)}$ to separate the embeddings of $X^-_{j,1}$ and $X^-_{j,2}$. The two divided regions form $\m O(X^-_{j,1})$ and $\m O(X^-_{j,2})$. The final, modified child bags are therefore $X^-_{j,1}$, $X^-_{j,2}$, and all the remaining untouched $X^-_i$'s.

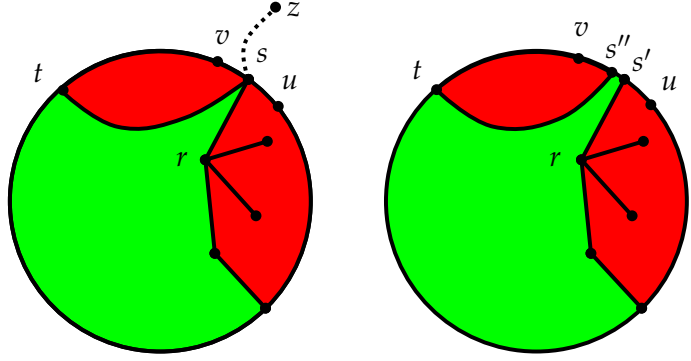
\begin{figure}\centering
\begin{tikzpicture}

\node at (6.0937,2.0575) {$s''$};
\node at (1.3576,1.9385) {$s$};
\node at (0.8255,2.1845) {$v$};
\node at (1.7232,1.5799) {$u$};
\node at (1.7784,2.5613) {$z$};
\node at (5.5975,2.2931) {$v$};
\node at (6.7679,1.5763) {$u$};
\node at (6.4069,1.8791) {$s'$};

 \tikzstyle{every node}=[circle, fill, scale=.4];
\draw[fill=green] [line width=1.5pt] (0,0) ellipse (2 and 2);

\node (v1) at (1.1691,1.6204) {};
\node (v2) at (0.5986,0.5505) {};
\node (v3) at (0.7254,-0.6937) {};
\node (v4) at (1.399,-1.4227) {};
\node (v5) at (1.4227,0.7962) {};
\node (v6) at (1.2721,-0.1944) {};
\node (v11) at   (1.1724,1.6207) {};
\draw[red, fill=red] (1.1724,1.6207) arc (54.1184:-45.4658:2.0003) -- (v3.center) -- (v2.center) --   (v11.center) {};
\draw [line width=1.5pt]  (v1.center) -- (v2.center) -- (v3.center) -- (v4.center);
\draw [line width=1.5pt]  (v5.center) -- (v2.center) -- (v6.center);
\draw [line width=1.5pt] (0,0) ellipse (2 and 2);

\draw [red,fill]  plot[smooth, tension=.7] coordinates {(v11) (0.3361,1.1112) (-0.573,0.9909) (-1.3283,1.4789)} --  (-1.3283,1.4855) arc (131.8024:54.2221:1.9928);
\draw [line width=1.5pt]  plot[smooth, tension=.7] coordinates {(v11) (0.3361,1.1112) (-0.573,0.9909) (-1.3283,1.4789)} --  (-1.3283,1.4855) arc (131.8024:54.2221:1.9928);

\node (v1x) at (1.1691,1.6204) {};
\node (v2x) at (0.5986,0.5505) {};
\node (v3x) at (0.7254,-0.6937) {};
\node (v4x) at (1.399,-1.4227) {};
\node (v5x) at (1.4227,0.7962) {};
\node (v6x) at (1.2721,-0.1944) {};
\node (v11x) at   (1.1724,1.6207) {};

 \tikzstyle{every node}=[circle, fill, scale=.4];
\draw[fill=green] (5,0) ellipse (2 and 2);

\node (vv1) at (6.1691,1.6204) {};
\node (vv2) at (5.5986,0.5505) {};
\node (vv3) at (5.7254,-0.6937) {};
\node (vv4) at (6.399,-1.4227) {};
\node (vv5) at (6.4227,0.7962) {};
\node (vv6) at (6.2721,-0.1944) {};
 \tikzstyle{every node}=[circle, scale=1];
\node (vv12) at (6.008,1.6899) {};
\node (vv11) at (6.1724,1.6207) {};
\draw[red, fill=red] (6.1724,1.6207) arc (54.1184:-45.4658:2.0003) -- (vv3.center) -- (vv2.center) --   (vv11.center) {};
\draw [line width=1.5pt]  (vv1.center) -- (vv2.center) -- (vv3.center) -- (vv4.center);
\draw [line width=1.5pt]  (vv5.center) -- (vv2.center) -- (vv6.center);
\draw [line width=1.5pt] (5,0) ellipse (2 and 2);

\node at (0.2877,0.5633) {$r$};
\node at (-1.6218,1.709) {$t$};
\node at (5.2526,0.5633) {$r$};
\node at (3.4169,1.7459) {$t$};

\draw [red,fill]  plot[smooth, tension=.7] coordinates {(vv12) (5.3361,1.1112) (4.427,0.9909) (3.6717,1.4789)} --  (3.6717,1.4855) arc (131.6298:58.9232:1.9803);
\draw [line width=1.5pt]  plot[smooth, tension=.7] coordinates {(vv12) (5.3361,1.1112) (4.427,0.9909) (3.6717,1.4789)} --  (3.6717,1.4855) arc (131.6298:58.9232:1.9803);

 \tikzstyle{every node}=[circle, fill, scale=.4];

\node (vv1x) at (6.1691,1.6204) {};
\node (vv2x) at (5.5986,0.5505) {};
\node (vv3x) at (5.7254,-0.6937) {};
\node (vv4x) at (6.399,-1.4227) {};
\node (vv5x) at (6.4227,0.7962) {};
\node (vv6x) at (6.2721,-0.1944) {};

\node at (6.0041,1.7141) {};
\node at (3.6727,1.4841) {};
\node at (-1.292,1.4773) {};
\node at (1.5678,1.2618) {};
\node at (0.7617,1.858) {};
\node at (6.5215,1.2829) {};
\node at (5.5606,1.8989) {};

\draw [line width=1.5pt,dotted] plot[smooth, tension=.7] coordinates {(v11x) (1.1255,2.106) (1.532,2.5619)};
\node at (1.532,2.5815) {};
\end{tikzpicture}

\label{fig:split-s}
\caption{Splitting $s$ into $s'$ and $s''$.}
\end{figure}

\paragraph{Step 4: Computing the children $X^-_i$.}
To compute each such $X^-_i$ simultaneously, we will assume that every vertex $s$ knows whether it is on the boundary $\ol{\m O(X)}\setminus\m O(X)$, and if so, its two consecutive neighbors $u,v$ in the planar embedding that lie on opposite ``sides'' (see Figure~\ref{fig:split-s}) defined as follows: if we draw a curve on $\R^2$ from (the embedding of) $s$ to an ``imaginary'' neighbor $z$ outside $\ol{\m O(X)}$ such that the curve does not intersect $\m O(X)$ (see Figure~\ref{fig:split-s}), then these neighbors $u,v$ of $s$ are the ones clockwise and counter-clockwise from $z$, respectively \footnote{The only purpose of defining $z$ is to formally define the neighbors $u$ and $v$ of $s$. So from this point on, we can forget about the existence of $z$, which was ``imaginary'' to begin with.}. 
We assume that we know this information for now, and describe how to maintain it later on.

First, if there is no vertex in $S$ whose embedding lies on $\ol{\m O(X)}\setminus\m O(X)$, then 
there is only one child bag $(X \setminus X^+)\cup S$, which can be found by a simple modification of Step~3. Otherwise, the algorithm first detects the connected regions of  $\m O(X) \setminus \ol{\m O(X^+)}$ as follows. We construct an auxiliary graph where every vertex $s\in S$ whose embedding lies on $\ol{\m O(X)}\setminus\m O(X)$ is \emph{split} into two vertices $s',s''$ as follows: let $u,v$ be defined as before, and define $r$ and $t$ as the clockwise-to-$u$ and counterclockwise-to-$v$ neighbors of $s$ on $S$, respectively; see Figure~\ref{fig:split-s}. The neighbors of $s$ from $u$ to $r$ in the clockwise direction become neighbors of $s'$, and the neighbors from $t$ to $v$ in the clockwise direction become neighbors of $s''$. We do not add an edge between $s'$ and $s''$. The algorithm then computes connected components in this auxiliary graph, which corresponds to the bags $X^-_i$. Recall that every bag has diameter $O(D\log n)$ by property~(6'), so this takes $O(D\log n)$ time by having node $s$ simulate both $s'$ and $s''$ in the graph.

Also, we may split the bag with the virtual edge into two, which is easy to compute.

Finally, the information $u,v$ for each relevant $s$ can be kept track of recursively as follows: there are no such vertices at the beginning, and given a bag $X$, the set of new vertices $s$ on the boundary  $\ol{\m O(X')}\setminus\m O(X')$ of a child bag $X'$ equals $S_X\cap X'$, and their two consecutive neighbors $u,v$ can be easily computed from knowing the orientation of $S_X$ and the planar embedding. 

In particular, (i) every vertex $s$ on $S_X$ is now on the boundaries $\ol{\m O(X^-_*)}\setminus\m O(X^-_*)$\footnote{The ``wildcard'' $*$ can take value $i$ or value $j,1$ or $j,2$.} of the relevant regions $\m O(X^-_*)$ and its neighbors $u,v$ can be easily computed given the orientation of $S_X$, and (ii)~the regions $\m O(X^-_*)$ can be defined so that every vertex not on $S_X$ is not on any boundary.

\BL
This decomposition algorithm satisfies all the properties of a BDD decomposition.
\EL
\BP
Properties (2) and (4) follow immediately from the construction.

For property (1), for each bag $X$ with $|X|\ge\Om(D\log n)$, all the children $X^+,X^-_i$ have size at most $(5/6)|X|$, so $O(\logn)$ iterations are needed.

Since $G[X^+]$ is connected (\Cref{clm:conn}), and the graphs $G[X^-_i]$ are all connected except possibly $X^-_*$ (\Cref{clm:conn2}), and the possible $X^-_*$ is broken into two connected components, we have properties (3) and (5). Moreover, since (i) $S$ separates $\m O(X^+)$ from $\m O(X^-_i)$ for all $i$, and (ii) the $\m O(X^-_i)$ are connected open regions in $O(X) \setminus \ol{\m O(X^+)}$, and (iii) the possibly one $X^-_i$ that is broken into two components has its region divided into two disjoint regions $\m O(X^-_{i,1})$ and $\m O(X^-_{i,2})$, a simple induction from the root to the leaves proves property (8).  Taking the planar embedding of the edges in $S$ produces the desired closed curve. Since the only nodes shared among multiple children are the ones in $S$, we have property (9). 

The most interesting part is proving property~(6). To do so, for each bag $X$, define $T'_X$ as the spanning tree $T'$ of $T[X]$ defined in Step~1 of the algorithm for bag $X$.   
For a bag $X$ with cycle separator $S_X$ and a child $X'$, $(T'_X\cup S_X)[X']$ is connected by \Cref{clm:conn}~or~\Cref{clm:conn2},  depending on whether $X'$ is inside or outside $S_X$ in the planar embedding. It follows that $T[X']$ has at most one more connected component than $T[X]$, so by a top-down induction, for a bag $X$ at depth $i$,  $T[X]$ has at most $i$ connected components. Since the depth $i$ is at most $O(\log n)$, this completes property~(6).

Lastly, we prove property (7) assuming property~(6). By construction, all nodes shared among more than one child bag of $X$ must belong to $S_X$. Since $T[X]$ has at most $O(\logn)$ connected components, by construction of $T'_X$, there are at most $O(\logn)$ edges in $T'_X\setminus T$. Since $S_X$ minus a single edge is a path contained in $T'_X$, and since $T'_X$ itself has only $O(\logn)$ edges not in $T$, property (7) follows. 
\EP

\subsection{Distributed Computation of (Exact) Distance Labels} \label{sec:exact-labels}
In a seminal paper, Gavoille et al. \cite{gavoille2004distance} presented a distance
labeling scheme with labels of size $O(r(n)\log^2 n)$ for the class of graphs with (recursive) $r(n)$-separators. 
Using $r(n)=O(\sqrt{n})$ yields the well known bound of $\widetilde{O}(\sqrt{n})$ distance labels. 
In our distance labels, we will use $r(n)=O(D\logn)$ for $n>D$. To make sure that an $O(D\logn)$-separator can be computed recursively, we use the BDD decomposition.

We first describe the labeling scheme by \cite{gavoille2004distance}. For the given graph $G$, the label $L_G(v)$ of each $v \in G$ consists of the following fields: (i) a list of distances $d_G(v,s)$ for every $s \in \sep(G)$, where $\sep(G)$ is the separator of $G$, (ii) the component ID of $v$ in $G \setminus \sep(G)$, and (iii) the (recursive) label $L_{G'}(v)$ where $G'$ is the component of $v$ in $G \setminus \sep(G)$. For sake of explanation, we denote the first two fields in the label by $\widehat{L}_{G}(v)$. Hence, a label of vertex $v$ consists of the concatenation of $k=O(\log n)$ sub-labels 
$$L_G(v)=\widehat{L}_{G_0}(v) \circ \widehat{L}_{G_1}(v) \ldots  \circ \widehat{L}_{G_k}(v), \mbox{~~where~~} \widehat{L}_{G_i}(v)=\langle ID(G_{i+1}), \,\{(s,d_{G_i}(v,s)),\, s \in \sep(G_{i})\} \rangle,$$ 
$G_0=G$ and $G_i$ is the component of $v$ in $G_{i-1}\setminus \sep(G_{i-1})$ for every $i \in \{1,\ldots, k\}$. The decoding of the distance labels $L_G(u)$ and $L_G(v)$ is done by computing the minimum $u$--$v$ distance via the distances to the separator vertices in each level. 

Our construction of the labels is also recursive, but requires a small adaptation from the scheme of \cite{gavoille2004distance}. 
Since our labels are based on $O(D\logn)$-size separator, the recursion is applied on the components of the BDD decomposition. Observe that in the recursive scheme of \cite{gavoille2004distance}, all subgraphs in a given recursive level are \emph{vertex disjoint}. This is because the child subgraphs of a graph $G'$ are the connected components of $G' \setminus \sep(G')$.
This disjointness property is indeed crucial for the size bound of the final labels, as the label of a vertex $v$ contains the concatenation of the sublabels $\widehat{L}_{G'}(v)$ for all the subgraphs $G'$ in the recursion tree that contain $v$.
In contrast, the level-$i$ subgraphs of the BDD decomposition are not necessarily vertex disjoint, as the vertices of the separators $S_X$ might be added to several subgraphs in order to keep the diameter small. To prevent a blow-up in the label size, we slightly modify the scheme of \cite{gavoille2004distance} as follows. 

Initially, all vertices are marked. In every independent level of the BDD recursion with bag $X$, we are given a subgraph $G':=G[X]$ and assume first that $|X| =\Omega(D \log n)$. By property~(7'), $S_X=O(D\logn)$. The algorithm computes a BFS tree from each $s \in S_X$ in $G'$. For the child bags $X_1,\lds,X_i$ of $X$, define $G'_i:=G'[X_i]$. The label $L_G(v)$ of each \emph{marked} vertex $v \in G'$ is appended with the sub-label $\widehat{L}_{G'}(v)$ which consists of (i) list of distances in $G'$ from $v$ to each $s \in S_{X}$, and (ii) the ID of the subgraph  $G'_j$ to which $v$ belongs.
Finally, all vertices in $S_X$ are \emph{unmarked}.  If $|G'|=O(D\log n)$ (leaf bag in the decomposition), we compute a BFS tree from each $v \in G'$, and append the label of each marked vertex in $G'$ with the list of distances to every $u \in G'$.
This completes the description of the algorithm. 

\paragraph{Round complexity (given the BDD decomposition).}
By property (6) of the BDD decomposition, the diameter of all subgraphs is bounded by $O(D\logn)$. In addition, by property (8'), each edge is shared among at most two subgraphs. Hence computing the $O(D\logn)$ BFS trees in all the subgraphs of level-$i$ in the recursion can be done simultaneously within $\tO(D)$ rounds. Since there are $O(\log n)$ recursion levels, the total round complexity is also bounded by $\tO(D)$.

\paragraph{Correctness.}
We will claim that for every subgraph $G'$ and every pair of marked nodes $u,v \in G'$, the distance $d_{G'}(u,v)$ can be computed from $L_{G'}(u)$ and $L_{G'}(v)$. Since in the first level $G$, all the vertices are marked, this will establish the correctness.
The base of the induction holds trivially for the leaf bags. Assume that it holds for all bags up to level $i$, and consider a bag $G'$ in level $i-1$ with children $G'_1,\ldots, G'_k$ in level $i$. Let $u$ and $v$ be two marked nodes in $G'$. 
There are two options. (i) The shortest $u$-$v$ path $P$ in $G'$ intersects $S_{X}$ at some vertex $w$. In such a case, $d_{G'}(u,v)=d_{G'}(u,w)+d_{G'}(w,v)$. Since the sub-labels $\widehat{L}_{G'}(u), \widehat{L}_{G'}(v)$  contain the distances $d_{G'}(u,w)$ and $d_{G'}(w,v)$ respectively, the distance $d_{G'}(u,v)$ can be be computed. (ii) The shortest $u$-$v$ path $P$ in $G'$ does not intersect $S_X$. By property~(3'), this implies that $u$ and $v$ must be in the same component $G'_j$ in the BDD tree for some $j \in \{1,\ldots, k\}$. 
We have that $d_{G'}(u,v)=d_{G'_j}(u,v)$ and by the induction assumption, $d_{G'_j}(u,v)$ can be computed from 
the labels $L_{G'_j}(u)$ and $L_{G'_j}(v)$. 
\paragraph{Label Size.}
The key observation is that the marked vertices of the level $i$-subgraphs are \emph{vertex disjoint}, for every $i$. This can be shown by induction on $i$. For the base case of $i=0$ the claim holds vacuously. Assume that it holds for $i-1$ and consider some bag $G'$ in level $i-1$, with children bags $G'_1,\ldots, G'_k$. By definition of $S_X$ (property~(7')), $G'_a \cap G'_b \s S_X$ for every distinct $a, b \in \{1,\ldots, k\}$. Since the vertices of $S_X \cap G'_a$ are unmarked for every $a$, combining with the induction assumption for level $i-1$, we get that all \emph{marked} vertices in each level $i$ subgraph are vertex-disjoint. 
The label of a vertex $v$ contains the sub-labels $\widehat{L}_{G'}(v)$ for every bag $G'$ in the BDD tree such that $v$ is a marked node in $G'$. We have that $v$ can be a marked vertex in at most $O(\log n)$ subgraphs, one per level. Since each $\widehat{L}_{G'}(v)$ has $O(D\log ^2n)$ bits, the total label size is bounded by $O(D \log^3 n)$. 

%
%

\subsection{The Distributed Diameter Algorithm}\label{sec:unw-diam}
The diameter is computed on the BDD tree $\mathcal{T}$ from the leaf bags up to the root.
The invariant for phase $i$ is that every node in every bag $X$ in level $D(\mathcal{T})-i+1$ knows
$$d(X)=\max_{u, v \in X}d_{G}(u,v)~.$$
Note that $d(X)$ might be smaller than the diameter of $G[X]$, since it is based on the distances in $G$ rather than in $G[X]$. By keeping this invariant after each step, we get that after $D(\mathcal{T})=O(\log n)$ phases, the root vertex $r$ knows $D=d(V)$.

For the leaf bag $X$, simply assign a leader to collect the distance labels $L_G(v)$ of all vertices $v \in X$, locally compute $d(X)$, and then broadcast it to other nodes in the bag. Assume that the invariant holds up to phase $i$ for all bags in level $\ell_i=D(\mathcal{T})-i+1$, and we now describe phase $i+1$.
Let $X$ be a bag in level $\ell_i-1$ and let $X_1,\ldots, X_k$ be its children bags in level $\ell_i+1$. 
By the invariant, every node in each child bag $X_j$ knows $d(X_j)$. 

Let $u,v \in X$ be the pair of largest $G$-distance in $X$, i.e., $d(X)=d_{G}(u,v)$.
There are two options: (i) $u,v \subseteq X_j$ for some $j \in \{1,\ldots, k\}$, or (ii) $u \in X_j$ and $v \in X_{j'}$ for some $j \neq j'\in [k]$. Case (i) can be easily handled since nodes in $X_j$ know $d(X_j)$, and the maximum $d(X_j)$ value over all $j\in [k]$ can be computed along a BFS spanning tree in $G[X]$ of depth $O(D\logn)$, guaranteed by property (6'). Thus, the nodes in $X$ can compute $\max_j d(X_j)$. 

Before we handle case (ii), let us define $G^+$ to be all vertices and edges whose planar embedding is on or inside the closed curve guaranteed by property (9) of BDDs. Similarly, define $G^-$ to be all vertices and edges whose planar embedding is on or outside the closed curve. Define $X^+$ and $X^-$ as the vertex sets of $G^+$ and $G^-$, respectively. Observe that by property (9), every child $X_j$ satisfies $X_j\s X^+$ or $X_j\s X^-$, but not both. For a vertex $v\in X$, define $G^\pm(v)$ as the graph $G^+$ or $G^-$ that contains $v$. If both do (e.g., when $v\in S_X$), then choose one arbitrarily. Define $X^\pm(v)$ as the vertex set of $G^\pm(v)$. Finally, define $G':=G[X]$ and $G'_j:=G[X_j]$ as before.

By property (3') of BDDs, any $u$--$v$ path in  case (ii) must pass through a vertex in $S_X$. In particular, the shortest $u$--$v$ path must travel inside $X^\pm(u)$ until reaching some node $s\in S_X$, then take the shortest $s$--$t$ path in $G$ to some $t\in S_X$ (possibly $t=s$), and finally travel inside $X^\pm(v)$ to $v$. Therefore, it suffices to compute
\begin{gather}  \max_{u\in X_j,\ v \in X_{j'},\ j\ne j'}\ \min_{s,t \in S_X} (d_{G^\pm(u)}(u,s) + d_G(s,t) + d_{G^\pm(v)}(t,v)) \label{eq:diam} .\end{gather}
The value $d_G(s,t)$ in the $\min$ expression in (\ref{eq:diam}) can be computed using the distance labels in $G$: every node in $S_X$ can simply broadcast its distance label to all nodes in $X$. 

For the other values $d_{G^\pm(u)}(u,s)$ and $d_{G^\pm(v)}(t,v)$, we compute distance in $G^+$ and $G^-$ so that every vertex $v\in G'$ knows its $S_X$-distance tuple in $G^\pm(v)$. Note that distances in $G$ will not work, since we want to apply \Cref{thm:exact-compr} to the graphs $G^+$ and $G^-$ (which each have $S_X$ lying on a single face), and the theorem requires distances in $G^+$ and $G^-$. To accomplish this for $G^-$, for each child bag $X_j$, first compute distance labels in $X_j$,\footnote{The structure of the distance labels in \Cref{sec:exact-labels} make computing these easy: simply read off the suffix of the distance label $L_G(v)=\widehat{L}_{G_0}(v) \circ \widehat{L}_{G_1}(v) \ldots  \circ \widehat{L}_{G_k}(v)$ beginning at $\widehat L_{G_j'}(v)$. Alternatively, if the reader wants to black-box \Cref{sec:exact-labels}, then the distance labels can be computed from scratch in each $G_j$. By property~(8'), every edge is in at most two $G_j$ at this level, so the computations over different $G_j$ can be done simultaneously in $\tO(D)$ rounds.} and have every pair of nodes $u,v\in S_X\cap X_j$ learn their distance in $G_j$. Then, consider an auxiliary graph $H$ on vertices $S_X$ where vertices $u,v\in S_X$ have an edge of length $\min_{G'_j\s G^-:\,u,v\in X_j}d_{G'_j}(u,v)$ (possibly $\infty$). Broadcast this graph to all nodes in $G^-$, so that each node $u\in X_j\s X^-$ can compute, for each $t\in S_X$, $\min_{s\in S_X\cap X_j} (d_{G_j}(u,s)+d_H(s,t))$.
By property (3'), for every $u\in X^-$ and $t\in S_X$, the (edges of the) shortest $u$--$t$ path can be broken into subpaths with endpoints in $S_X$ and which stay entirely inside some $G'_j$, and the distance between these endpoints is correctly computed. Therefore, this correctly computes distances in $G^-$. We can then repeat the same argument for $G^+$.\footnote{Note that \Cref{clm:conn2} implies that there is only one child bag $X_j$ inside $X^+$, so this case can be made even easier.}


Therefore, from now on, we assume that every vertex $v\in G'$ has computed its $S_X$-distance tuple in $G^\pm(v)$. Let us assume that all component IDs are represented by $B=O(\log n)$ bits. We then repeat the following process for each $b\in[B]$: Let $\m B_{b,0}$ be all components $G'_j$ whose $b$'th bit in their component ID is $0$, and let $\m B_{b,1}$ be those whose $b$'th bit in their component ID is $1$. Then, every vertex sends its distance tuple along the rooted spanning tree, together with the $b$'th bit of its component ID, so that the root learns the set of distinct (distance tuple, $b$'th~bit~of~component~ID) pairs. By \Cref{thm:exact-compr}, there are at most $\tO(D^4)$ many distinct tuples, so this can be done efficiently, whose analysis is deferred to the Round Complexity section. 
Once the root of the spanning tree receives all these pairs, it computes
\begin{gather}
\max_{\substack{ u\in G'_j\in\m B_{b,0}\\ v\in G'_{j'}\in B_{b,1}}} \ \min_{s,t\in S_X}(d_{G'^\pm(u)}(u,s) + d_G(s,t) + d_{G'^\pm(v)}(t,v)) , \label{eq:diamHT}
\end{gather}
by trying all pairs of a $(\cd,0)$ tuple and a $(\cd,1)$ tuple. Finally, the two components $G_j,G_{j'}$  achieving the $\max$ in Eq. (\ref{eq:diam}) must have their components IDs differ in some bit position $b$,  which means the root of the spanning tree correctly computes Eq. (\ref{eq:diam}) for that value of $b$. So by trying all $b\in[B]$, one of them will succeed.

\paragraph{Round Complexity.}
We show that phase $i$ can be implemented in $\widetilde{O}(D^5)$ rounds, and since there are $O(\log n)$ phases this establishes the claim. Since each edge $e$ appears on at most $2$ subgraphs in level $i$ by property (8') of BDDs, it is sufficient to analyze the round complexity for one such subgraph $G'=G[X]$. Then, we can work on all level subgraphs of level $i$ in parallel while paying a factor $2$  in the round complexity.

First, sending the distance labels of all $S_X$ vertices to all the vertices in $G'$ can be done by pipelining these labels along the rooted spanning tree. Since $G'$ has diameter $O(D\log n)$ by propery~(6'), $|S_X|=O(D\log n)$ by property~(7'), and distance tuples have size $\tO(D)$, this can be done in $\tO(D^2)$ rounds. As for computing the $S_X$-distance tuples in $G^\pm(v)$ for each $v\in X$, the distance labels can be computed in $\tO(D)$ time and pipelining the graph $H$ takes $O(|S_X|^2)=\tO(D^2)$ time, so this also takes $\tO(D^2)$ rounds.

The most expensive part is gathering the distance tuples. Since $S_X$ is a face in $G^+$, by \Cref{thm:exact-compr}, there are $O(|S_X|^3D) = \tO(D^4)$ many distinct distance tuples among the vertices in $G^+$; the same holds for $G^-$. Therefore, sending the distance labels can be implemented in $\widetilde{O}(D^6)$ rounds trivially: starting from $d=O(\logn)$ to $d=1$, all nodes at depth $d$ in the spanning tree send their $\tO(D^4)$ many $\tO(D)$-sized labels to their parents, and then each parent computes the union of its own labels, together with all labels it received.
 We can speed this up to $\tO(D^5)$ rounds as follows. Compute a hash of $O(\log n)$ bits for each distance tuple; w.h.p., no two hashes of different distance tuples collide. Then, repeat the same procedure as above, except with the hashes, taking $\tO(D^5)$ time. For each hash, we can add a tag of the ID of an arbitrary node that contains the distance tuple with this hash, to be sent along the spanning tree. The root then broadcasts down all (hash, ID) pairs. Finally, for each pair, the node with the corresponding ID sends the original $\tO(D)$-sized distance tuple. This is $\tO(D^4)$ many $\tO(D)$-sized labels to send along the spanning tree, which can be pipelined to run in $\tO(D^5)$ rounds.

\section{Distributed Distance Labels and SSSP in Weighted Graphs}\label{sec:labelssp}
In this section we consider a weighted graph $G=(V,E,w)$. 
Throughout, let $D(G)$ be the unweighted diameter of the graph $G$, when $G$ is clear from the context, we simply write $D$. Let $W$ be the largest edge weight in $G$, we will assume that $W=\poly(n)$. 
%
We will show the following:
\begin{lemma}\label{lem:weighted-exact-label}
For every weighted graph $G=(V,E,\omega)$, there exists a randomized algorithm for computing exact distance labels of size $\widetilde{O}(D)$ bits within $\widetilde{O}(D^2)$ rounds.
\end{lemma}

The label structure will be exactly the same as that of the unweighted case, but the computation procedure is slightly more involved.  In contrast to the unweighted setting, we will not compute the distances from the separator vertices by computing their shortest path trees, as the latter step is too costly for weighted graphs. Instead, these distances will be computed in a bottom-up manner from the leaf bags (in the BDD tree) up to the root. 

Given the BDD decomposition, the labels are computed from the leaf bags up to the root, as follows.
%
The label $L_{G[X]}(v)$ for a leaf bag $X$ contains the list of all $X \times \{v\}$ distances in $G[X]$.
These distances can be computed locally at each node by collecting all edges of $G[X]$ (recall that $|X|=O(D \cdot \log n)$ by property (4)). 

Assume that the invariant holds up to level $i+1$, and consider a bag $X$ in level $i$ and its children $X_1,\ldots, X_k$ in level $i+1$, where we define $G':=G[X]$ and $G'_i:=G[X_i]$ as before. By the induction invariant, we assume that for every $X_i$, every vertex $v \in X_i$ has computed the distance label $L_{G'_i}(v)$.
Recall that to compute $L_{G'}(v)$, it is sufficient to compute the sub-label $\widehat{L}_{G'}(v)$ which contains the list of $S_X \times \{v\}$ distances in $G'$. To compute these distances, every vertex $s \in S_X$ sends to all the vertices in $G'$, its label $L_{G'_{j}}(s)$ for every $G'_j$ that contains $s$. 
In the analysis section, we will show that each $s \in S_X$ might appear on at most \emph{three} such subgraphs. All these labels are sent in a pipeline manner on a BFS tree in $G'$. Equipped with this distance information, each vertex $v \in G'$ computes its distances to $S_X$ in $G'$ by computing (locally) a weighted complete graph $\widehat{G}(v)$ on the vertex set $V(\widehat{G}(v)):=\{v\} \cup S_X$ with the following edge weights.
For every $x,y \in V(\widehat{G}(v)) \cap V(G'_i)$, the distance $d_{G'_i}(x,y)$ can be computed from the labels $L_{G'_i}(x)$ and $L_{G'_i}(y)$ for every $i \in \{1,\ldots, k\}$. The edge weight of $(x,y)$ is defined by
$w(x,y)=\min_{i=1}^k \widehat{d}_{G'_i}(x,y),$
where $\widehat{d}_{G'_i}(x,y)=d_{G'_i}(x,y)$ if both $x,y \in G'_i$ and $\widehat{d}_{G'_i}(x,y)=\infty$ otherwise. 
The distances  $S_X \times \{v\}$ for the sub-label $\widehat{L}_{G'}(v)$ are computed by running Dijkstra (locally) on $\widehat{G}(v)$. 
The label $L_{G'}(v)$ for every vertex $v \in G'_i$ and every $i \in \{1,\ldots, k\}$ is defined as follows:
$$
  L_{G'}(v) =
  \begin{cases}
    \widehat{L}_{G'}(v) \circ L_{G'_i}(v) & \text{for } v \notin S_X \\
    \widehat{L}_{G'}(v) & \text{for } v \in S_{X}.
        \end{cases}
$$

%
%
%
%
%
This completes the description of the algorithm.
%
We now claim: 
\BL\label{lem:ExactDistLabel}
The algorithm computes exact distance labels of size $O(D \log n\cdot \log W)$ within $\widetilde{O}(D^2)$ rounds.
\EL
\paragraph{Correctness and label size.}
The correctness is shown by induction on the subgraphs of the BDD decomposition, from the leaf subgraphs up to the root $G$. Let $d$ be the depth of $\mathcal{T}$. 
The base case is immediate as the label contains all distances in the leaf subgraph. Assume that 
for all subgraphs $G''$ up to level $i+1$, it holds that: (i) using $L_{G''}(u), L_{G''}(v)$, one can compute $d_{G''}(u,v)$ for every $u,v \in G''$; (ii) $|L_{G''}(u)|\leq \lceil D \log W \rceil \cdot (d-i+1)$. 

We will now consider a subgraph $G'$ in level $i$, with its children $G'_1,\ldots, G'_k$ in level $i+1$. We first show the correctness of the labels.
Starting with claim (i), we first show that each vertex $v$ correctly computes the $S_X \times \{v\}$ distances in $G'$. Recall that $v$ locally computes a weighted clique $\widehat{G}(v)$ with edge weights assigned based on the labels of $v$ and the labels of $S_X$ in $G'_1,\ldots, G'_k$. Observe that by the induction assumption on the labels of the subgraphs in level $i+1$, we have that $w(x,y) \geq d_{G'}(x,y)$ for every $x, y \in \widehat{G}(v)$.
Fix $s \in S_X$ and let $P$ be an $v$--$s$ shortest path in $G'$. Let $s_1,\ldots, s_k=s$ be the vertices on $P \cap S_X$ in order of their appearance on $P$ (i.e., sorted in increasing distance from $v$). We claim by induction on $j$ that $d_{\widehat{G}(v)}(v,s_j)=d_{G'}(v,s_j)$. For the base of the induction, consider $s_1$ and note that $P[v,s_1]$ is fully contained in one of the $G'_{\ell}$ subgraphs for some $\ell \in \{1,\ldots, k\}$. Thus $d_{G'}(v,s_1)=d_{G'_\ell}(v,s_1)$ and by the induction assumption on the labels $L_{G'_\ell}(v),  L_{G'_\ell}(s_1)$, we have that $w(v,s_1)=d_{G'_\ell}(v,s_j)$. Assume that the claim holds up to $s_{j-1}$. We will show that $d_{G'}(v,s_j)=d_{\widehat{G}(v)}(v,s_j)$, it is sufficient to show that $w(s_{j-1},s_j)=d_{G'}(s_{j-1},s_j)$. Since the internal segment of $P[s_{j-1},s_j]$ does not intersect $S_X$, it is fully contained in one of the subgraphs $G'_{\ell'}$. The claim then follows by the induction assumption on the labels of $G'_{\ell'}$. 
So-far, we have proved the correctness of the sub-label $\widehat{L}_{G'}(v)$ for every $v \in G'$. In the final step, for every non-separator vertex $v \in G'_i \setminus S_X$, we set $L_{G'}(v)=\widehat{L}_{G'}(v) \circ L_{G'_i}(v)$; for $s \in S_X$, we let $L_{G'}(v)=\widehat{L}_{G'}(v)$. This exactly follows the labeling scheme of \cite{gavoille2004distance} for the graph $G'$, thus the correctness follows immediately.

Finally, we bound the label size. By definition of $S_X$ (property~(7)), every $v \in G'\setminus S_X$ belongs to exactly one of the $G'_i$ subgraphs. Since the label $L_{G'_i}(v)$ is extended by adding the sub-label $\widehat{L}_{G'}(v)$ which consists of $\lceil D \log W \rceil$ bits, the claim follows by combining with the induction assumption (ii) on the size of $L_{G'_i}(v)$. For $v \in S_X$, we have that $L_{G'}(v)=|\widehat{L}_{G'}(v)|\leq \lceil D \log W \rceil$, the claim follows. 
\vspace{-10pt}
\paragraph{Round complexity.}
Consider a subgraph $G'$ in level $i$, and let $G'_1,\ldots, G'_k$ be its children. 

We first claim that each $s \in S_X$ belongs to at most three $G'_j$ subgraphs. This holds by the properties of the BDD decomposition.
Since the label size is bounded by $\widetilde{O}(D)$, overall the total amount of label information is $\widetilde{O}(D^2)$. Sending this information on a BFS tree in a pipeline manner takes $\widetilde{O}(D^2)$ rounds. The remaining computation of the distances based on these labels are local. 

\paragraph{From labels to SSSP.} Let $s \in V$ be the input source. To compute the $s \times V$ distances, it is sufficient to send the label of $s$ to all the vertices. This can be done in $\widetilde{O}(D)$ rounds. At this point, all vertices $v$ can compute $d_G(s,v)$. We let nodes exchange this distance information with their neighbors. To define the tree, every vertex $v$ picks as it parent the neighbor $u=\arg\min\{d_G(s,x)+w(u,x), x \in N(u)\}$, breaking ties based on IDs.

\begin{lemma}\label{lem:ExactSSSP}
For every weighted graph $G=(V,E,\omega)$ and a fixed source vertex $r \in V$, there exists a distributed algorithm that computes an (exact) SSSP tree in planar graph within $\widetilde{O}(D^2 \log W)$ rounds where $W$ is the maximum edge weight.
\end{lemma}
%
%
%

\section{$(1+\epsilon)$ Diameter Approximation in Weighted Graphs}\label{sec:approxdiam}

\paragraph{Step (1): (Exact) SSSP.}
We begin with the SSSP tree $T$ from an arbitrary source $r$, computed in $\tO(D^2)$ time (\Cref{lem:ExactSSSP}). This also gives us a $2$-approximation of the weighted diameter, by finding the maximum distance from the source $r$ and multiplying that distance by $2$. Let $\widetilde D$ be this diameter estimate, so that the true diameter is in the range $[\widetilde D,2\widetilde D]$.

\paragraph{Step (2): BDD Decomposition on SSSP.}
Our next step is to compute a BDD on the \emph{SSSP tree} $T$. Note that $T$ might have arbitrarily large (unweighted) diameter. Nevertheless, we can modify the BDD decomposition scheme to run in $\tO(D^2)$ rounds\footnote{As elaborated in Sec. \ref{app:sssp-bdd} this might improve to $\tO(D^2)$ rounds if the separator algorithm of \cite{GhaffariP17}, and connectivity identification algorithm of \cite{ghaffari2016distributed} are applied in a non black-box manner.}. Another important property is that we modify the BDD so that every leaf bag has at most $O(1/\e \cdot \log^2 n)$ \emph{non-separator} nodes; By non-separator nodes we refer to vertices that are not in the separator $S_{X'}$  for all parent bags $X'$ of $X$ for each leaf bag $X$.
Again, it is easy to see that the BDD still has $O(\logn)$ levels.

For every bag $X$, the separator algorithm of~\cite{GhaffariP17} runs in time $\tO(D)$ on the computed spanning tree of $X$, regardless of its diameter. Moreover, this can be done in $\tO(D^2)$ time total in parallel over all bags $X$ at a given level, via the use of low-congestion shortcuts. One potential issue is the fact that the bags $X$ of a given level may share nodes, which means the shortcut parts may not be vertex-disjoint.

\dnsparagraph{Shortcuts on bags}
We remedy this issue by constructing a graph $\widehat G$ similar to the one in Section~4.2.1~of~\cite{GhaffariP17}. Our construction is as follows:
\BE
\im[i.] For each edge $v$ in more than one bag $X$ on a given level, $v$ makes a copy $v_X$ of itself for each such bag $X$ containing $v$. Add an edge $(v,v_X)$ for each such $X$.
\im[ii.] For each edge $(u,v)$ where $u$ has a copy $u_X$ and $v$ has no copy, replace $(u,v)$ with $(u_X,v)$. For each edge $(u,v)$ where $u$ and $v$ have copies $u_X$ and $v_X$, replace $(u,v)$ with $(u_X,v_X)$.
\EE
Note that we construct a different graph $\widehat G$ for each given level of the BDD.

\BL\label{lem:widehatG}
The graph $\widehat G$ satisfies the following properties:
\BE
\im $\widehat G$ is still planar.
\im $\widehat G$ has diameter $O(D)$.
\im An $r$-round distributed algorithm can be simulated in $2r$ rounds on $G$.
\EE
\EL
\BP
To show (1), consider the regions $\m O(X)$ for each bag $X$ on the given level, which are disjoint by property (8b) of the BDD. Imagine shrinking each $\m O(X)$ infinitesimally in the plane into a region $\m O'(X)$, so that their closures $\ol{\m O'(X)}$ are now also disjoint. For each vertex $v$ in multiple bags $X$, consider its new location $v'_X$ in each $\m O'(X)$; these locations are infinitesimally close to each other. Therefore, we can draw an infinitesimal curve from $v$ to each $v'_X$ in the plane without the curves intersecting. For each edge $(u,v)$ where $u$ has a copy $u_X$ but $v$ does not, the drawing of the edge can be shifted infinitesimally so that it now travels from $u'_X$ to $v$. For each edge $(u,v)$ where both $u$ and $v$ have copies $u_X$ and $v_X$, shift both endpoints of the edge infinitesimally so that it now travels from $u'_X$ to $v'_X$. Since the new closures $\ol{\m O'(X)}$ are disjoint, there can be no new edge crossings that arise. Since there were no crossings to begin with, there are still none, so the resulting drawing, which is an embedding of $\widehat G$, is planar. This proves property~(1). 

For property~(2), observe that for any path $v^1,v^2,\lds,v^\el$, every edge $(v^i,v^{i+1})$ in the path  can be replaced by either the path $v^i,v^i_X,v^{i+1}$, or the path $v^i,v^{i+1}_X,v^{i+1}$, or the path  $v^i,v^i_X,v^{i+1}_X,v^{i+1}$, depending on which of $v^i,v^{i+1}$ are duplicated in the graph. (If neither is duplicated, then the edge $(v^i,v^{i+1})$ still exists, so no replacement is necessary.) If we replace each edge on the path, the new path length is at most $3$ times the old length. Therefore, since $G$ has diameter $O(D)$, so does $\widehat G$.

Lastly, we prove property~(3). Every node $v$ can simulate all of its copies $v_X$. By property~(8'), each edge $(u,v)$ belongs to at most $2$ subgraphs of that level, so there are at most two edges $(u_X,v_X)$ in $\widehat G$. Therefore, for each distributed round on $\widehat G$, the messages that get passed along the different $(u_X,v_X)$ can be sent along $(u,v)$ in $G$ in $2$ rounds. Since every round in $\widehat G$ takes $2$ rounds on $G$, we have property~(3).
\EP

For each bag $X$ on the given level, define the vertex set $\widehat X$ in $\widehat G$ as follows: for each vertex $v\in X$, add $v$ to $\widehat X$ if $v$ has no copy of itself, and $v_X$ otherwise. The graph $\widehat G[\widehat X]$ is isomorphic to $G[X]$, and every vertex $v\in X$ knows its corresponding vertex in $\widehat X$. Therefore, we may compute shortcuts in each bag $\widehat X$ simultaneously, which are vertex-disjoint. By \Cref{lem:widehatG}, $\widehat G$ is still planar, so efficient shortcuts exist on $\widehat G$, and moreover, the computation on $\widehat G$ can be simulated efficiently back on $G$. From that point on, we assume that for every subgraph $G'$, we have a shortcut subgraph $H'$, such that $G' \cup H'$ has diameter $O(D\log n)$, and each edge appears on $O(D\log n)$ many $H''$ subgraphs for every shortcut subgraph $H''$ for a subgraph $G''$ in that level. This allows us working in all subgraphs of the same level efficiently.

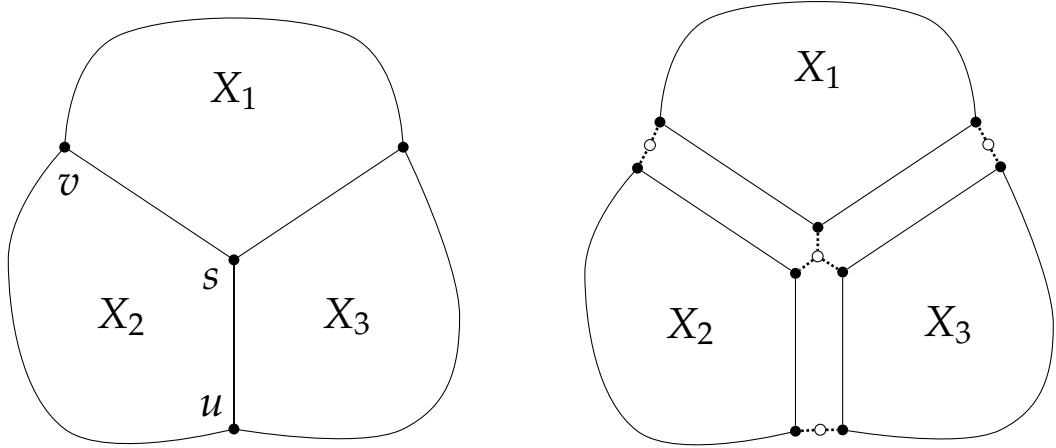
\begin{figure}
\centering
\begin{tikzpicture}[scale=.75]
\node[scale=1.5] at (0,3) {$X_1$};
\node[scale=1.5] at (-2,-1) {$X_2$};
\node[scale=1.5] at (2,-1) {$X_3$};
\node[scale=1.5] at (-0.4005,-0.3405) {$s$};
\node[scale=1.5] at (-0.3909,-2.5805) {$u$};
\node[scale=1.5] at (-2.9295,1.3365) {$v$};
\tikzstyle{every node} = [circle, fill, scale=.4];
\node (v2) at (0,0) {};
\node (v1) at (-3,2) {};
\node (v4) at (3,2) {};
\node (v3) at (0,-3) {};
\draw (v1.center) -- (v2.center) -- (v3.center) -- (v2.center) -- (v4.center);
\draw  plot[smooth, tension=.7] coordinates {(v1) (-2,4) (2,4) (v4)};
\draw  plot[smooth, tension=.7] coordinates {(v1) (-4,0) (-3,-3) (v3)};
\draw  plot[smooth, tension=.7] coordinates {(v3) (3,-3) (4,-1) (v4)};
\end{tikzpicture}
\qquad\qquad
\begin{tikzpicture}[scale=.7]
\node[scale=1.5] at (0,4) {$X_1$};
\node[scale=1.5] at (-2.4261,-0.875) {$X_2$};
\node[scale=1.5] at (2.4698,-0.8497) {$X_3$};
\tikzstyle{every node} = [circle, fill, scale=.4];
\node (v2) at (0,1) {};
\node (v1) at (-3,3) {};
\node (v4) at (3,3) {};
\node (v3) at (-0.4261,-2.875) {};
\node (v5) at (-3.4261,2.125) {};
\node (v6) at (-0.4261,0.125) {};
\node (v7) at (3.4698,2.1503) {};
\node(v8) at (0.4698,-2.8497) {};
\node(v9) at (0.4698,0.1503) {};

\draw  plot[smooth, tension=.7] coordinates {(v1) (-2,5) (2,5) (v4)};
\draw  plot[smooth, tension=.7] coordinates {(v5) (-4.4261,0.125) (-3.4261,-2.875) (v3)};
\draw  plot[smooth, tension=.7] coordinates {(v8) (3.4698,-2.8497) (4.4698,-0.8497) (v7)};
\draw  (v1) edge (v2);
\draw  (v4) edge (v2);
\draw  (v5) edge (v6);
\draw  (v6) edge (v3);
\draw  (v8) edge (v9);
\draw  (v7) edge (v9);

\tikzstyle{every node} = [draw,shape=circle,  scale=.4];
\node (v11) at (-0.005,0.4485) {};
\node (v10) at (-3.1906,2.563) {};
\node (v13) at (0.0447,-2.8476) {};
\node (v12) at (3.2359,2.5741) {};
\draw [densely dotted,line width=1] (v10) edge (v1);
\draw [densely dotted,line width=1] (v5) edge (v10);
\draw [densely dotted,line width=1] (v2) edge (v11);
\draw [densely dotted,line width=1] (v11) edge (v6);
\draw [densely dotted,line width=1] (v11) edge (v9);
\draw [densely dotted,line width=1] (v12) edge (v4);
\draw [densely dotted,line width=1] (v12) edge (v7);
\draw [densely dotted,line width=1] (v3) edge (v13);
\draw [densely dotted,line width=1] (v13) edge (v8);
\end{tikzpicture}
\caption{The graph $\widehat G$, obtained from splitting vertices of $G$, similar to the procedure in \cite{GhaffariP17}.}\label{fig:split}
\end{figure}

\dnsparagraph{Step (3): Diameter Approximation}
The diameter is computed on the BDD tree $\mathcal{T}$ (of the SSSP tree) from the leaf bags up to the root. 
The algorithm is similar to the one in \Cref{sec:unw-diam}, except that again, we use shortcuts in the simulated graph and mark special \emph{portal} nodes along the separators. Here, for each bag $X$ with subgraph $G'=G[X]$, we settle for computing an approximate diameter $\tilde d(G')\leq \max_{u,v \in G'}d_G(u,v)+\epsilon \widetilde{D}$. Set $\delta=\epsilon/3 \cdot \widetilde{D}$.

For every bag $X$, let $X_{a}$ be the set of all \emph{active} nodes defined by the nodes that did not appear on the separator of any ancestor bag of $X$.
Specifically, the invariant for phase $i$ would be that every node in every bag $X$ in level $D(\mathcal{T})-i+1$ knows
$$\widetilde{d}(X) \in  d(X) \pm \delta  \mbox{~~where~~} d(X)=\max_{u \in X_m, v \in X}d_{G}(u,v)~.$$
That is, the value $d(X)$ is restricted to the $G$-distances of $X_m \times X$. This definition is important for handling efficiently the leaf bags $X$, which by definition should satisfy that $|X_m|=\widetilde{O}(1/\epsilon)$.
By keeping this invariant after each step, we get that after $D(\mathcal{T})=O(\log n)$ phases, the root vertex $r$ knows $d(V) \in \widetilde D \pm \epsilon \cdot \widetilde D$.

For the leaf bag $X$, simply assign a leader to collect the distance labels $L_G(v)$ (i.e., in the graph $G$) of all vertices $v \in X_m$, locally compute (the exact) $d(X)$, and then broadcast it to other nodes in the bag. 
Overall, we send $\widetilde{O}(D/\epsilon)$ bits of information and using the low-congestion shortcuts it can be done in $\widetilde{O}(D^2/\epsilon)$ for all the leaf bags. Recall that in the low-congestion shortcuts each edge $e$ appears on $\widetilde{D}$ many subgraphs, and for each subraph, we send $\widetilde{O}(D/\epsilon)$ bits through an edge.

Assume that the invariant holds up to phase $i$ for all bags in level $\ell_i=D(\mathcal{T})-i+1$, and we now describe phase $i+1$.
Let $X$ be a bag in level $\ell_i-1$ and let $X_1,\ldots, X_k$ be its children bags in level $\ell_i+1$. 
Define $G'=G[X]$ and $G^+ \subset G'$ to be all the vertices and edges whose planar embedding is on or inside the closed curve guaranteed by property (9) of BDDs. Similarly, define $G^- \subset G'$ to be all the vertices and edges whose planar embedding is on or outside the closed curve. Define $X^+$ and $X^-$ as the vertex sets of $G^+$ and $G^-$, respectively. Observe that by property (9), every child $X_j$ satisfies $X_j\s X^+$ or $X_j\s X^-$, but not both. 

By the invariant, every node in each child bag $X_j$ knows $\widetilde{d}(X_j)$.
Let $u,v \in X_m \times X$ be the pair of largest $G$-distance in $X$, i.e., $d(X)=d_{G}(u,v)$.
There are two cases: (i) $u \notin S_X$ and $u,v \subseteq X_j$ for some $j \in \{1,\ldots, k\}$, and (ii) $u \in G^+$ and $v \in G^-$. Note that if $u \in S_X$ then it appears on both $G^+$ and $G^-$ and this is taken care of in the second case.

Case (i) can be easily handled since $u$ is also active in $X_j$, and thus the leader knows $\widetilde{d}(X_j)$, and the maximum $\widetilde{d}(X_j)$ value over all $j\in [k]$ can be computed along the shortcuts. Thus, the nodes in $X$ can compute $\max_j \widetilde{d}(X_j)$. 

By property (3') of BDDs, any $u$--$v$ path in case (ii) must pass through a vertex in $S_X$. In particular, the shortest $u$--$v$ path must travel inside $G^-$ until reaching some node $s\in S_X$, then take the shortest $s$--$t$ path in $G$ to some $t\in S_X$ (possibly $t=s$), and finally travels inside $G^+$ to $v$. 
Unlike the unweighted case, $S_X$ might be arbitrarily large, and will not be able to let all vertices in $S_X$ broadcast their labels. To reduce the number of vertices on the separators while introducing some approximation, we apply the technique from  \cite{weimann2016approximating}  of adding \emph{portals} on the separator. By property (7), the separator consists of $O(\log n)$ many paths along the SSSP tree. Therefore, the separator consists of $2\cd O(\logn)=O(\logn)$ many shortest paths. 
To compute the $O(\log n)$ segments, we can simply omit from the shortest path separator of $T'_X$ the at most $O(\log n)$ edges not in $T$, and compute connectivity identification on what remains.

In each shortest path, we mark its first and last vertex, as well as nodes every $\e' \widetilde D$ weighted distance along the path for some $\e':= \Th(\e/\log^2n)$ (see Sec. \ref{app:marking} for implementation details).
Note that since each shortest path has weighted length at most $\widetilde D$, there are at most $O(1/\e')$ many portals per shortest path, or $O(1/\e' \cd \logn) = \tO(1/\e)$ many portals total. Let $Q$ be the set of these portals. 
The portal vertices then satisfy that forcing a $u$-$v$ shortest path between $u \in G'_j$ and $v \in G'_{j'}$ to pass through $Q$ rather than any other vertex in $S_X$ increases the $u$-$v$ distance by an additive term of $\delta$.
That is, we pick the portal $Q$ such that for every $u \in G^-$ and $v \in G^+$, it holds:
\begin{equation}\label{eq:portal}
d_G(u,v)
\leq
\min_{s,t\in Q} d_{G^-}(u,s)+d_G(s,t)+\tilde d_{G^+}(t,v)+\delta~.
\end{equation}
It is then sufficient to compute:
\[
\tilde d(G')
=
\max_{\substack{ u \in G^+, v\in G^-}}\ 
\min_{s,t\in Q} \tilde d_{G^-}(u,s)+d_G(s,t)+\tilde d_{G^+}(t,v)~,
\]
where $\widetilde{d}(.)$ will be an additive $\delta$-approximation for the true distances.
The value $d_G(s,t)$ in the $\min$ expression above can be computed using the distance labels in $G$: every portal node in $Q$ can simply broadcast its distance label $L_G(s)$ to all nodes in $X$. Hence, we send a total of $\widetilde{O}(D)\cdot \poly(\log n/\epsilon)$ bits, and using the shortcuts, we can do it in $\widetilde{O}(D^2)\cdot \poly(\log n/\epsilon)$ rounds for all $i$-level subgraphs.

For the other values $\widetilde{d}_{G^-}(u,s)$ and $\widetilde{d}_{G^+}(t,v)$, we will use the approximate core-sets in $G^-$ and $G^+$. Let's zoom into $G^-$ and explain how to compute the exact distance tuple in this graph with respect to $Q$. Denote by $\widehat{G}^-$ to be the subgraph $G' \cap H'$ where $H'$ are the shortcut edges of $G'$. The edge weights in $\widehat{G}^-$ are set as follows: all edges \emph{not} in $G^-$ are given infinite weights (i.e., weights of $100\widetilde D$) and we keep the edge weights for the edges in $G^-$.
Now, we apply the exact distance label algorithm of Lemma \ref{lem:weighted-exact-label} on $\widehat{G}^-$. This can be done in a total of $\widetilde{O}(D^3)$ rounds for all the $G^-$ subgraphs of that level. Note that the unweighted diameter of $\widehat{G}^-$ is $\tO(D)$, but the distances between vertices in $G^-$ are only based on the edges in $G^-$, as all other edges have large weights. Note that the extra factor of $D$ is due to the $\widetilde{O}(D)$ congestion of the shortcuts.

Then, we let all portal vertices $s \in Q$ send their exact distance tuples $L_{G^-}(s)$ to all the vertices in $G^-$ along the shortcut edges of $G' \cup H'$. The same is repeated to the subgraph $G^+$.
Overall, sending the corresponding distance labels of $Q$ in $G^-$ and $G^+$ takes $\widetilde{O}(D^2/\epsilon)$ rounds for all the subgraphs in this level (since each edge appears on $\widetilde{O}(D)$ shortcuts).

At this point, every vertex $u \in G^-$ knows $\tuple_Q(u)$ in the subgraph $G^-$ (and same for every $v \in G^+$).  
It remains for the leader in $G'$ to compute the additive $\delta$-approximate core-sets in  $G^+$ and in $G^-$. To do this efficiently, every vertex $u \in G^-$ hashes its distance tuple $\tuple_Q(u)$ in a randomly shifted grid as explained in Sec. \ref{sec:fastweighted}, and sends this hash tuple. Since ``nearby" distance tuples are hashed into the same value with good probability, overall the leader should collect $\poly(Q/\epsilon)$ distinct hash values. As explained before this hashing technique only holds in expectation. Thus once the nodes see that the procedure takes too long (e.g.\ exceeds twice its expected runtime) then the nodes start over the process together. W.h.p., this procedure can repeat for $O(\log n)$ times, and w.h.p., one of the tries has its number of entries at most twice the expectation, which is $\poly(\el d/\de')$.
Once all these distinct hash values are gathered, it can collect one representative distance tuple in $G^-$ for each hash value and for each $b \in \{0,1\}$. Same algorithm is applied to compute the core-set in $G^+$. At this point, the leader has computed the approximate core-sets of in $G^+$ and $G^-$, thus it can compute:
\begin{gather}
\max_{\substack{ u\in G^-\\ v\in G^+}} \ \min_{s,t\in Q}(\widetilde{d}_{G^-}(u,s) + d_G(s,t) + \widetilde{d}_{G^+}(t,v)) , \label{eq:diamHT}
\end{gather}
Note that by the definition of the approximate core-set, for every vertex $v \in G^+$, the root knows a distance tuple of some other vertex $v'\in G^+$ such that the distance tuples of $v$ and $v'$ are $\delta$-close and the same of $G^-$. Specifically, let $u \in G^+, v \in G^-$ be the pair that achieve the $\max$ in Eq. (\ref{eq:diamHT}). 
By the definition of $\delta$-approximate core-sets, we know that the leader of $G'$ has collected the tuples of vertices $u' \in G^-(u)$ and $v'\in G^+(v)$ such that $|\tuple_Q(u)-\tuple_Q(u')|\leq \delta$ and $|\tuple_Q(v)-\tuple_Q(v')|\leq \delta$. Therefore, the leader has computed the distance
\begin{eqnarray*}
\widetilde{d}(G')\geq \min_{s,t\in Q} ~~d_{G^-}(u',s) &+&d_G(s,t)+d_{G^+}(v',s) 
\\&\leq& 
2\delta+\min_{s,t\in Q} d_{G^-}(u,s) +d_G(s,t)+d_{G^+}(v,s) 
\\&\leq& 
3\delta+\min_{s,t\in S_X} d_{G^-}(u,s) +d_G(s,t)+d_{G^+}(v,s)=d_G(u,v)+3\delta~,
\end{eqnarray*}
where the last inequality follows by Eq. (\ref{eq:portal}).

\begin{theorem}
For every $n$-vertex weighted planar graph $G$ and $\epsilon \in (0,1)$, there exists a distributed algorithm that computes an $(1+\epsilon)$ approximation for the weighted diameter in $\widetilde{O}(D^3)+\poly(D^2/\epsilon)$ where $D$ is the unweighted diameter of $G$. 
\end{theorem}

\newpage

\bibliographystyle{alpha}
\bibliography{ref}

\newpage
\appendix
\section{Proof of \Cref{lem:drawing}}\label{sec:planar-drawing}

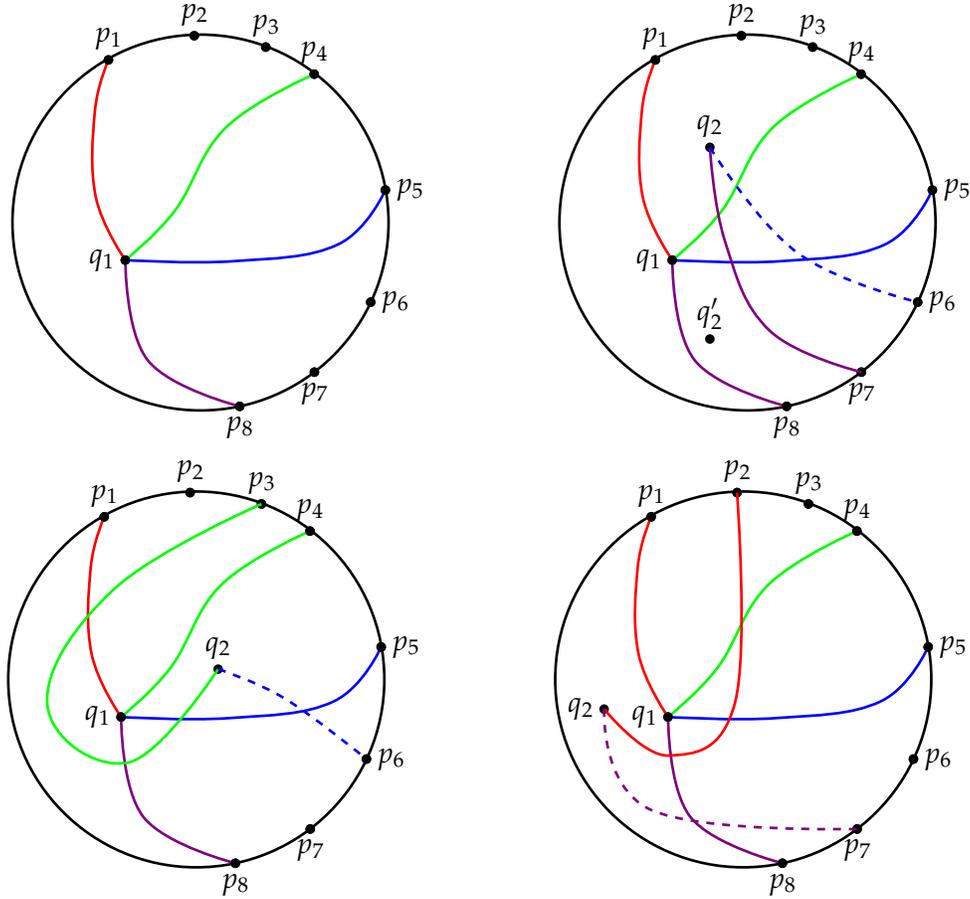
\begin{figure}\centering
\begin{tikzpicture}[scale=.5]
\node [above] at (-2.4481,5.3268) {$p_1$};
\node[above] at (-0.165,5.9713) {$p_2$};
\node[above]at (1.7342,5.6687) {$p_3$};
\node [above] at  (3.0223,4.9478)  {$p_4$};
\node[right] at (4.9155,1.8635) {$p_5$};
\node[right]  at (4.5247,-1.119) {$p_6$};
\node[below] at (3.0287,-2.9803) {$p_7$};
\node[below] at (1.0416,-3.8917) {$p_8$};
\node[left] at (-2,0) {$q_1$};
\tikzstyle{every node} = [circle, fill, scale=.35];
\draw [line width=1] (0,1) ellipse (5 and 5);
\node (v1) at (-2.4481,5.3268) {};
\node at (-0.165,5.9713) {};
\node at (1.7342,5.6687) {};
\node (v2) at (3.0223,4.9478) {};
\node(v3) at (4.9155,1.8635) {};
\node  at (4.5247,-1.119) {};
\node at (-2,0) {};
\node at (3.0287,-2.9803) {};
\node(v4) at (1.0416,-3.8917) {};
\draw[red,line width=1]  plot[smooth, tension=.7] coordinates {(-2,0) (-2.784,1.6784) (-2.808,3.9649) (v1)};
\draw[green,line width=1]  plot[smooth, tension=.7] coordinates {(-2,0) (-0.6783,1.3189) (0.6737,3.5259) (v2)};
\draw[blue,line width=1]  plot[smooth, tension=.7] coordinates {(-2,0) (0.884,-0.0288) (3.6525,0.4146) (v3)};
\draw [violet,line width=1] plot[smooth, tension=.7] coordinates {(-2,0)(-1.407,-2.6595) (v4)};
\node  at (-2.4481,5.3268) {};
\node at (-0.165,5.9713) {};
\node at (1.7342,5.6687) {};
\node  at (3.0223,4.9478) {};
\node at (4.9155,1.8635) {};
\node  (v5) at (4.5247,-1.119) {};
\node at (-2,0) {};
\node at (3.0287,-2.9803) {};
\node  at (1.0416,-3.8917) {};
\end{tikzpicture}
\qquad\qquad
\begin{tikzpicture}[scale=.5]
\node [above] at (-2.4481,5.3268) {$p_1$};
\node[above] at (-0.165,5.9713) {$p_2$};
\node[above]at (1.7342,5.6687) {$p_3$};
\node [above] at  (3.0223,4.9478)  {$p_4$};
\node[right] at (4.9155,1.8635) {$p_5$};
\node[right]  at (4.5247,-1.119) {$p_6$};
\node[below] at (3.0287,-2.9803) {$p_7$};
\node[below] at (1.0416,-3.8917) {$p_8$};
\node[left] at (-2,0) {$q_1$};
\node [above] at (-1,3) {$q_2$};
\node [above] at (-1,-2.1) {$q_2'$};
\tikzstyle{every node} = [circle, fill, scale=.35];
\draw [line width=1] (0,1) ellipse (5 and 5);
\node (v1) at (-2.4481,5.3268) {};
\node at (-0.165,5.9713) {};
\node at (1.7342,5.6687) {};
\node (v2) at (3.0223,4.9478) {};
\node(v3) at (4.9155,1.8635) {};
\node  at (4.5247,-1.119) {};
\node at (-2,0) {};
\node at (3.0287,-2.9803) {};
\node(v4) at (1.0416,-3.8917) {};
\draw[red,line width=1]  plot[smooth, tension=.7] coordinates {(-2,0) (-2.784,1.6784) (-2.808,3.9649) (v1)};
\draw[green,line width=1]  plot[smooth, tension=.7] coordinates {(-2,0) (-0.6783,1.3189) (0.6737,3.5259) (v2)};
\draw[blue,line width=1]  plot[smooth, tension=.7] coordinates {(-2,0) (0.884,-0.0288) (3.6525,0.4146) (v3)};
\draw [violet,line width=1] plot[smooth, tension=.7] coordinates {(-2,0)(-1.407,-2.6595) (v4)};
\node  at (-2.4481,5.3268) {};
\node at (-0.165,5.9713) {};
\node at (1.7342,5.6687) {};
\node  at (3.0223,4.9478) {};
\node at (4.9155,1.8635) {};
\node  (v5) at (4.5247,-1.119) {};
\node at (-2,0) {};
\node(p7) at (3.0287,-2.9803) {};
\node  at (1.0416,-3.8917) {};

\node (q) at (-1,3) {};
\node (q2) at (-1,-2.1) {};
\draw [violet,line width=1] plot[smooth, tension=.7] coordinates {(q) (-0.5974,0.4985)  (0.5812,-1.8702) (p7)};
\draw  [blue,dashed,line width=1]  plot[smooth, tension=.7] coordinates {(q) (1.3149,0.2039) (v5)};
\end{tikzpicture}
\begin{tikzpicture}[scale=.5]
\node [above] at (-2.4481,5.3268) {$p_1$};
\node[above] at (-0.165,5.9713) {$p_2$};
\node[above]at (1.7342,5.6687) {$p_3$};
\node [above] at  (3.0223,4.9478)  {$p_4$};
\node[right] at (4.9155,1.8635) {$p_5$};
\node[right]  at (4.5247,-1.119) {$p_6$};
\node[below] at (3.0287,-2.9803) {$p_7$};
\node[below] at (1.0416,-3.8917) {$p_8$};
\node[left] at (-2,0) {$q_1$};
\node [above] at (0.5812,1.2727) {$q_2$};
\tikzstyle{every node} = [circle, fill, scale=.35];
\draw [line width=1] (0,1) ellipse (5 and 5);
\node (v1) at (-2.4481,5.3268) {};
\node at (-0.165,5.9713) {};
\node(p3) at (1.7342,5.6687) {};
\node (v2) at (3.0223,4.9478) {};
\node(v3) at (4.9155,1.8635) {};
\node  at (4.5247,-1.119) {};
\node at (-2,0) {};
\node at (3.0287,-2.9803) {};
\node(v4) at (1.0416,-3.8917) {};
\draw[red,line width=1]  plot[smooth, tension=.7] coordinates {(-2,0) (-2.784,1.6784) (-2.808,3.9649) (v1)};
\draw[green,line width=1]  plot[smooth, tension=.7] coordinates {(-2,0) (-0.6783,1.3189) (0.6737,3.5259) (v2)};
\draw[blue,line width=1]  plot[smooth, tension=.7] coordinates {(-2,0) (0.884,-0.0288) (3.6525,0.4146) (v3)};
\draw [violet,line width=1] plot[smooth, tension=.7] coordinates {(-2,0)(-1.407,-2.6595) (v4)};
\node  at (-2.4481,5.3268) {};
\node(p2) at (-0.165,5.9713) {};
\node at (1.7342,5.6687) {};
\node  at (3.0223,4.9478) {};
\node at (4.9155,1.8635) {};
\node  (v5) at (4.5247,-1.119) {};
\node at (-2,0) {};
\node(p7) at (3.0287,-2.9803) {};
\node  at (1.0416,-3.8917) {};

\node (q) at (0.5812,1.2727) {};
\draw [green,line width=1] plot[smooth, tension=.7] coordinates {(q) (-1.8106,-1.2116) (-3.9598,0.3252) (-2.2035,3.3872) (p3)};
\draw  [blue,dashed,line width=1]  plot[smooth, tension=.7] coordinates {(q) (2.3607,0.487) (v5)};
\end{tikzpicture}
\qquad\qquad
\begin{tikzpicture}[scale=.5]
\node [above] at (-2.4481,5.3268) {$p_1$};
\node[above] at (-0.165,5.9713) {$p_2$};
\node[above]at (1.7342,5.6687) {$p_3$};
\node [above] at  (3.0223,4.9478)  {$p_4$};
\node[right] at (4.9155,1.8635) {$p_5$};
\node[right]  at (4.5247,-1.119) {$p_6$};
\node[below] at (3.0287,-2.9803) {$p_7$};
\node[below] at (1.0416,-3.8917) {$p_8$};
\node[left] at (-2,0) {$q_1$};
\node [left] at (-3.7,0.2097) {$q_2$};
\tikzstyle{every node} = [circle, fill, scale=.35];
\draw [line width=1] (0,1) ellipse (5 and 5);
\node (v1) at (-2.4481,5.3268) {};
\node at (-0.165,5.9713) {};
\node at (1.7342,5.6687) {};
\node (v2) at (3.0223,4.9478) {};
\node(v3) at (4.9155,1.8635) {};
\node  at (4.5247,-1.119) {};
\node at (-2,0) {};
\node at (3.0287,-2.9803) {};
\node(v4) at (1.0416,-3.8917) {};
\draw[red,line width=1]  plot[smooth, tension=.7] coordinates {(-2,0) (-2.784,1.6784) (-2.808,3.9649) (v1)};
\draw[green,line width=1]  plot[smooth, tension=.7] coordinates {(-2,0) (-0.6783,1.3189) (0.6737,3.5259) (v2)};
\draw[blue,line width=1]  plot[smooth, tension=.7] coordinates {(-2,0) (0.884,-0.0288) (3.6525,0.4146) (v3)};
\draw [violet,line width=1] plot[smooth, tension=.7] coordinates {(-2,0)(-1.407,-2.6595) (v4)};
\node  at (-2.4481,5.3268) {};
\node(p2) at (-0.165,5.9713) {};
\node at (1.7342,5.6687) {};
\node  at (3.0223,4.9478) {};
\node at (4.9155,1.8635) {};
\node  (v5) at (4.5247,-1.119) {};
\node at (-2,0) {};
\node(p7) at (3.0287,-2.9803) {};
\node  at (1.0416,-3.8917) {};

\node (q) at (-3.7,0.2097) {};
\draw [red,line width=1] plot[smooth, tension=.7] coordinates {(q) (-1.984,-1.0267)  (-0.2161,0.3945) (p2)};
\draw  [violet,dashed,line width=1]  plot[smooth, tension=.7] coordinates {(q) (-2.4808,-2.4595) (p7)};
\end{tikzpicture}
\caption{
The setting and cases for the proof of \Cref{lem:drawing}. In each of the three cases, the dotted curve is unavoidable and establishes the contradiction.
}
\label{fig:drawing}
\end{figure}

Here, we prove \Cref{lem:drawing}, restated below.

\drawing*

\BP
For simplicity, we assume that $p_2\ne p_3$, $p_4\ne p_5$, and $p_6\ne p_7$; the case when some of them are equal is analogous.

The curves from $q_1$ to $p_1$, $p_4$,  $p_5$, and $p_8$ divide the region inside the curve into four regions; see \Cref{fig:drawing}, top left. There are four cases depending on which region $q_2$ belongs to; \Cref{fig:drawing} displays three cases, with the last being symmetric.

In the first case (top right), the arc $(q_2,p_7)$ must cross the arc $(q_1,p_5)$ because it cannot cross $(q_1,p_8)$. Since arcs out of $q_2$ cannot pairwise cross, this forces the arc $(q_2,p_6)$ to cross $(q_1,p_5)$, contradiction. Note that this case also handles the symmetric case when $q_2$ is inside the region containing the point $q_2'$ in the figure.

In the second case (bottom left), the arc $(q_2,p_3)$ must cross the arc $(q_1,p_5)$ because it cannot cross $(q_1,p_4)$. This forces the arc $(q_2,p_6)$ to cross $(q_1,p_5)$, contradiction.

In the third case (bottom right), the arc $(q_2,p_2)$ must cross the arc $(q_1,p_4)$ because it cannot cross $(q_1,p_1)$. This forces the arc $(q_2,p_7)$ to cross $(q_1,p_8)$, contradiction.
\EP

\section{Auxiliary Distributed Procedures}\label{sec:aux}

%
\subsection{Computation of a Balanced Cycle Separator}\label{sec:balanced-cycle-general}
For a graph $G=(V,E)$, a subset of vertices $S \subseteq V$ is a \emph{balanced separator} if the removal of $S$ breaks $G$ into connected components that are constant factor smaller than the number of vertices in $G$. 
For a graph $G$ and a spanning tree $T \subseteq G$, a \emph{balanced cycle separator} is a balanced separator $S$
which forms a cycle as follow: the vertices of $S$ are connected by two tree-paths $\pi(x,y)$ and $\pi(y,z)$ (where possibly $z=y$) plus an additional edge $(x,z)$ which is not necessarily in $G$. A cycle separator naturally defines two regions in $G$, the region inside the cycle and the region outside, the number of vertices in both these regions should be at most $2n/3$. 

In this section we prove Thm. \ref{thm:cycle} by extending the algorithm of \cite{GhaffariP17} to $1$-connected graphs. To do that, we will augment $G$ with a subset of virtual edges $E'$ such that $G \cup E'$ is \emph{biconnected}. We will then see that (1) computing those virtual edges can be done in $\widetilde{O}(D)$ rounds\footnote{Here $D$ is the diameter of the input graph $G$. This would hold even if it is required to compute the augmentation for a subgraph $G' \subseteq G$ with possibly larger diameter than that of $G$.}; and (2) that any $r$-round algorithm on $G \cup A$ be simulated in $G$ within $\widetilde{O}(r)$ rounds. Thus computing a cycle separator in $G$ can be done by simulating the algorithm of \cite{GhaffariP17} on $G \cup E'$ in the graph $G$, using $\widetilde{O}(D)$ rounds. The output will be a cycle separator that contains at most one edge that is not in $T \subseteq G$. Note that $V(G)=V(G \cup E')$ and thus a balanced cycle separator in $G \cup E'$ is also a balanced cycle separator in $G$. To define the virtual edges that transform $G$ into a biconnected graph we use the block-cut tree representation of the graph.

\paragraph{The Block-Cut Tree.} The block-cut tree is a representation of the biconnected components in the graph, called \emph{blocks}. The tree has two types of vertices: cut vertices and block vertices, where a cut-vertex $v$ is connected to all the block vertices that contain $v$. For our purposes, we root the tree at an arbitrary cut-vertex. 
A representation of the block-cut tree $\mathcal{T}$ can be computed in $O(D)$ rounds, by using the biconnectivity identification algorithm of \cite{GhaffariP17}, we explain more about the implementation details at the end of the section. For every block $B$, let $r(B)$ be the cut-vertex in $B$ that is closest to the root, thus $r(B)$ is also the parent of the block-vertex $B$ in the block-cut tree. 
For every cut vertex $v$, let $\ell(v)$ be the \emph{level} of the cut-vertex in the block-cut tree, for every non-cut vertex $v'$ the level of $v'$ is the level of the (unique) block containing $v'$ in the block-cut tree.
Finally, the level  $\ell(e)$ of every edge $e$ is the level of the block containing $e$. 



\paragraph{Planar Biconnectivity Augmentation.}
We add to $G$ two subsets of virtual edges, namely, $A$ and $B$ such that $G \cup A \cup B$ is a biconnected planar graph. The first subset $A$ is defined as follows.
Every cut-vertex $u$ in level $\ell(u)$ connects by a virtual edge every two consecutive neighbors $u_1,u_2$ in the clockwise ordering satisfying that $\ell((u,u_1)),\ell((u,u_2))=\ell(u)+1$. That is, the cut-vertex $u$ connects by a virtual edge consecutive neighbors that belong to child blocks of $u$ in the block-cut tree\footnote{In fact it is sufficient to connect incident neighbors from distinct child blocks.}. The level of every newly added virtual edge $e'=(u_1,u_2)$ is set to $\ell(e')=\ell(u)+1$. This is done simultaneously for all cut-vertices. 

From that point on, we treat those edges $A$ edges as part of our graph, and the second step is applied on the graph $G \cup A$. We also applied the planar embedding algorithm of Ghaffari andHaeupler \cite{ghaffari2016embedding} on the graph $G \cup A$ so that every vertex $v$ knows the clockwise ordering of its neighbors in $G \cup A$. We do \emph{not} re-compute the block-cut tree of $G \cup A$, and we only use the level information of the edges in $A$. 
In the second step, a subset of virtual edges $B$ is computed as follows. Every cut-vertex $u$ with an even value of $\ell(u)/2$ considers its neighbors  $u_1,\ldots, u_k$ in $G \cup A$ in a \emph{clockwise ordering}. We say that the neighbor $u_j$ is of level $i'$ if $\ell((u,u_j))=i'$. It then connects a neighbor $u_j$ in level $\ell(u)+1$ to a consecutive neighbor in the ordering that belongs to level $\ell(u)-1$. For every cut vertex $u'$ with an odd value of $\ell(u')/2$, we do the same only in the reverse direction, i.e., connecting neighbor in level $\ell(u)+1$ to a consecutive neighbor in level $\ell(u)-1$ in the counter clockwise ordering. See Fig. \ref{fig:virtual-bct} for an illustration.

\begin{figure}[h!]
\includegraphics[scale=0.5]{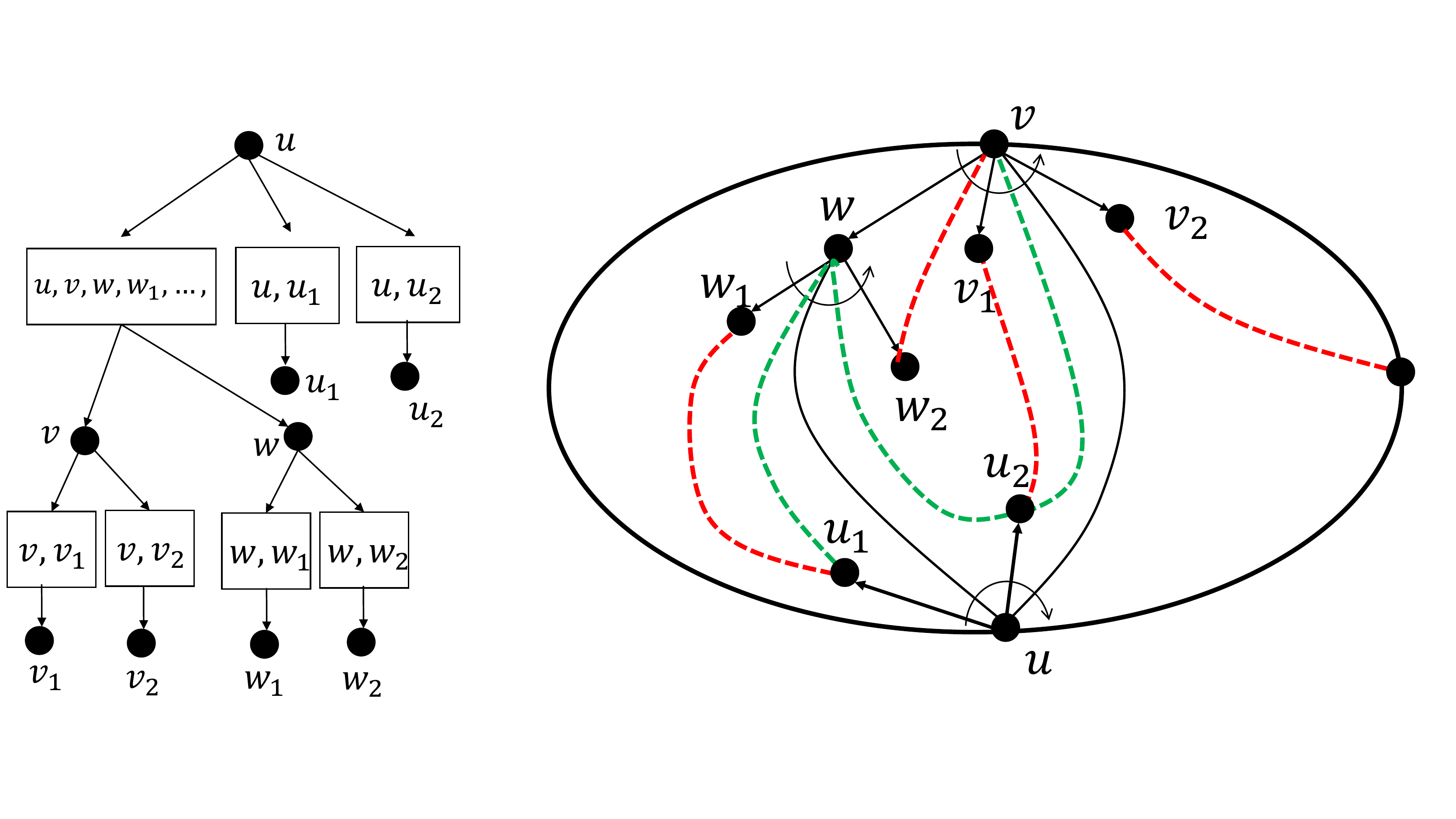}
\caption{Illustration for the virtual edges added. Right: Black edges are real $G$-edges, virtual type $A$ edges are in green, and type $B$ edges are in red. Left: The block-cut tree representation that is used to define those edges.
\label{fig:virtual-bct}}
\end{figure}

We first claim that $G \cup A \cup B$ is biconnected. For every level-$i$ cut-vertex $u$ it will be sufficient to show that all its neighbors in level $i+1$ are connected --not through $u$ -- to at least one neighbor in level $i-1$.
In the first step, we connect consecutive neighbors on level-$(i+1)$ by $A$-edges. In the second step, w.l.o.g. assume that $i/2$ is even. There are only two types of $u$'s neighbors, namely, those in level $i+1$ and those of level $i-1$. Since all level $(i+1)$ neighbors between two consecutive level $(i-1)$ neighbors are connected by $A$ edges, by adding the $B$ edges we connects the consecutive neighbors of level $(i+1)$ and $i$, as desired.  

We next show that those edges do not violate planarity, and start with the following observation that basically says that if those edges are added for a single vertex then planarity is preserved. 
\begin{observation}\label{obs:safe-parallel}
Fix a vertex $v$ and let $u_1,u_2,\ldots, u_\ell$ be its neighbors ordered in a clockwise ordering (based on the embedding). Then one can add the edges $(u_{i},u_{i+1})$ without violating planarity, i.e., $G \cup \{(u_{i},u_{i+1}), i \in \{1,\ldots, \ell-1\}\}$ is planar.
\end{observation}
To show that connecting many consecutive neighbors in parallel does not violate planarity, we start with the following observation:
\begin{observation}
Edges added by non-neighboring cut-vertices do not intersect with each other. 
\end{observation}

We first show that adding all the $A$ edges is safe. For a cut-vertex $u$ in level $\ell(u)$, we orient its incident edges in level $(\ell(u)+1)$ away from $u$, see Fig. \ref{fig:virtual-bct}.
Note that an edge cannot be directed in both directions in this definition, as the tail of the edge is the parent of the block containing $(u,v)$ in the block-cut tree. It might also be the case that an edge is not directed at all with this definition, e.g., in case where both the endpoints of an edge $e=(u,v)$ are in level $\ell(e)+1$ (i.e., the edge $(u,v)$ is in the parent block of $u$ and $v$).

Thus, in the first step, we connect two outgoing neighbors $u_1,u_2$ that are consecutive on the clockwise ordering. The only problematic case is where one of these neighbors say $u_1$ is a cut-vertex that connects $u$ to some of $u_1$'s neighbor. Here, this case cannot happen, as a cut-vertex connects only outgoing neighbors with $A$ edges, and $u$ is an incoming neighbor of $u_1$. 

We next show that also the edges of $B$ are safe to be added to $G \cup A$. Also here it remains to consider the case of two neighboring cut-vertices $u,v$. Case (I): $u$ and $v$ are in the same level. In this case, it must be that $\ell(u)=\ell(v)=\ell((u,v))+1$. That is, the edge $(u,v)$ belongs to the parent block of $u$ and $v$ in the block-cut tree. In order words, the edge $(u,v)$ serves as the parent edge for both $u$ and $v$. Since we apply the same rule (connecting neighbors in the clockwise or counterclockwise ordering) for these vertices, their edges do not intersect. See Fig. \ref{fig:type-A-same-level} for an illustration.
Case (II): $u$ and $v$ are in distinct levels. In this case, we can assume without loss of generality, that $v$ is a descendent of $u$ in the block-cut tree. Letting $i$ be the level of the edge $(u,v)$, we get that $\ell(u)=i-1$ and $\ell(v)=i+1$. Note that in such a case, a cut-vertex $v$ connects consecutive neighbors $u$ and $w$ such that $(u,w)$ is in level $i+2$. Therefore, the edges are oriented from $u \to v \to w$. We want to show that the $B$ edges added by $u$ do not intersect the edges with the $B$ edges added by $v$. If $v$ connects $u$ and $w$ by a $B$-edge, since $u$ and $v$ have a difference of $+2$ in their levels, $u$ adds edges in the counter direction to $v$. 

\begin{figure}[h!]
\includegraphics[scale=0.5]{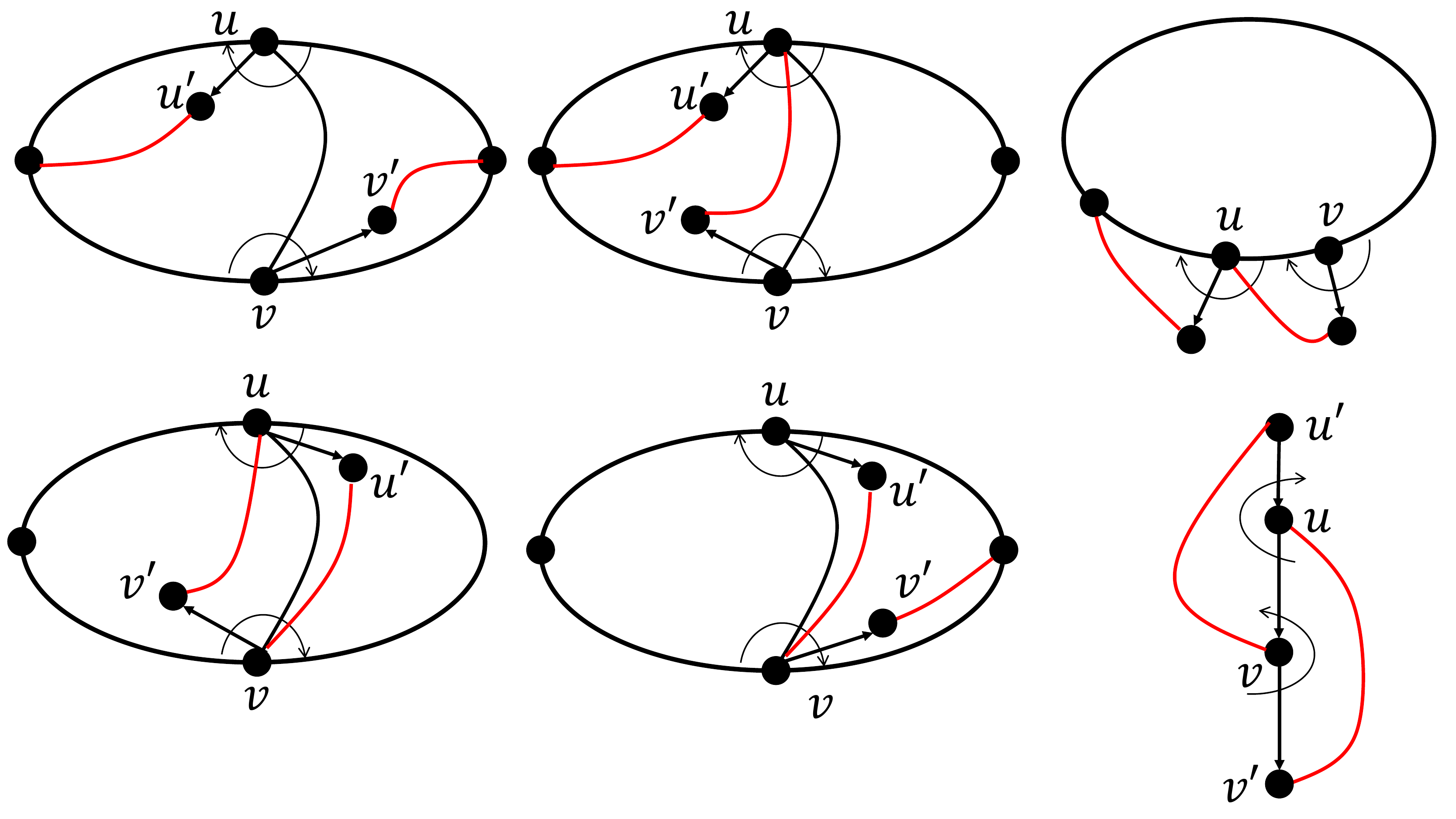}
\caption{The first four drawing consider the case where the edge $(u,v)$ is inside the block and $u$ and $v$ are cut-vertices in the same level in the block-cut tree. In the top-right we consider the case where $(u,v)$ is on the boundary of the block. Bottom-right is an illustration for the case where $u$ and $v$ have different levels, in such a case, w.l.o.g. $u$ is a parent of the block containing the edge $(u,v)$. 
\label{fig:type-A-same-level}}
\end{figure}


\paragraph{Computation of the Planar Biconnectivity Augmentation.}
To compute the virtual edges of type $A$ it is sufficient to know the planar embedding, the biconnected component identification of the block-cut tree and their levels in the block-cut tree. 
This can be easily computed within $O(\diam(G'))$ where $G'$ is the subgraph on which the biconnectivity augmentation is computed. In the unweighted case, $G'$ is guaranteed to have diameter $O(D\log n)$. 

In the weighted case, the diameter of $G'$ might be large, and in such a case we follow the procedures of \cite{ghaffari2016distributed} and use the low-congestion shortcuts in $G$ to compute the layering in $G'$. Let $T'$ be any spanning tree of $G'$, the block-cut tree of $G'$ clearly ``agrees" with $T'$. 
We will mark the cut-vertices, and arbitrarily pick one such cut-vertex as a root denoted by $r$. 
The goal is for each cut-vertex $u$ to count the number of cut-vertices on its tree path $\pi(r, u,T')$. This allows $u$ to compute its level in the block-cut tree. Once every cut-vertex has computed its level, all other edges and vertices can compute their level as well: for every block $B$ the level of $B$ is the level of its root cut-vertex $r(B)$ plus one, the level of all edges and non cut vertices is the level of the unique block that contains them. In every block $B$, all nodes can learn $\ell(r(B))$ using low-congestion shortcuts with $\widetilde{O}(D)$ rounds. 

We now explain how each cut-vertex $v$ can count the number of cut-vertices on its tree path $\pi(r, u,T')$. 
First, the tree $T'$ is oriented towards the root $r$ using the tree orientation procedure of \cite{ghaffari2016distributed}, which taken $\widetilde{O}(D)$ rounds. 
Next, we apply a recursive fragment procedure on this rooted tree. In this procedure, every fragment has a root which is the root of the subtree that spans the vertices in the fragment. In every merging step, each child fragment suggests the tree edge to its parent fragment (i.e., the fragment containing the parent of the root of the fragment in the tree) for the merge and then, the head-fragments
accept all their children tail fragments. 
We will keep the invariant that up to phase $i$, every vertex knows the number of cut vertices on its tree path from the root of its level $(i-1)$ fragment. In phase $i$,  an $i$-level fragment is created by merging a parent fragment $F'$ with a subset of child $(i-1)$-level fragments. Thus, it is easy to update the information in each child fragment $F$ by adding the value (i.e, number of cut-vertices from the root of $F'$) for a vertex $u \in F'$, where $u$ is the vertex that is connected via a tree edge to a vertex in $F$.  After $O(\log n)$ phases, the fragment contains the entire tree and thus each vertex knows the number of cut-vertices on its tree path from $r$. This allows each cut-vertex to compute its layer number in the block-cut tree. Overall, this computation takes $\widetilde{O}(D)$ rounds.

\paragraph{Completing the Proof of Thm. \ref{thm:cycle}}
Computing the set of virtual edges is done in $\widetilde{O}(D)$ rounds using the procedure above.
Note that the cycle separator $S$ computed in the augmented graph $G \cup A \cup B$ consists of two tree paths in $T$ plus one additional edge which is possibly not in $G$. Since $T \subseteq G$ is spanning a tree in $G \cup A \cup B$, the cycle separator has at most one edge that is not in $G$ (i.e., it might not even be in $A$).

It remains to show that the separator algorithm on $G \cup A \cup B$ of \cite{GhaffariP17} can be simulated in $G$ within $\widetilde{O}(D)$ rounds. This follows by showing that the endpoints of all edges in $A$ are connected in $G$ by nearly disjoint paths of length at most $5$.
\begin{claim}
All virtual edges can be simulated in $\widetilde{O}(1)$ rounds in $G$.
\end{claim}
\begin{proof}
We first show the claim for the edges in $A$. Clearly, each virtual edge $(u,v)$ has a common cut-vertex neighbor $w$. Thus it can be simulated via the path $u-w-v$. Recall that we view those edges as directed away from $u$, therefore the edge $(u,w)$ is only used to simulate at most one $A$ edges.
We now consider the second subset of virtual edges $B$. Here we connect the neighbor of a cut vertex $u$ in its child component to a neighbor of $u$ in its parent component. In the worst case both edges are virtual edges in $A$ (but in fact, at most one such edge can be in $A$). Thus, the endpoints of each edge in $B$ are connected by a path of length $4$ in $G$. Next, note that each edge in $G \cup A$ is used to simulate at most two edges in $B$, since an edge $e \in G \cup A$ can get connected at most twice, once at each of its side. Thus overall each edge in $G$ appears on constant many paths connecting the endpoints of the $B$ edges. 
\end{proof}

\subsection{Modifications for Computing the BDD on Weighted SSSP}\label{app:sssp-bdd}

\paragraph{Computing Separators.} Our goal is to compute shortest path separators in the SSSP until every bag contains $O(\log n\cdot 1/\epsilon)$ non-separator vertices. To do that we will apply the following modifications. Initially we unmark all vertices and throughout the recursion, we mark vertices that belong to the separator. For bag $X$, all vertices that appear as part of the separator in the ancestor bags 
$X'$ of $X$ are marked. Then, we compute a \emph{weighted} separator in $G[X]$ by assigning all marked nodes weight $0$ and all remaining nodes have weight of $1$. The algorithm of Ghaffari and Parter \cite{GhaffariP17} can be easily modified to work in this weighted version. The output shortest path separator satisfies that the total weight strictly inside and outside the cycle is at most a constant fraction of the total weight.

Since the diameter of $G[X]$ might be large, we will be working on the transformed graph $\widehat{G}$ and compute low-congestion shortcuts $H_1,\ldots, H_k$ for every subgraph $G[X_1], \ldots, G[X_k]$ in that level. Since every edge $e$ appears on $\widetilde{O}(D)$ subgraphs $H_i \cup G[X_i]$, and since the the separator algorithm of Ghaffari and Parter \cite{GhaffariP17} takes $\widetilde{O}(D)$ rounds, overall on each edge that algorithm sends $\widetilde{O}(D^2)$ messages. Using the random delay approach, we can compute the separator in each subgraph $G[X_i]$ simultaneously using a total of $\widetilde{O}(D^2)$ rounds. 

We note that if instead of using the separator algorithm of Ghaffari and Parter \cite{GhaffariP17} as a black box in each subgraph $H_i \cup G[X_i]$, the algorithm computes shortcuts for all the faces insides all subgraphs of the same level (as repeatedly done in \cite{GhaffariP17}), the round complexity can be improved to $\widetilde{O}(D)$ rounds.

\paragraph{Defining the Child Components.}
In the same manner as for the unweighted case, $T' \setminus T$ has $O(\log n)$ edges. However, the diameter of $T'$ might be large. The child bag $X^+$ is defined in the same manner as in the unweighted case. 
We next consider the remaining components. First, suppose that $S$ does not contain any vertices in $\ol{\m O(X)}\setminus\m O(X)$, the boundary of $\m O(X)$. In this case,  $\m O(X) \setminus \ol{\m O(X^+)}$ is actually connected, an the other child bag is computed in the same manner as $X^+$.
Otherwise, $S$ has vertices lying on $\ol{\m O(X)}\setminus\m O(X)$. 
In such a case, we define the remaining components almost in the same manner as in the unweighted case (see Step 4 in Sec. \ref{sec:4.2}), with the only difference is that we apply a connectivity identification algorithm using low-congestion shortcuts, in the auxiliary graph. 
Overall, the computation of the BDD on the SSSP can be done in $\widetilde{O}(D^2)$ rounds.
(Also here the computation can be made  $\widetilde{O}(D)$ rounds, if we compute the shortcuts for all faces of the subgraphs in the same level). 

%
%

\subsection{Marking $O(1/\epsilon)$ Portals on a Shortest Path}\label{app:marking}
Let $u$-$v$ be a shortest path segment on the separator path. Using the recursive merging procedure of \cite{ghaffari2016distributed,GhaffariP17}, we can mark all vertices on the path and orient it towards $u$. By letting $u$ send its exact distance label, all the vertices on the path can compute their distance from $u$ in $G$. Then, vertices exchange this distance with their path neighbors. Each vertex $x$ on the path then computes $\lceil d_G(u,x) / (\epsilon \cdot \widetilde{D}) \rceil$. This is the distance class of $x$. The portal vertices are the switching points of the distance classes along the path. That is, a vertex $x$ is a portal if its distance class is distinct from its parent on the path. Since the distance label of $u$ has $\widetilde{O}(D)$ bits, and since there are $O(\log n)$ shortest-path segments, this step is implemented in $\widetilde{O}(D^2)$ rounds (as we use low-congestion shortcuts to send the information along the path). 

%
%

\end{document}